\documentclass[11pt]{article}
\usepackage{amsmath}
\numberwithin{equation}{section}
\usepackage{amsfonts,amsthm,amssymb,bm,bbm}
\usepackage{graphicx}
\usepackage{hyperref,setspace,natbib}
\usepackage{float}
\usepackage{booktabs,multicol,multirow}
\usepackage{cleveref}
\usepackage{threeparttable}
\usepackage[svgnames]{xcolor}
\usepackage{mleftright} 
\usepackage{algorithm}
\usepackage{algorithmic}

\usepackage{enumitem}

\usepackage{subcaption} 
\usepackage{caption}
\usepackage{makecell}

\usepackage{threeparttable}
\usepackage{tabularx}
\newcolumntype{Y}{>{\centering\arraybackslash}X}

\newcommand{\prob}{\mathbb{P}}				
\newcommand{\E}{\mathbb{E}}					
\newcommand{\F}{\mathcal{F}}					
\newcommand{\G}{\mathcal{G}}					

\newcommand{\R}{\mathbb{R}}					

\newcommand{\st}{\mathrm{s.t.}\ }			

\renewcommand{\geq}{\geqslant}
\renewcommand{\leq}{\leqslant}

\newcommand{\Q}{\mathcal{Q}}					

\newcommand{\dif}{\mathop{}\!\mathrm{d}} 
\newcommand{\set}[2]{\mleft\{ #1 \ \middle|\  #2 \mright\}}

\newcommand{\possessivecite}[1]{\citeauthor{#1}'s (\citeyear{#1})}

\newtheorem{thm}{Theorem}

\newtheorem{lemma}{Lemma}
\newtheorem{prop}{Proposition}
\newtheorem{defn}{Definition}
\newtheorem{problem}{Problem}
\newtheorem{Eg}{Example}

\newtheorem{assumption}{Assumption}
\newtheorem{rmk}{Remark}

\allowdisplaybreaks
\newcommand{\major}{\mathrm{Mj}}
\newcommand{\minor}{\mathrm{Mn}}
\newcommand{\period}{\mathrm{per}}
\newcommand{\trend}{\mathrm{tre}}
\newcommand{\nogame}{\mathrm{nogame}}
\newcommand{\agg}{\mathrm{agg}}

\begin{document}

\title{~~~ \bf \LARGE \\[-.75in] 
\makebox[\textwidth]{
\parbox[c]{1.2\textwidth}{\centering Periodic Trading Activities in Financial Markets: \\ Mean-field Liquidation Game with Major-Minor Players%
\thanks{We thank Guanxing Fu for helpful and constructive comments and discussion.}}}
}
\author{Yufan Chen\thanks{School of Mathematical Sciences, Peking University. \textbf{Email:} \url{yufan_chen@pku.edu.cn}.}, \ 
Lan Wu\thanks{School of Mathematical Sciences, Center for Statistical Science, and LMEQF, Peking University. \textbf{Email:} \url{lwu@pku.edu.cn}.}, \ 
Renyuan Xu\thanks{Department of Finance and Risk Engineering, New York University. \textbf{Email:} \url{rx2364@nyu.edu}.}, \ 
and Ruixun Zhang\thanks{Corresponding Author. School of Mathematical Sciences, Center for Statistical Science, the National Engineering Laboratory for Big Data Analysis and Applications, and LMEQF, Peking University. \textbf{Email:} \url{zhangruixun@pku.edu.cn}. Ruixun Zhang's research is partially supported by the National Key R\&D Program of China Project (2022YFA1007900) and the National Natural Science Foundation of China Project (12271013).}}
\date{
\today
}
\maketitle
\thispagestyle{empty}
\allowdisplaybreaks
\centerline{\bf Abstract} \baselineskip 14pt \vskip 10pt
{\noindent
Motivated by recent empirical findings on the periodic phenomenon of aggregated market volumes in equity markets, we aim to understand the causes and consequences of periodic trading activities through a game-theoretic perspective, examining market interactions among different types of participants. Specifically, we introduce a new mean-field liquidation game involving major and minor traders, where the major trader evaluates her strategy against a periodic targeting strategy while a continuum of minor players trade against her. We establish the existence and uniqueness of an open-loop Nash equilibrium. In addition, we prove an $\mathcal{O}(1/\sqrt{N})$ approximation rate of the mean-field solution to the Nash equilibrium in a major-minor game with $N$ minor players. In equilibrium, minor traders exhibit front-running behaviors in both the periodic and trend components of their strategies, reducing the major trader's profit. Such strategic interactions diminish the strength of periodicity in both overall trading volume and asset prices. Our model rationalizes observed periodic trading activities in the market and offers new insights into market dynamics.}


\vskip 20pt\noindent {\bf Keywords}: Periodicity in financial market; Algorithmic trading; Optimal liquidation; Mean-field game.
%



\setcounter{page}{1}
\pagenumbering{arabic}
\setlength{\baselineskip}{1.5\baselineskip}
\onehalfspacing
\setcounter{equation}{0}
\setcounter{table}{0}
\setcounter{figure}{0}

\section{Introduction}
\label{sec:introduction}
\paragraph{Motivation.}
It has long been documented that trading activities and trading volumes contain persistent high-frequency periodicities beyond the well-known intraday U shapes.
For instance, \cite{heston2010intraday} observed that stock returns and trading volumes in half-hour intervals show periodicities at daily frequencies; \cite{hasbrouck2013low} found that trading messages arrive at the market with large peaks in several milliseconds shortly after the one-second boundary; \cite{wu2022spectral} documented strong and universal second- and minute-level periodicities in a large panel of US and Chinese stocks; \cite{hansen_periodicity_2024} found intraday periodic patterns in volatility and volume in the cryptocurrency market.

These periodic patterns yield important implications for the underlying asset and market participants.
They contain information on the timing and magnitude of price changes \citep{karpoff_relation_1987,wang_model_1994}, lead to abnormal excess returns for stocks dominated by high-frequency fluctuations in volume \citep{chinco2017investment}, are associated with price efficiency and price informativeness in the market \citep{wu2022spectral}, and provide better predictions for intraday volume and execution costs \citep{brownlees2011intra,wu2022spectral}.

Despite the persistence and importance of these periodicities, it remains unclear what drives them, and how the interactions between market participants affect these patterns. One potential driving factor is the {\it periodic targeting strategy} of institutional traders. When liquidating positions, institutional traders have incentives not only to reduce price impacts but also to minimize deviations from pre-specified targeting strategies \citep{cartea_optimal_2015,cheng_optimal_2024} that track, for example, implementation shortfall orders \citep{almgren_optimal_2001}, volume-weighted average price (VWAP), or time-weighted average price (TWAP) \citep{gueant_accelerated_2015}. 
Such targeting strategies are naturally implemented with discrete-time periodicities in practice.

\paragraph{Our work and contributions.} To understand the causes and consequences of periodic trading activities through a game-theoretic perspective, we introduce a new mean-field liquidation game involving a major and a continuum of minor traders. The major trader evaluates and improves her strategy against a periodic and deterministic targeting strategy while a continuum of minor players trades against her and tries to make profits. To the best of our knowledge, this is the first time a mean-field game (MFG) has been used to explain the periodic phenomena in financial markets that have been extensively documented in the empirical studies discussed above.

To explain and understand periodicity across the trader population, this liquidation game introduces three key modeling innovations that distinguish it from the existing literature, presenting unique mathematical challenges.
Firstly, the game stands out as an extended MFG, where traders' interactions are shaped by the permanent impacts of their strategies on the fundamental price. Secondly, all traders, both major and minor, are influenced by a common noise that drives the benchmark price process, reflecting the impact of market uncertainties faced by all participants. Lastly, traders face terminal state constraints due to liquidation requirements, which are consistent with financial practices.

Theoretically,  we adopt a probabilistic approach to establish the existence and uniqueness of an open-loop Nash equilibrium. Similar approaches have been used to study the Nash equilibrium of extended MFGs {\it without} a major player \citep{fu_mean_2021} and MFGs with a major player, where players interact with each other through the states {\it but not the controls} \citep[Chapter 7.1]{carmona_probabilistic_2018}.  In addition, we prove an $\mathcal{O}(1/\sqrt{N})$ approximation rate of the mean-field solution to the Nash equilibrium in a major-minor game with $N$ minor players.

We show both theoretically and numerically that, in equilibrium, minor traders exhibit front-running behaviors in both the periodic and trend components of their strategies, which reduces the major trader's profit, an implication consistent with the findings of \possessivecite{brunnermeier_predatory_2005} model of predator trading. 
The major trader reduces the strength of her periodic behavior in response to the strategic actions of the minor trader.
Overall, such strategic interactions diminish the strength of periodicity in both the overall trading volume and asset prices, thereby serving as a stabilizing mechanism for the aggregate market. To conclude, our model rationalizes observed periodic trading activities in the market and offers new insights into market dynamics.

\paragraph{Related literature.} Our work is closely related to two lines of literature.\vspace{.5em}

\noindent \underline{Mean-field games (MFGs).} 
MFG was first introduced in the seminal works of \citep{caines_large_2006,lasry_mean_2007}, which studied a stochastic game with a continuum of homogeneous agents interacting through the distribution of the population state. Since these foundational studies, extensive research has explored various extensions of MFG. While a comprehensive survey is well beyond the scope of this paper, we summarize several key directions that closely align with our formulation. 

Extended MFGs focus on game interactions through action distributions instead of state distributions. In this framework, \cite{gomes_extended_2016} proved an existence result and established uniqueness conditions under deterministic dynamics; \cite{graber_linear_2016} studied the linear-quadratic formulation. See also \cite{gomes_existence_2014} for the results on infinite-horizon stationary setting and \cite{bonnans_schauder_2021} for potential games.

MFGs with common noise consider scenarios where the dynamics of individual agents is influenced by a common process such as Brownian motion. \cite{ahuja_wellposedness_2016} and \cite{carmona_mean_2016} provided the existence and uniqueness of solutions for MFGs with common noise. \cite{bensoussan_well-posedness_2015} and \cite{cardaliaguet_master_2019} studied the well-posedness of the corresponding master equations.  In addition, \cite{djete_mean_2023} provided the existence and approximated Nash equilibrium for extended MFGs with common noise.

To address heterogeneity among agents, the framework of Major-Minor MFG, which involves a major player and a continuum of minor players, has gained popularity in both theoretical exploration and practical applications. Under the Nash equilibrium criterion, \cite{huang_large-population_2010} first introduced the concept of MFG with a major player in a linear-quadratic setting (with linear dynamics). \cite{nourian_epsilon-nash_2013} and \cite{carmona_probabilistic_2016} further extended this framework to nonlinear stochastic state dynamics using a probabilistic approach. Recently, \cite{lasry_mean-field_2018} and \cite{cardaliaguet_remarks_2020} studied a general setting using the master equation approach. In addition to the Nash equilibrium criterion, the Stackelberg game criterion has also been investigated; see \cite{bensoussan_mean_2016,huang2021linear,guo2022optimization,aurell2022optimal,dayanikli2023machine}. 

We remark that none of the existing theoretical developments mentioned above can be directly applied to establish the existence and uniqueness of our proposed Major-Minor MFG, due to the unique mathematical formulation that arises from our context of periodic trading.\vspace{.5em}

\noindent \underline{Liquidation games.} The research on optimal liquidation for a single trader traces back to at least \cite{bertsimas_optimal_1998} and \cite{almgren_optimal_2001}, which sought to identify dynamic trading strategies that minimize the expected trading cost over a fixed period. For a review of this extensive literature, we direct interested readers to \cite{cartea_algorithmic_2015} and references therein. 

A natural extension is to study optimal execution games with a finite number of players to understand the strategic interactions and competitions among traders. For example, \cite{brunnermeier_predatory_2005} and \cite{carlin_episodic_2007} analyzed a two-player game in which one player has a liquidation target, while the other exploits this information to engage in predatory trading. \cite{cont_fast_2023} studied a stochastic Stackelberg liquidation game between a low-frequency and a high-frequency trader, in which equilibrium strategies exhibit daily periodic patterns when the high-frequency trader adopts a periodic inventory penalty. Other studies on liquidation games, such as \cite{schoeneborn_liquidation_2009}, \cite{moallemi_strategic_2012}, \cite{schied_market_2018}, and \cite{vos_two-player_2022}, however, do not focus on the periodic phenomena. 

Stochastic games with more than two players are usually difficult to solve, which limits the insights they generate. In contrast, the mean-field approximation provides an ingenious and tractable aggregation approach to approximate the otherwise challenging $N$-player stochastic games. In the mean-field regime, \cite{carmona_probabilistic_2015} developed a probabilistic formulation to study price impact models, where players interact through the impact of their strategies on the payoff functions of the other players. This model has been further discussed in \cite{cardaliaguet_mean_2018} and \cite{casgrain_meanfield_2020}. \cite{fu_mean_2021} considered an MFG of optimal portfolio liquidation under asymmetric information with hard terminal conditions. Extensions include the incorporation of the market drop-out time \citep{fu2024mean} and the market entering time \citep{fu2024mean2}.

When both major and minor market participants are considered, \cite{huang_mean-field_2019} studied the Nash equilibrium of an optimal execution game using a dynamic programming approach. 
Although our formulation is largely inspired by \cite{huang_mean-field_2019}, their dynamic programming approach cannot handle the terminal state constraint in our setting. In addition, \cite{huang_mean-field_2019} does not consider targeting strategies for the major player. Finally, the Stackelberg MFG was analyzed in \cite{feron_price_2020}, \cite{fu_mean-field_2020} and \cite{evangelista_finite_2021} for liquidation games.

\paragraph{Organization and notations.}
Our paper is organized as follows. Section \ref{sec:setting} introduces the mathematical framework of the mean-field liquidation game and the definition of periodic trading strategies. Section \ref{sec:open-loop_equilibrium} establishes the existence and uniqueness of the open-loop Nash equilibrium. With these theoretical developments, we further investigate the periodic phenomena of the liquidation game in a few special scenarios in Section \ref{sec:periodic_targeting_strategy}. Finally, Section  \ref{sec:numerical_experiments} performs extensive numerical experiments to further investigate the periodicity and equilibrium behavior of traders.

Throughout, we adopt the convention that $C$ denotes a constant that may vary from line to line. Moreover, for a filtration $\mathbb{G}$ in a probability space $(\Omega,\mathcal{F},\prob)$, $\mathrm{Prog}(\mathbb{G})$ denotes the $\sigma$-field of progressive subsets of $[0, T]\times \Omega$, and for $\mathbb{I}$, which could be a subset of $\R^n$ with $n \geq 1$ or $\R\cup\{+\infty\}$, we consider the set of progressively measurable processes with respect to $\mathbb{G}$:
\begin{equation*}
\mathcal{P}_{\mathbb{G}}([0, T] \times \Omega ; \mathbb{I})=\set{u:[0, T] \times \Omega \rightarrow \mathbb{I}}{u \text { is } \mathrm{Prog}(\mathbb{G})\text { - measurable}}.
\end{equation*}
We define the following subspaces of $\mathcal{P}_{\mathbb{G}}([0, T] \times \Omega ; \mathbb{I})$:
\begin{equation*}
\begin{aligned}
L_{\mathbb{G}}^p([0, T] \times \Omega ; \mathbb{I}) &:= \set{u \in \mathcal{P}_{\mathbb{G}}([0, T] \times \Omega ; \mathbb{I})}{ \E\left(\int_0^T|u(t, \omega)|^2 d t\right)^{p/2}<\infty}, \\
S_{\mathbb{G}}^p([0, T] \times \Omega ; \mathbb{I}) &:= \set{u \in \mathcal{P}_{\mathbb{G}}([0, T] \times \Omega ; \mathbb{I})}{ \E\left(\sup _{0 \leq t \leq T}|u(t, \omega)|^p\right)<\infty}. 
\end{aligned} 
\end{equation*}
Whenever the notation $T-$ appears in the definition of a function space, we mean the set of all functions whose restrictions satisfy the respective property on $[0,\tau]$ for any $\tau<T$. For example, by $f \in L^{2}([0, T-] \times \Omega; \mathbb{I})$, we mean $f \in L^{2}([0,\tau] \times \Omega; \mathbb{I})$.

In addition, we denote the space of continuous functions on $[0,T]$ as $\mathcal{C}([0,T])$, and the space of piecewise differentiable functions on $[0,T]$ as $\mathrm{PD}([0,T])$.

Throughout, we measure the periodicity of the trading rates or prices by their amplitudes. Formally, for a periodic function $f(t)$ on $[0,T]$ with period $T/n$, we define its amplitudes as
\begin{equation}
\mathcal{A}[f]:= \frac{1}{2} \left[ \max_{t\in[0,T]} f(t) - \min_{t\in[0,T]} f(t) \right].
\label{eq:defn_amplitude}
\end{equation} 

\section{Mean-field Liquidation Game and Periodic Strategies}\label{sec:setting}

In this section, we introduce the mathematical framework of a mean-field liquidation game, 
 where the major trader evaluates her strategy against a periodic targeting strategy while a continuum of minor traders trade against her.

\subsection{Mean-field Liquidation Game with Major-Minor Players}
\label{sec:infinite_minor_traders}
Consider a usual probability space $(\Omega,\G,\prob)$ that supports two independent standard Brownian motions $W^{0}$ and $W$ of dimensions $1$ and $m - 1$, respectively. $(\Omega,\G,\prob)$ also carries an one-dimensional random variable $\Q^{\minor}_0$ that is independent of $(W^{0}, W)$. Let $\mathbb{F}^{\major} := (\F_{t}^{\major}, 0 \leq t \leq T)$ with $\F^{\major}_{t} = \sigma(W^{0}_{s}, 0 \leq s \leq t)$ the filtration generated by $W^0$ and let $\mathbb{F}^{\minor}:=(\F^{\minor}_{t}, 0 \leq t \leq T)$ with $\F^{\minor}_{t} = \sigma(\mathcal{Q}_0, W_{s}^{0},W_{s}, 0 \leq s \leq t)$. 

Now, we consider a market where a major trader and a continuum of minor traders aim to liquidate their inventories of a common asset during a fixed time interval $[0,T]$. Given that the minor traders are homogeneous, it is sufficient to study one representative among them. 
The inventory processes $Q^{\major} := (Q^{\major}_{t})_{t\in[0,T]} $ and $Q^{\minor} := (Q^{\minor}_{t})_{t\in[0,T]} $ for the major trader and the representative minor trader follow
\begin{equation*}
\begin{aligned}
\dif Q^{\major}_{t} = v^{\major}_{t} \dif t &, \text{ with boundary conditions } Q^{\major}_{0} = q^{\major}_{0},\, Q^{\major}_{T} = 0, \\
\dif Q^{\minor}_{t} = v^{\minor}_{t} \dif t &, \text{ with boundary conditions } Q^{\minor}_{0} = \mathcal{Q}^{\minor}_0,\, Q^{\minor}_{T} = 0,
\end{aligned}    
\end{equation*}
where $v^{\major} = (v^{\major}_{t})_{t\in[0,T]} $ and $v^{\minor} = (v^{\minor}_{t})_{t\in[0,T]} $ represent their trading rates.

Following the literature on price impact models \citep{almgren_optimal_2001,obizhaeva_optimal_2013,huang_mean-field_2019}, we assume that the trading activities result in both permanent and temporary price impacts subject to linear forms. More specifically, the price of the asset follows:
\footnote{Without loss of generality, we assume that the initial price $S_0 \equiv 0$. Otherwise, $S_t$ can be replaced by $S_t-S_0$.}
\begin{equation}
d S_t = \left(\lambda_0 v^{\major}_{t} + \lambda \bar{v}^{\minor}_t\right) \dif t + \sigma \dif W^{0}_{t}, \quad S_0 = 0.
\label{eq:fundamental_price_mfg}
\end{equation}
where $\sigma>0$, $\bar{v}^{\minor}_t := \E\left[v^{\minor}_{t}\middle|\F^{\major}_t\right] $ is the average trading rate of the minor traders, and $\lambda_0, \lambda>0$ are the strengths of the permanent impacts. In addition, the prices at which the major trader and the representative minor trader sell the asset, called the execution prices, follow
\begin{equation*}
\begin{aligned}
\bar{S}^{\major}_{t} = S_t + a_0 v^{\major}_{t}, \quad  \bar{S}^{\minor}_{t} = S_t + a v^{\minor}_{t},
\end{aligned}    
\end{equation*}
where $a_0>0$ and $a>0$ model the linear temporary price impact.

Consequently, the cash processes of the major trader $X^{\major} = (X^{\major}_t)_{t\in[0,T]}$ and the representative minor trader $X^{\minor}= (X^{\minor}_t)_{t\in[0,T]}$ follow, respectively
\begin{align}
X^{\major}_t &= - \int_0^t \bar{S}^{\major}_{u} v^{\major}_{u} \dif u, \label{eq:terminal_cash_major} \\ 
X^{\minor}_t &= - \int_0^t \bar{S}^{\minor}_{u} v^{\minor}_{u} \dif u. \label{eq:terminal_cash_minor}
\end{align}
Below, a more explicit expression of the terminal wealth is provided, which will be useful for the subsequent analysis.
\begin{prop}
\label{prop:terminal_cash_detail}
The terminal wealth of the major trader \eqref{eq:terminal_cash_major} and the representative minor trader \eqref{eq:terminal_cash_minor} are, respectively, 
\begin{equation}
\begin{aligned}
X^{\major}_{T} 
&= - \frac{\lambda_0}{2}(q^{\major}_0)^2 + \lambda \int_0^T Q^{\major}_{t}  \bar{v}^{\minor}_t \dif t  - a_0 \int_0^T \left(v^{\major}_{t}\right)^2 \dif t + \sigma \int_0^T Q^{\major}_{t} \dif W^{0}_{t},
\end{aligned}
\label{eq:terminal_cash_major_expand}
\end{equation}
and
\begin{equation}
\begin{aligned}
X^{\minor}_{T}
&=  \int_0^T Q^{\minor}_{t} \left(\lambda_0 v^{\major}_{t}+\lambda \bar{v}^{\minor}_t\right) \dif t  - a \int_0^T \left(v^{\minor}_{t}\right)^2 \dif t + \sigma \int_0^T Q^{\minor}_{t}  \dif W^{0}_{t}.
\end{aligned}
\end{equation}
\end{prop}
The proof of Proposition \ref{prop:terminal_cash_detail} is deferred to Section~\ref{sec:proof_terminal_cash_detail}.

Given the trading rate of the major trader, $v^{\major}$, and the average trading rate of minor traders, $\bar{v}^{\minor}$, the representative minor trader aims to minimize the following cost functional:
\begin{equation}
\begin{aligned}
&J^{\minor}\left(v^{\minor}; v^{\major}, \bar{v}^{\minor}, \mathcal{Q}^{\minor}_0\right) := 
- \E\left[X^{\minor}_{T}\right] + \phi \mathbb{E}\left[\int_0^T\left(Q^{\minor}_{s}\right)^2 \dif s \right] \\
& \qquad = - \underbrace{\E\left[ \int_0^T Q^{\minor}_{t} \left(\lambda_0 v^{\major}_{t} + \lambda \mu_t\right) \dif t - a \int_0^T \left(v^{\minor}_{t}\right)^2 \dif t \right]}_{\text{Expected profit of $Q^{\minor}$}} + \underbrace{\phi \E\left[ \int_0^T \left(Q^{\minor}_{t}\right)^2 \dif t \right]}_{\text{Risk penalty}},
\end{aligned}
\label{eq:cost_functional_minor}
\end{equation}
where $\phi>0$ is the risk penalty parameter of the minor traders. Note that the cost of the minor trader includes the expected profit from $Q^{\minor}$ and the risk penalty term $\E\left[ \int_0^T \left(Q^{\minor}_{t}\right)^2 \dif t \right]$, which is associated with variance of $X^{\minor}_{T}$.

The major player makes decisions based on  a {\it deterministic} targeting strategy $R = (R_t)_{t\in[0,T]}$ with $R_0 = q^{\major}_0$ and $R_T = 0$. More specifically,  the major trader aims to maximize the excess wealth compared to the targeting strategy subject to a  penalty, which follows a similar form as in \eqref{eq:cost_functional_minor}:
\begin{equation}
\begin{aligned}
&  J^{\major}\left(v^{\major};\bar{v}^{\minor}, R, q_{0}^{0} \right)\\
& := - \Bigg(\underbrace{\E\left[\lambda \int_0^T Q^{\major}_{t} v^{\minor}_{t} \dif t + a_0 \int_0^T (v^{\major}_{t})^2 \dif t\right]}_{\text{Expected profit of $Q^{\major}$}} -  \underbrace{\E\left[ \lambda\int_{0}^{T} R_{t} v^{\minor}_{t} \dif t\right]}_{\text{Expected profit of $R$}}\Bigg)  + \underbrace{\phi_0 \E\left[ \int_0^T (Q^{\major}_{t} - R_{t})^2 \dif t\right]}_{\text{Risk penalty}},   
\end{aligned}
\label{eq:cost_functional_major}
\end{equation}
where $\phi_0>0$ is the risk penalty parameter of the major trader. 
Note that for a special case without temporary impact ($a_0=0$), we have
\begin{eqnarray*}
     J^{\major}\left(v^{\major};\bar{v}^{\minor}, R, q_{0}^{0} \right)= - \left(\E\left[X_T^{\major}\right] - \E\left[X_T^{\major,\mathrm{Target}}\right] \right) + \phi_0 \E\left[\int_0^T (Q^{\major}_{t}-R_t)^2 \dif t\right]. 
\end{eqnarray*}
This is because by applying \eqref{eq:terminal_cash_major_expand}, we have
\begin{eqnarray*}
X^{\major,\mathrm{Target}}_{T} 
&= & - \frac{\lambda_0}{2}(q^{\major}_0)^2 + \lambda \int_0^T R_{t}  \bar{v}^{\minor}_t \dif t + \sigma \int_0^T R_{t}  \dif W^{0}_{t},\\
X^{\major}_{T} - X^{\major,\mathrm{Target}}_{T}  &=& \lambda \int_0^T \left(Q^{\major}_{t} - R_t\right) \bar{v}^{\minor}_t \dif t + \sigma \int_0^T \left( Q^{\major}_{t} - R_{t} \right)  \dif W^{0}_{t}.
\end{eqnarray*}

\begin{rmk} In practice, large institutional traders frequently face the challenge of liquidating specific shares based on a targeted strategy, often at the request of their clients \citep{cartea_algorithmic_2015, cheng_optimal_2024}. While these traders can refine the targeted strategy, they cannot deviate significantly from it.
This targeting strategy can be a TWAP strategy, an Almgren-Chriss optimal execution strategy, or a periodic targeting strategy (see details in Section~\ref{sec:periodic_strategy}). 
\end{rmk}

Throughout, we focus on open-loop Nash equilibrium for the analysis. To start, we define the admissible sets for the major trader and the representative minor trader as, respectively,
\begin{align}
\mathcal{A}^{\major} &:= \set{v^{\major} \in L_{\mathbb{F}^{\major}}^2([0, T] \times \Omega ; \R)}{\int_{0}^{T} v^{\major}_t \dif t = -q^{\major}_{0}}, \label{eq:admissible_set_major}\\
\mathcal{A}^{\minor} &:= \set{v^{\minor} \in L_{\mathbb{F}^{\minor}}^2([0, T] \times \Omega ; \R)}{\int_{0}^{T} v^{\minor}_t \dif t = - \mathcal{Q}^{\minor}_0,\, \prob-\mathrm{a.s.}}, \label{eq:admissible_set_minor}
\end{align}
{where the filtrations $\mathbb{F}^{\major},\mathbb{F}^{\minor}$ are defined at the beginning of Section \ref{sec:infinite_minor_traders} and  the notations $L_{\mathbb{F}^{\major}}^2,L_{\mathbb{F}^{\minor}}^2$ are defined at the end of Section \ref{sec:introduction}.} 
Note that traders have {\it asymmetric information}, as minor traders have access to {\it more} information compared to the major trader. Intuitively, $W^0$ represents the public information that drives the price, observable to all traders, while $W$ can be viewed as additional private information only available to the minor traders. In this framework, the major trader decides his strategy solely on public information, whereas minor traders trade with both public and private information.

With the above setup, our Major-Minor MFG encompasses the following two problems. First, the major trader faces a standard optimal control problem.
\begin{problem}[The Major Trader's Problem]
\label{Problem:major_trader_major}
Given the average trading rate of the minor traders  $\bar{v}^{\minor} = \E\left[v^{\minor}_{t}\middle|\F^{\major}_t\right]$  with an admissible strategy from the representative minor trader $v^{\minor} \in \mathcal{A}^{\minor} $, the stochastic optimal control problem of the major trader is
\begin{equation}
\begin{aligned}
\min_{v^{\major} \in \mathcal{A}^{\major}} \quad  & J^{\major}(v^{\major};\bar{v}^{\minor},R,q^{\major}_0)\\
\st \quad & d Q^{\major}_{t} = v^{\major}_{t} \dif t, \quad Q^{\major}_{0} = q^{\major}_{0},\, Q^{\major}_{T} = 0.
\end{aligned}
\label{problem:major_problem_major_trader}
\end{equation}
Here the cost functional $J^{\major}$ of the major trader is defined in \eqref{eq:cost_functional_major}, and the admissible set of the major trader, $\mathcal{A}^{\major}$, is defined in \eqref{eq:admissible_set_major}.
\end{problem}

{If \textbf{Problem}~\ref{Problem:major_trader_major} has a unique optimal strategy $v^{\major,*}$ (which can be proved; see Remark \ref{rmk:well-definedness})}, we can define the following response function:
\begin{equation}
\begin{aligned}
\Gamma^{\major}: \mathcal{A}^{\minor} &\to \mathcal{A}^{\major} 
\\
v^{\minor} & \mapsto v^{\major,*}.
\end{aligned} 
\label{eqn:major_reaction}
\end{equation}

Next, for the MFG of the minor traders, the representative minor trader solves the following fixed-point problem.
\begin{problem}[Representative Minor Trader's Problem]
\label{Problem:minor_trader}
Given the strategy of the major trader $v^{\major} = (v^{\major}_t)_{t\in[0,T]} \in \mathcal{A}^{\major}$, define a map $\mathcal{M}(\cdot;v^{\major}): L_{\mathbb{F}^{\major}}^2([0, T] \times \Omega ; \R) \to L_{\mathbb{F}^{\major}}^2([0, T] \times \Omega ; \R) $ as follows:
\begin{enumerate}
\item Fix a process $\mu = (\mu_t)_{t\in[0,T]} \in L_{\mathbb{F}^{\major}}^2([0, T] \times \Omega ; \R) $ that represents the average trading rate of the minor traders.
\item Solve the control problem faced by a representative minor trader:
\begin{equation}
\label{problem:minor_problem}
\begin{aligned}
\min_{v^{\minor} \in \mathcal{A}^{\minor}} \quad & J^{\minor}(v^{\minor}; v^{\major}, \mu, \Q^{\minor}_{0}) \\
\st \quad & \dif Q^{\minor}_{t} = v^{\minor}_{t} \dif t, \quad Q^{\minor}_{0} = \Q^{\minor}_{0},\, Q^{\minor}_{T} = 0,\\
\end{aligned}
\end{equation}
where the cost functional $J^{\minor}$ of the representative minor trader is defined as \eqref{eq:cost_functional_minor}, and the admissible set of the representative minor trader is defined in \eqref{eq:admissible_set_minor}.
\item Let $v^{\minor,*}$ be an optimal strategy of \eqref{problem:minor_problem}. {If the optimal strategy is unique (which can be proved; see Remark \ref{rmk:well-definedness})}, we can define the mapping $$\mathcal{M}(\mu;v^{\major}) = \left(\E\left[v^{\minor,*}_t \middle| \F^{\major}_t \right]\right )_{t\in[0,T]}.$$
\end{enumerate}
Conditioned on a fixed strategy from the major trader $v^{\major} = (v^{\major}_t)_{t\in[0,T]} \in \mathcal{A}^{\major}$, we say that $(\mu^{*}, v^{\minor,*}) \in L_{\mathbb{F}^{\major}}^2([0, T] \times \Omega ; \R) \times  L_{\mathbb{F}^{\minor}}^2([0, T] \times \Omega ; \R)$ is the corresponding Conditional Mean Field Equilibrium (CMFE),  if the following two conditions are satisfied:
\begin{itemize}
    \item[\bf (a)] {\bf Consistency:} $\mu^{*}$ satisfies $\mathcal{M}(\mu^{*};v^{\major}) = \mu^{*}$ (i.e., $\mu^{*}$ is a fixed-point of $\mathcal{M}(\cdot;v^{\major})$), and 
   \item[\bf (b)] {\bf Optimality:} $v^{\minor,*}$ is the optimal strategy of the control problem \eqref{problem:minor_problem} under $\mu = \mu^{*}$.
\end{itemize}
 
\end{problem}

We note that the CMFE of \textbf{Problem}~\ref{Problem:minor_trader}, $(\mu^{*}, v^{\minor,*})$, satisfies the following consistency condition:
\begin{equation}
\mu^{*}_t = \E\left[v^{\minor,*}_t \middle| \F^{\major}_t \right],\, \prob-\mathrm{a.s.},\quad t\in[0,T].
\label{eq:consistency_condition}
\end{equation}
Moreover, \textbf{Problem}~\ref{Problem:minor_trader} is an extended MFG \citep{gomes_extended_2016,graber_linear_2016}, where players interact with each other through their controls.

{If \textbf{Problem}~\ref{Problem:minor_trader} has a unique {CMFE} solution $(\mu^{*}, v^{\minor,*})$ for any given $v^{\major} \in \mathcal{A}^{\major}$ (which can be proved; see Remark \ref{rmk:well-definedness})}, the following well-defined response function of the representative minor trader
\begin{equation}
\begin{aligned}
\Gamma^{\minor}: \mathcal{A}^{\major} &\to \mathcal{A}^{\minor}\\
v^{\major} & \mapsto v^{\minor,*}.    
\end{aligned}
\label{eqn:minor_reaction}
\end{equation}

With the response functions of the major trader~\eqref{eqn:major_reaction} and the representative minor trader~\eqref{eqn:minor_reaction}, we define the open-loop Nash equilibrium for the above Major-Minor MFG as follows. 

\begin{defn}[Open-loop Nash Equilibrium for the Major-Minor MFG]
\label{defn:nash_equibrium}
A couple  of admissible strategies $(v^{\major,*},v^{\minor,*})\in \mathcal{A}^{\major} \times \mathcal{A}^{\minor}$ is an open-loop Nash equilibrium for the Major-Minor MFG if
\begin{equation*}
v^{\major,*} =  \Gamma^{\major}(v^{\minor,*}), \quad v^{\minor,*} =  \Gamma^{\minor}(v^{\major,*}).
\end{equation*}
\end{defn}

\begin{rmk}[Well-definedness of the open-loop Nash equilibrium]\label{rmk:well-definedness}
The well-definedness of the open-loop Nash equilibrium in \Cref{defn:nash_equibrium} relies on two key assumptions: the uniqueness of the solutions to the stochastic control problems \eqref{problem:major_problem_major_trader} and \eqref{problem:minor_problem}, and the uniqueness of the CMFE solution $(\mu^{*}, v^{\minor,*})$ in \textbf{Problem}~\ref{Problem:minor_trader}. The first uniqueness condition is easy to establish because both \eqref{problem:major_problem_major_trader} and \eqref{problem:minor_problem} are of linear-quadratic forms. In addition, the second uniqueness condition is proved in \Cref{thm:nash_fbsde_verification}.
\end{rmk}

\begin{rmk}[Nash vs. Stackelberg Equilibria] Note that both Major-Minor Nash equilibrium and Stackelberg equilibrium are used to study the game between a major player and a continuum of minor players. See \cite{feron_price_2020,fu_mean-field_2020,evangelista_finite_2021} for references on Stackerberg games. Mathematically, the difference is that in the Stackelberg game,   $\bar{v}^{\minor}$ in \textbf{Problem}~\ref{Problem:major_trader_major} is replaced by $\E\left[\Gamma^{\minor}(v^{\major})\middle|\F^{\major}_t\right]$. Intuitively, in the Stackelberg equilibrium, the major trader commits to a strategy first and then the minor traders make their decisions based on this strategy. However, in practice, public commitment to a strategy is infeasible because minor traders can arbitrage the committed strategy, leading major traders to naturally avoid such information-revealing disadvantages. Consequently, for the financial application we focus on, it is more appropriate to consider the Nash equilibrium defined in Definition~\ref{defn:nash_equibrium}, where all traders determine their strategies simultaneously. 
\end{rmk}

\subsection{Periodic Strategy}
\label{sec:periodic_strategy}
Throughout the paper, we focus on scenarios where the major trader evaluates her strategy against a periodic targeting strategy. In this section, we introduce the mathematical definitions of periodic strategies, starting with differentiable periodic strategies characterized by their trading rates and extending this concept to generic periodic strategies with piecewise differentiable inventory processes.

\begin{defn}[Differentiable Periodic Strategy]
\label{defn:periodic_strategy_cont}
A differentiable trading strategy with inventory process $Q = (Q_t)_{t\in[0,T]}$ and trading rate $v = (v_t)_{t\in[0,T]}$ is called a differentiable periodic strategy with $n$ periods and initial inventory $q_0$, if $Q_0=q_0$ and $v_t$ is a periodic function with $n$ periods in $[0,T]$. That is,
\begin{equation*}
v_t = v_{t+T/n},\quad  t\in\left[0,(n-1)T/n \right].
\end{equation*}
\end{defn}
The differentiable TWAP strategy is an example of a differentiable periodic strategy with a constant trading rate.
\begin{Eg}[Differentiable TWAP Strategy]
\label{eg:continuous_TWAP}
The inventory process of a differentiable TWAP strategy with initial inventory $q_0$ is defined as 
\begin{equation*}
Q_t^{\text{\emph{D-TWAP}},q_0} := q_0 \cdot \left(1-\frac{t}{T}\right),    
\end{equation*}
and its trading rate is 
\begin{equation*}
v_t^{\text{\emph{D-TWAP}},q_0} := - \frac{q_0}{T}.
\end{equation*}
\end{Eg}
Among the differentiable periodic strategies, the cosine trading strategy will be frequently used in later analysis.
\begin{Eg}[Cosine Trading Strategy]
\label{eg:cos_target}
The inventory process of a cosine trading strategy is defined as 
\begin{equation}
Q^{\cos,n,q_0,b}_t := q_0 \left(1 - \frac{t}{T}\right) + b \sin\left(\frac{2\pi n}{T}t\right),\quad t \in [0,T].
\end{equation}
The corresponding trading rate is
\begin{equation}
v^{\cos,n,q_0,b}_t = -\frac{q_0}{T} + b \frac{2\pi n}{T}\cos\left(\frac{2\pi n}{T} t\right).
\end{equation}
\end{Eg}

For a differentiable periodic strategy $Q$ with $n$ periods and an initial inventory $q_0$, one can show that $Q_{kT/n} = q_0(1-k/n)$ for $k=0,1,\ldots,n$. Then, for $\tilde{Q}_t:=Q_t - Q_t^{\text{D-TWAP},q_0}$, we have
\begin{equation*}
\tilde{Q}_{t+T/n} - \tilde{Q}_{t} = \int_{t}^{t+T/n} \left(v_t + \frac{q_0}{T}\right) \dif t = \int_{0}^{T/n} \left(v_t + \frac{q_0}{T} \right) \dif t = (Q_0 - Q_{T/n}) + \frac{q_0}{n} = 0
\end{equation*}
for any $t\in\left[0,(n-1)T/n \right]$. Therefore, $\tilde{Q}_t$ is a periodic function with $n$ periods in $[0,T]$ satisfying $\tilde{Q}_{kT/n}=0$ for $k=0,1,\ldots,n$. In other words, any differentiable periodic strategy following Definition~\ref{defn:periodic_strategy_cont} can be written as a combination of a differentiable TWAP strategy and a periodic round-trip strategy. 

Based on this observation, we can define a more general periodic strategy, relying solely on the inventory process, which may not have a trading rate as assumed in \Cref{defn:periodic_strategy_cont}.
\begin{defn}[Periodic Strategy]
\label{defn:periodic_strategy}
A trading strategy is called a periodic strategy with $n$ periods and initial inventory $q_0$, if the inventory process $Q = (Q_t)_{t\in[0,T]}$ has a decomposition:
\begin{equation}
Q_t = Q_t^{\text{\emph{D-TWAP}},q_0} + \tilde{Q}_t,
\label{eq:periodic_reservation}
\end{equation}
where $Q_t^{\text{\emph{D-TWAP}},q_0}$ is the inventory process of a differentiable TWAP strategy with initial inventory $q_0$, and $\tilde{Q}_t$ is a periodic piecewise differentiable function with $n$ periods in $[0,T]$ satisfying
\begin{equation}
\tilde{Q}_{k\cdot T/n } = 0, \quad k=0,1,\ldots,n.
\label{eq:periodic_round-trip_condition}
\end{equation}
\end{defn}

With \Cref{defn:periodic_strategy}, we can further consider a TWAP strategy that trades only at given discrete timestamps $t=kT/n$ for $k=0,1,\ldots,n$.  
\begin{Eg}[TWAP Strategy]
\label{eg:piecewise_const_target}
A trading strategy is called a TWAP strategy with $n$ periods and an initial inventory $q_0$ if its inventory process is
\begin{equation}
Q^{\mathrm{TWAP},q_0,n}_t = \begin{cases}
\frac{n-k}{n} q_{0},  & t = \frac{k}{n}T,\, k=0,1,\ldots,n,\\
\frac{2n-2k+1}{2n} q_{0},  & \frac{k-1}{n}T < t < \frac{k}{n}T,\, k=1,2,\ldots,n.\\
\end{cases}
\label{eq:piecewise_const_reservation}
\end{equation}
\end{Eg}
To see that \eqref{eq:piecewise_const_reservation} satisfies Definition \ref{defn:periodic_strategy},
\begin{equation*}
Q^{\mathrm{TWAP},q_0,n}_t - Q_t^{\text{\emph{D-TWAP}},q_0} = \begin{cases}
0,  & t = \frac{k}{n}T,\, k=0,1,\ldots,n,\\
-\frac{q_0}{2n} + \frac{q_0}{T} \left(t-\frac{k-1}{n}T\right),  & \frac{k-1}{n}T < t < \frac{k}{n}T,\, k=1,\ldots,n\\
\end{cases}
\end{equation*}
is a periodic piecewise constant function with $n$ periods in $[0,T]$ and satisfies \eqref{eq:periodic_round-trip_condition}.

\section{The Open-Loop Nash Equilibrium}
\label{sec:open-loop_equilibrium}
 In this section, we adopt a probabilistic approach to characterize the Nash equilibrium of the Major-Minor MFG introduced in Section~\ref{sec:infinite_minor_traders}. Our method is inspired by \cite{fu_mean_2021} and \cite{carmona_probabilistic_2018}. In particular,  \cite{fu_mean_2021} study the Nash equilibrium of extended MFGs {\it without} a major player, and Chapter 7.1 of \cite{carmona_probabilistic_2018} studies MFGs with a major player, where players interact with each other through the states {\it but not the controls}.  We integrate these existing methods to solve our Major-Minor MFG, with a major player and interactions through controls.

\subsection{Stochastic Maximum Principle}
For the major trader, the Hamiltonian of \textbf{Problem}~\ref{Problem:major_trader_major}  is defined
\begin{equation*}
    H^{\major}(t,q^{\major},p^{\major},v^{\major};\mu,R) = p^{\major} v^{\major} + \lambda \left(q^{\major} - R_t\right) \mu_t - a_0 \left(v^{\major}\right)^2 -\phi \left(q^{\major} - R_t\right)^{2},
\end{equation*}
where $H^{\major}:[0,T] \times \R \times \R \times \R \times L_{\mathbb{F}^{\major}}^2([0, T] \times \Omega ; \R) \times \mathrm{PD}([0,T]) \rightarrow \R$.
According to the stochastic maximum principle \citep{yong_stochastic_1999}, the solution to the optimization problem can be characterized as a forward-backward stochastic differential equation (FBSDE):
\begin{equation}
\left\{\begin{aligned}
\dif Q^{\major}_{t} & = v^{\major}_{t} \dif t, \\
- \dif P^{\major}_{t} & = \left(\lambda \E\left[v^{\minor}_{t} \middle|\F^{\major}_{t}\right] - 2\phi_0 \left(Q^{\major}_{t} - R_t\right) \right) \dif t - Z^{\major}_t \dif W^0_t, \\
Q^{\major}_{0} &= q^{\major}_{0}, \, Q^{\major}_{T} = 0.
\end{aligned}\right.
\label{eq:fbsde_major_nash_origin}
\end{equation}
Here the process $P^{\major}$ is the adjoint process to the controlled state process $Q^{\major}$, and $Z^{\major}_t\in L_{\F^{\major}}^2\left([0, T-] \times \Omega ; \R\right)$ is the martingale representation term of $P^{\major}$ with respect to $W^0$. The liquidation constraint $Q^{\major}_{T} = 0$ results in a {\it singularity} of the value function at the liquidation time. As a result, the terminal condition for $P^{\major}$ {\it cannot be determined a priori}. In particular, the first equation of \eqref{eq:fbsde_major_nash_origin} holds on $[0, T]$, and the second equation holds on $[0, T)$. A standard approach \citep{yong_stochastic_1999} yields the following candidate optimal control:
\begin{equation}
\begin{aligned}
v^{\major}_{t} & = \frac{P^{\major}_{t}}{2 a_{0}}.  
\end{aligned}
\label{eq:condidate_control_major}
\end{equation}
Combining \eqref{eq:fbsde_major_nash_origin} and \eqref{eq:condidate_control_major}, \textbf{Problem}~\ref{Problem:major_trader_major} reduces to the following FBSDE:
\begin{equation}
\left\{\begin{aligned}
\dif Q^{\major}_{t} & = \frac{P^{\major}_{t}}{2 a_{0}} \dif t, \\
- \dif P^{\major}_{t} & = \left(\lambda \E\left[v^{\minor}_{t} \middle| \F^{\major}_{t}\right] - 2\phi_0 \left(Q^{\major}_{t}-R_t\right)\right) \dif t - Z^{\major}_t \dif W^{0}_t, \\
Q^{\major}_{0} &= q^{\major}_{0}, \, Q^{\major}_{T} = 0.
\end{aligned}\right.
\label{eq:fbsde_major_nash}
\end{equation}

Similarly, the Hamiltonian of {\bf Problem}~\ref{Problem:minor_trader} for the representative minor trader is defined as 
\begin{equation*}
H^{\minor}(t, q^{\minor}, v^{\minor}, p^{\minor};v^{\major}, \mu) = p^{\minor} v^{\minor} + q^{\minor} \left(\lambda_0 v^{\major}_t + \lambda \mu_t \right) -a \left(v^{\minor}\right)^2 -\phi \left(q^{\minor}\right)^2,
\end{equation*}
where $H^{\minor}:[0,T] \times \R \times \R \times \R \times L_{\mathbb{F}^{\major}}^2([0, T] \times \Omega ; \R) \times L_{\mathbb{F}^{\major}}^2([0, T] \times \Omega ; \R) \rightarrow\R$. The stochastic maximum principle suggests that the solution to the optimization problem can be characterized in terms of the FBSDE
\begin{equation}
\left\{\begin{aligned}
\dif Q^{\minor}_{t} & = v^{\minor}_{t} \dif t, \\
- \dif P^{\minor}_{t} & = (\lambda_0 v^{\major}_{t} + \lambda \mu_t - 2\phi Q^{\minor}_{t}) \dif t -
Z^{\minor}_t \dif \tilde{W}_t, \\
Q^{\minor}_{0} &= \mathcal{Q}^{\minor}_{0}, \, Q^{\minor}_{T} = 0,
\end{aligned}\right.
\label{eq:fbsde_minor_nash_origin}
\end{equation}
where $\tilde{W} = (W^0, W)$, and $Z^{\minor}_t\in L_{\F^{\minor}}^2\left([0, T-] \times \Omega ; \R\right)$ is the martingale representation of $P^{\minor}$ with respect to $\tilde{W}$. Similar to \eqref{eq:fbsde_major_nash_origin}, the first equation of \eqref{eq:fbsde_minor_nash_origin} holds on $[0, T]$, and the second equation holds on $[0, T)$. A standard approach \citep{yong_stochastic_1999} yields the candidate optimal control
\begin{equation}
\begin{aligned}
v^{\minor}_{t} & = \frac{P^{\minor}_{t}}{2a}.
\end{aligned}
\label{eq:condidate_control_minor}
\end{equation}
Combining \eqref{eq:fbsde_minor_nash_origin} and \eqref{eq:condidate_control_minor} with the consistency condition \eqref{eq:consistency_condition}, the analysis of {\bf Problem} \ref{Problem:minor_trader} reduces to the following McKean-Vlasov FBSDE (MV-FBSDE):
\begin{equation}
\left\{
\begin{aligned}
\dif Q^{\minor}_{t} & = \frac{P^{\minor}_{t}}{2a} \dif t, \\
- \dif P^{\minor}_{t} & = \left(\lambda_0 v^{\major}_{t} + \lambda \E\left[\frac{P^{\minor}_{t}}{2a}\middle|\F^{\major}_t \right] - 2\phi Q^{\minor}_{t}\right) \dif t - Z^{\minor}_t \dif \tilde{W}_t, \\
Q^{\minor}_{0} &= \mathcal{Q}^{\minor}_{0}, \, Q^{\minor}_{T} = 0.
\end{aligned}
\right. 
\label{eq:fbsde_minor_nash}
\end{equation}

Finally, combining \eqref{eq:fbsde_major_nash} and \eqref{eq:fbsde_minor_nash}, the open-loop Nash equilibrium of the Major-Minor MFG is characterized by the following MV-FBSDE:
\begin{equation}
\left\{\begin{aligned}
\dif Q^{\major}_{t} & = \frac{P^{\major}_{t}}{2 a_{0}} \dif t, \\
\dif Q^{\minor}_{t} & = \frac{P^{\minor}_{t}}{2a} \dif t, \\
- d P^{\major}_{t} & = \left(\lambda \E\left[\frac{P^{\minor}_{t}}{2a}\middle|\F^{\major}_t \right] - 2\phi_0 (Q^{\major}_{t} - R_t)\right) \dif t - Z^{\major}_t \dif W^0_t, \\
- \dif P^{\minor}_{t} & = \left(\lambda_0 \frac{P^{\major}_{t}}{2 a_{0}} + \lambda \E\left[\frac{P^{\minor}_{t}}{2a}\middle|\F^{\major}_t \right] - 2\phi Q^{\minor}_{t}\right) \dif t - Z^{\minor}_t \dif \tilde{W}_t, \\
Q^{\major}_{0} &= q^{\major}_{0}, \, Q^{\major}_{T} = 0,\, Q^{\minor}_{0} = \mathcal{Q}^{\minor}_{0}, \, Q^{\minor}_{T} = 0. 
\end{aligned}\right.
\label{eq:fbsde_nash}
\end{equation}

\subsection{General Existence and Uniqueness of Solutions}
\label{sec:general_game_solution}
Now we are ready to proceed with the existence and uniqueness of the solutions to the MV-FBSDE system \eqref{eq:fbsde_nash}. To establish these results, we assume that the following conditions hold.
\begin{assumption}~
\label{assumption:fbsde}
\begin{enumerate}
\item[(i)] The initial inventory of the representative minor trader, $\mathcal{Q}^{\minor}_0$ has bounded second moment, namely, $\E\left[|\mathcal{Q}^{\minor}_0|^2\right] < +\infty$.
\item[(ii)] The deterministic target strategy of the major trader, $R =  \left(R_t\right)_{t\in[0,T]}$, satisfies $\int_0^T \left(R_t\right)^2 \dif t < +\infty$.
\item[(iii)] There exist constants $\theta_1,\theta_2,\theta_3>0$ such that
\begin{equation}
\theta_2 > \frac{\lambda_0}{2},\quad 1/(\theta_{1}^{-1}+\theta_{3}^{-1}) > \frac{\lambda}{2},\quad 
\theta_1 < \frac{8\phi_0 a}{\lambda}, \quad 
\frac{\lambda_0}{a_{0}} \theta_2 + \frac{\lambda}{a}\theta_3 < 8\phi.
\label{eq:assumption_coefficients}
\end{equation}
\end{enumerate}
\end{assumption}

Assumption~\ref{assumption:fbsde} (i) is standard for the MFG to be well-defined \citep{bensoussan_mean_2013,carmona_probabilistic_2018}. Assumption~\ref{assumption:fbsde} (ii) guarantees the objective function of the major player is well-defined {\citep{cartea_algorithmic_2015, cheng_optimal_2024}}.  Assumption~\ref{assumption:fbsde}-(iii) ensures that the permanent impacts of the major and minor traders ($\lambda_0$ and $\lambda$) will not dominate the temporary impacts ($a_0$ and $a$) nor the inventory penalties ($\phi_0$ and $\phi$). This is consistent with empirical findings observed in the financial market \citep{almgren_optimal_2001,cartea_incorporating_2016}. On the other hand, given that the major and minor traders interact with each other only through their permanent impacts on the price dynamics \eqref{eq:fundamental_price_mfg}, Assumption~\ref{assumption:fbsde}-(iii) is consistent with the weak interaction condition often assumed in the game theory literature (e.g., \citep{horst_stationary_2005}).

We are now ready to state the first main result.
\begin{thm}\label{thm:nash_fbsde_unique}
Under Assumption~\ref{assumption:fbsde}, there exists a unique solution $(Q^{\major}, Q^{\minor} ,P^{\major}, P^{\minor}, Z^{\major}, Z^{\minor})$ to the FBSDE \eqref{eq:fbsde_nash} in the space
\begin{equation}
\begin{aligned}
\mathcal{H}^{1}_{\mathbb{F}^{\major}} \times \mathcal{H}^{1}_{\mathbb{F}^{\minor}} &\times D_{\mathbb{F}^{\major}}^2([0, T] \times \Omega; \R) \times D_{\mathbb{F}^{\minor}}^2([0, T] \times \Omega; \R)  \\
& \qquad\qquad \times L_{\mathbb{F}^{\major}}^2\left([0, T-] \times \Omega ; \R\right) \times L_{\mathbb{F}^{\minor}}^2\left([0, T-] \times \Omega ; \R^m\right).    
\end{aligned} 
\end{equation}
\end{thm}
The proof of Theorem \ref{thm:nash_fbsde_unique} is deferred to Section \ref{proof:nash_fbsde_unique}.

Next, we show that the solution of \eqref{eq:fbsde_nash} indeed characterizes the Nash equilibrium of the Major-Minor MFG.
\begin{thm}\label{thm:nash_fbsde_verification}
Given the unique solution $\left(Q^{\major,*}, Q^{\minor,*} ,P^{\major,*}, P^{\minor,*}, Z^{\major,*}, Z^{\minor,*}\right)$ of the FBSDE \eqref{eq:fbsde_nash}, the unique Nash equilibrium of the Major-Minor MFG (defined in {\bf Problems}~\ref{Problem:major_trader_major} and \ref{Problem:minor_trader}) is
\begin{equation}
v^{\major,*}_{t}  = \frac{P^{\major,*}_{t}}{2 a_{0}},\quad v^{\minor,*}_{t} = \frac{P^{\minor,*}_{t}}{2a},\quad 0\leq t \leq T.
\label{eq:nash_equilibrium_strategies}
\end{equation}
Under this unique Nash equilibrium, the optimal cost functionals of the major trader and the representative minor trader are
\begin{equation*}
\begin{aligned}
&J^{\major}\left(v^{\major,*};v^{\minor,*},R,q^{\major}_0\right) \\
&= \E\left[a_0 \int_0^T v^{\major,*}_{t} \dif R_t - \frac{\lambda}{2}  \int_0^T \left(Q^{\major,*}_{t} - R_t \right) \E\left[v^{\minor,*}_{t}\middle|\F^{\major}_{t}\right] \dif t \right],
\end{aligned}
\end{equation*}
and
\begin{equation*}
\begin{aligned}
& J^{\major}\left(v^{\major,*};v^{\minor,*},R,q^{\major}_0\right)\\
&= \E\left[a \mathcal{Q}^{\minor}_0 v^{\minor,*}_0 - \frac{1}{2}\int_0^T Q^{\minor,*}_{t} \left(\lambda_0 v^{\major,*}_{t} + \lambda \E\left[v^{\minor,*}_{t}\middle|\F^{\major}_{t}\right] \right) \dif t  \middle| \mathcal{Q}^{\minor}_0 \right],
\end{aligned}
\end{equation*}
respectively.
\end{thm}
The proof of Theorem \ref{thm:nash_fbsde_verification} is deferred to Section \ref{proof:nash_fbsde_verification}.

\subsection{\texorpdfstring{$(N+1)$}{(N+1)}-player Game Approximation}
Finally, we state the approximation property of the Nash equilibrium of the Major-Minor MFG to a finite-player game. 

Throughout this section, we consider the same probability space $(\Omega,\G,\prob)$ as in Section~\ref{sec:infinite_minor_traders}. This space further supports $N$ independent $(m-1)$-dimensional standard Brownian motions $W^{1}, \ldots, W^{N}$, which are independent of $W^{0}$. In addition, $(\Omega,\G,\prob)$ carries $N$ independent and identically distributed random variables $\mathcal{Q}^{1}_0,\dots,\mathcal{Q}^{N}_0$ that share the same distribution $\mathcal{Q}^{\minor}_0$, and are independent of $W^{0}, W^{1},\dots,W^{N}$. We define $\tilde{\mathbb{F}}^{i}:=(\tilde{\F}^{i}_{t}, 0 \leq t \leq T)$ with $\tilde{\F}^{i}_{t} = \sigma(\mathcal{Q}^{i}_{0}, W^{0}_{s},W^{i}_{s}, 0 \leq s \leq t)$ for $i=1,\ldots,N$.

Consider a liquidation game between a major trader and $N$ minor traders, which is a finite-player counterpart of the MFG in Section~\ref{sec:infinite_minor_traders}. The major trader has an initial inventory $q^{\major}_0 \in \R$ and the $i$-th minor trader has an initial inventory $q^{i}_0\in \R$ for $i=1,2,\dots,N$. The inventory process of the major trader $Q^{\major} = (Q^{\major}_{t})_{t\in[0,T]} $ and the $i$-th minor trader $Q^{i} = (Q^{i}_{t})_{t\in[0,T]} $ follow, respectively,
\begin{equation*}
\begin{aligned}
\dif Q^{\major}_{t} &= v^{0}_{t} \dif t, \text{ with boundary conditions } Q^{\major}_{0} = q^{\major}_{0},\, Q^{\major}_{T} = 0, \\
\dif Q^{i}_{t} &= v^{i}_{t} \dif t, \text{ with boundary conditions } Q^{i}_{0} = q^{i}_{0},\, Q^{i}_{T} = 0,
\end{aligned}    
\end{equation*}
where $v^{\major} = (v^{\major}_{t})_{t\in[0,T]} $ and $v^{i} = (v^{i}_{t})_{t\in[0,T]} $ are the trading rates. Moreover, the cost functionals of the major and the $i$-th minor trader are, respectively,
\begin{equation}
\begin{aligned}
&  J^{N,\major}\left(v^{\major};\bm{v}^{(N)}, R, q_{0}^{0} \right) \\
& \qquad := \E\left[-\lambda \int_0^T (Q^{\major}_{t}-R_t) \bar{v}^{(N)}_{t} \dif t + a_0 \int_0^T \left(v^{\major}_{t}\right)^2 \dif t + \phi_0 \int_0^T (Q^{\major}_{t}-R_t)^2 \dif t\right],
\end{aligned}
\label{eq:cost_functional_major_N}
\end{equation}
and
\begin{equation}
\begin{aligned}
&J^{N,i}\left(v^{i}; v^{\major}, \bm{v}^{(N),-i}, q^i_0\right) \\
&\qquad := \E\left[- \int_0^T Q^{i}_{t} \left(\lambda_0 v^{\major}_{t} + \lambda \bar{v}^{(N)}_{t}\right) \dif t + a \int_0^T \left(v^{i}_{t}\right)^2 \dif t + \phi \int_0^T\left(Q^{i}_{t}\right)^2 \dif t \middle| \mathcal{Q}^{i}_0=q^i_0\right].
\end{aligned}
\label{eq:cost_functional_minor_N}
\end{equation}
Here $\bm{v}^{(N)} := (v^{1},\ldots,v^{N})$, $\bm{v}^{(N),-i} := (v^{1},\dots,v^{i-1},v^{i+1},\dots,v^{N})$ for $i=1,2,\dots,N$, and $\bar{v}^{(N)}_{t} := \sum_{i=1}^{N}v^{i}_t / N$.

{ The Yamada-Watanabe result for MV-FBSDE \citep[Theorem 1.33]{carmona_probabilistic_2018} yields that}, there exist two measurable functions $\psi^{\major}: \R \times \mathcal{C}([0,T]) \to \mathcal{C}([0,T-]) $ and $\psi^{\minor}: \R \times (\mathcal{C}([0,T]))^{2} \to \mathcal{C}([0,T-]) $ such that the equilibrium trading rates can be expressed as 
\begin{equation}
v^{\major,*} = \psi^{\major}(q^{\major}_0, W^{0}),\quad v^{\minor,*} = \psi^{\minor}(\mathcal{Q}^{\minor}_0, W^{0}, W). 
\label{eq:MFG_Nash_Watanabe}
\end{equation}
From the integrability results in \Cref{thm:nash_fbsde_unique}, we know that, there exist a constant $\kappa \in \R$ satisfying
\begin{equation}    \E\left[\left|\psi^{\major}(q^{\major}_0, W^{0})\right|^2\right] \leq \kappa, \quad \E\left[\left|\psi^{\minor}(\mathcal{Q}^{\minor}_0, W^{0}, W)\right|^2\right] \leq \kappa,
\label{eq:condition_kappa}
\end{equation}
and a deterministic function $K:\R\to\R$ satisfying
\begin{equation}
\E\left[\left|\psi^{\minor}(\mathcal{Q}^{\minor}_0, W^{0}, W)\right|^2 \middle| \mathcal{Q}^{\minor}_0 = q^{\minor}_0\right] \leq K(q^{\minor}_0). 
\label{eq:condition_K}
\end{equation}
With the constant $\kappa$ in \eqref{eq:condition_kappa} and the function $K$ in \eqref{eq:condition_K}, we consider the following admissible sets for the major trader and the $i$-th minor trader:
\begin{equation}
\mathcal{A}^{\major,\kappa} := \set{v \in L_{\mathbb{F}^{\major}}^2([0, T] \times \Omega ; \R)}{\int_{0}^{T} v^{\major}_t \dif t = - q^{\major}_0,\, \E\left[\int_{0}^{T} \left|v^{\major}_t\right|^2 \dif t \right]\leq \kappa}
\label{eq:admissible_set_major_constrained}
\end{equation}
and 
\begin{equation}
\mathcal{A}^{i} := \set{v \in L_{\tilde{\mathbb{F}}^i}^2([0, T] \times \Omega ; \R)}{\int_{0}^{T} v^i_t \dif t = - \mathcal{Q}^i_0,\, \E\left[\int_{0}^{T} \left|v^i_t\right|^2 \dif t \middle| \mathcal{Q}^i_0 = q^{i}_0 \right]\leq K\left(q^{i}_0\right)}.
\label{eq:admissible_set_i_constrained}
\end{equation}

We are now ready to state the approximate Nash equilibrium result for the Major-Minor MFG in \Cref{sec:infinite_minor_traders}. More specifically, we prove that 
\begin{equation}
v^{\major,*} := \psi^{\major}(q^{\major}_0, W^{0}), \quad v^{i,*} := \psi^{\minor}(\mathcal{Q}^{i}_0, W^{0}, W^i),\, i=1,\dots,N.
\label{eq:approximate_Nash_trading_rate}
\end{equation}
form an $\mathcal{O}(N^{-1/2})$--Nash equilibrium for the finite-player game \eqref{eq:cost_functional_major_N}-\eqref{eq:cost_functional_minor_N}, with admissible sets defined in \eqref{eq:admissible_set_major_constrained} and \eqref{eq:admissible_set_i_constrained}. 
\begin{thm}[Approximate Nash Equilibrium]\label{thm:approximate_Nash_equilibrium} 
Given $v^{\major,*}$ and $v^{i,*}$ in \eqref{eq:approximate_Nash_trading_rate}, define $\bm{v}^{(N),*} := (v^{1,*},\ldots,v^{N,*})$ and $\bm{v}^{(N),-i,*} := (v^{1,*},\dots,v^{i-1,*},v^{i+1,*},\dots,v^{N,*})$ for $i=1,2,\dots,N$. Then, for any $v^{\major}\in\mathcal{A}^{\major,\kappa}$, we have
\begin{equation}
J^{N,\major}\left(v^{\major};\bm{v}^{(N),*}, R, q^{\major}_{0} \right)   \geq J^{N,\major}\left(v^{\major,*};\bm{v}^{(N),*}, R, q^{\major}_{0} \right) - \lambda \frac{2T\kappa}{\sqrt{N}},
\label{eq:approximate_Nash_major}
\end{equation}
where $J^{N,\major}$ is defined in \eqref{eq:cost_functional_major_N}.
In addition, for $i=1,2,\dots,N$ and for any $v^{i} \in \mathcal{A}^{i}$, 
\begin{equation}
\begin{aligned}
J^{N,i}\left(v^{i}; v^{\major,*}, \bm{v}^{(N),-i,*}, q^{i}_0\right) &\geq J^{N,i}\left(v^{i,*}; v^{*,0}, \bm{v}^{(N),-i,*}, q^i_0\right) \\
&\qquad \qquad - \frac{2T}{N} \sqrt{ K\left(q^{i}_{0}\right)  \left[ (N+1) \kappa + 2 K\left(q^{i}_{0}\right)\right] } ,     
\end{aligned}
\label{eq:approximate_Nash_minor}
\end{equation}
where $J^{N,i}$ is defined in \eqref{eq:cost_functional_minor_N}.
\end{thm}
Our $\mathcal{O}(N^{-1/2})$--approximation Nash equilibrium result is consistent with the approximation results in the literature for various MFGs settings; see \cite{nourian_epsilon-nash_2013}, \cite{huang_mean-field_2019}, and \cite{feron_price_2020} for MFGs with a major player. For MFGs without major players, see \cite{caines_large_2006} and \cite{lasry_mean_2007} for classical MFGs, \citet[Theorem 6.9]{carmona_probabilistic_2018} for MFGs with common noise, \cite{cao2022approximation}, \cite{guo2019stochastic}, and \cite{cao2023stationary} for MFG with singular controls, \cite{basei2019nonzero} for MFG with impulse controls, and  \cite{fu_mean_2021} and \cite{feron_price_2022} for extended MFGs.

The proof of Theorem~\ref{thm:approximate_Nash_equilibrium} is deferred to Section~\ref{proof:approximate_Nash_equilibrium}.

\subsection{Special Cases}
In this section, we discuss two special cases of the Major-Minor MFG.
\subsubsection{Special Case I: No Interactions}
\label{sec:no_game_case}

We consider a special case where the major trader does not affect the trading of minor traders, and vice versa. In this case, traders only consider their own price impacts, similar to the concept of ``orderly liquidation value'' introduced in \cite{brunnermeier_predatory_2005}. This scenario serves as a benchmark to highlight the effects of interactions between traders.

Setting $\lambda=0$ in {\bf Problem}~\ref{Problem:major_trader_major} suggests that the trading activities of minor traders do not affect the strategy of the major trader. In this case,  the stochastic control problem for the major trader reduces to
\begin{equation}
\begin{aligned}
\min_{v^{\major} \in \mathcal{A}^{\major}} \quad  & \E\left[a_0 \int_0^T (v^{\major}_{t})^2 \dif t  + \phi_0 \int_0^T (Q^{\major}_{t} - R_t)^2 \dif t\right] \\
\st \quad & d Q^{\major}_{t} = v^{\major}_{t} \dif t, \quad Q^{\major}_{0} = q^{\major}_{0},\, Q^{\major}_{T} = 0.
\end{aligned}
\label{problem:major_single_reservation}
\end{equation} 
This formulation was studied in \cite{cheng_optimal_2024}, and the optimal strategy of the major trader can be characterized explicitly as 
\begin{equation}
\begin{aligned}
Q^{\major,*,\lambda=0}_{t} &=  q^{\major}_0 \cdot \frac{\sinh (\theta_0 (T-t))}{\sinh(\theta_0 T)} \\
&\qquad + \frac{\theta_0}{\sinh (\theta_0 T)} \int_0^T R_s \sinh (\theta_0 \min (s, t)) \sinh (\theta_0 (T-\max (s, t)))  \dif s,
\end{aligned}
\label{eq:major_single_reservation}
\end{equation}
where $\theta_0 = \sqrt{\phi_0/a_0}$.

Similarly, we set $\lambda_0=0$ in {\bf Problem} \ref{Problem:minor_trader}, which implies that the trading activities from the major trader do not affect the strategy of the minor traders. In this case, minor traders face an MFG that was studied in~\cite{fu_mean_2021}. In this case, the equilibrium strategy of the representative minor trader follows
\begin{equation}
    Q^{\minor,*,\lambda_0=0}_t =  \exp \left(-\frac{\lambda}{4 a} t\right) \frac{\sinh (\gamma(T-t))}{\sinh (\gamma T)} \E[\mathcal{Q}^{\minor}_0]  +  \frac{\sinh (\theta(T-t))}{\sinh (\theta T)} (\mathcal{Q}^{\minor}_0-\E[\mathcal{Q}^{\minor}_0]),
    \label{eq:MFG_soln_different_q0}
\end{equation}
where $\theta=\sqrt{\phi/a}$ and $\gamma=\sqrt{\frac{\phi}{a} + \frac{\lambda^{2}}{16 a^{2}}}$.

 Comparing the results in Sections~\ref{sec:general_game_solution} and \ref{sec:no_game_case}, we see that the targeting strategy $R$ directly affects the trading of the major trader, and further influences minor traders through the interactions (if  exist) between the traders. More specifically, in the scenario without interactions, the targeting strategy $R$ influences {\it only} the major trader's strategy \eqref{eq:major_single_reservation} and does not affect the strategy of the representative minor trader \eqref{eq:MFG_soln_different_q0}. As a comparison, the equilibrium strategies of both the major and minor traders in Theorem~\ref{thm:nash_fbsde_verification} are affected by the targeting strategy.

\subsubsection{Special Case II: Identical Inventory of Minor Traders}
\label{sec:common_inventory}
Here we consider a special case where all minor traders have the same constant initial inventory, namely, $\mathcal{Q}^{\minor}_0 \equiv q^{\minor}_{0}$ for some constant $q^{\minor}_{0}$ almost surely. In this case, the strategies of the major trader and the representative minor trader are both deterministic, and are characterized by the first-order ODE system of $(Q^{\major}, Q^{\minor} ,P^{\major}, P^{\minor})$:
\begin{equation}
\left\{\begin{aligned}
\dif Q^{\major}_{t} & = v^{\major}_{t} \dif t, \\
\dif Q^{\minor}_{t} & = v^{\minor}_{t} \dif t, \\
- d v^{\major}_{t} & = \left(\frac{\lambda}{2 a_0} v^{\minor}_{t}  -  \frac{\phi_0}{a_0} (Q^{\major}_{t} - R_t)\right) \dif t, \\
- \dif v^{\minor}_{t} & = \left(\frac{\lambda_0}{2a} v^{\major}_{t}  + \frac{\lambda}{2a} v^{\minor}_{t} - \frac{\phi}{a} Q^{\minor}_{t}\right) \dif t, \\
Q^{\major}_{0} &= q^{\major}_{0}, \, Q^{\major}_{T} = 0,\, Q^{\minor}_{0} = q^{\minor}_{0}, \, Q^{\minor}_{T} = 0. 
\end{aligned}\right.
\label{eq:fbsde_nash_deterministic}
\end{equation}
In addition, \eqref{eq:fbsde_nash_deterministic} is equivalent to the following second-order ODE system
\begin{equation}
\left\{
\begin{aligned}
a_0 \frac{\dif^2 Q^{\major}_{t}}{\dif t^2} - \phi_0 Q^{\major}_{t} &= - \phi_0 R_t - \frac{\lambda}{2}\frac{\dif Q^{\minor}_{t}}{\dif t},\\ 
a \frac{\dif^2 Q^{\minor}_{t}}{\dif t^2} - \phi Q^{\minor}_{t} &= - \frac{\lambda_0}{2}\frac{\dif Q^{\major}_{t}}{\dif t} - \frac{\lambda}{2}\frac{\dif Q^{\minor}_{t}}{\dif t},
\end{aligned}\right.
\label{eq:major_trader_ODE}
\end{equation}
with boundary conditions $Q^{\major}_{0}=q^{\major}_{0}$, $Q^{\major}_{T}=0$, $Q^{\minor}_{0} = q^{\minor}_0$, and $Q^{\minor}_{T} = 0$. 

This special case serves as a primary example in the discussion of the periodicity phenomena of Major-Minor MFG, as detailed in Section~\ref{sec:periodic_targeting_strategy}.

\section{Periodic Phenomena in the Liquidation Game}
\label{sec:periodic_targeting_strategy}
In this section, we analyze the periodicity in the Major-Minor MFG from a theoretical perspective,  focusing on the scenario that the targeting strategy of the major trader is a periodic strategy (see Definition~\ref{defn:periodic_strategy}).
Section~\ref{sec:periodic-trend_decomposition} provides a periodic-trend decomposition of the Nash equilibrium, and Section~\ref{sec:major_trader_decomposition_cosine} further analyzes the periodic phenomena in both the traders' behaviors and the dynamics of the aggregated market when the targeting strategy is a cosine strategy.

\subsection{Periodic-Trend Decomposition of the Nash Equilibrium}
\label{sec:periodic-trend_decomposition}
To analyze the periodicity in the Nash equilibrium of the Major-Minor MFG, we first prove a general result to decompose the strategies of the major and the representative minor traders into periodic components and trend components.
\begin{thm}[Periodic-Trend Decomposition of the Nash Equilibrium] \label{thm:major_trader_decomposition}
Assume that the targeting strategy $R$ in \eqref{problem:major_single_reservation} is a periodic strategy with $n$ periods and initial inventory $q^{\major}_0$. In addition, assume  all minor traders have the same initial inventories, that is, $\mathcal{Q}^{\minor}_0\equiv q^{\minor}_0$ almost surely. Then the Nash equilibrium $(Q^{\major,*}_{t}, Q^{\minor,*}_{t})$ of the Major-Minor MFG given by {\bf Problems}~\ref{Problem:major_trader_major} and \ref{Problem:minor_trader} has the following decomposition:
\begin{equation}
\begin{aligned}
Q^{\major,*}_{t} &=  Q^{\major,\period}_{t} + Q^{\major,\trend}_{t} = q^{\major}_{0} \cdot \left(1-\frac{t}{T}\right) + \tilde{Q}^{\major,\period}_{t} + Q^{\major,\trend}_{t},\\
Q^{\minor,*}_{t} &= Q^{\minor,\period}_{t} + Q^{\minor,\trend}_{t},
\end{aligned}
\label{eq:periodic_decomposition}
\end{equation}
where
\begin{itemize}
\item $Q^{\major,\period}_{t}$ and $Q^{\minor,\period}_{t}$ are both periodic strategies with $n$ periods, and $(\tilde{Q}^{\major,\period}_{t}, Q^{\minor,\period}_{t})$ is the unique periodic solution to the following ODE:
\begin{equation}
\begin{aligned}
a_0 \frac{\dif^2 \tilde{Q}^{\major,\period}_{t}}{d t^2} - \phi_0 \tilde{Q}^{\major,\period}_{t} &= - \phi_0 \tilde{R}_t - \frac{\lambda}{2}\frac{d Q^{\minor,\period}_{t}}{d t},\\ 
a \frac{d^2 Q^{\minor,\period}_{t}}{d t^2} - \phi Q^{\minor,\period}_{t} &= - \frac{\lambda_0}{2}\frac{d \tilde{Q}^{\major,\period}_{t}}{d t} - \frac{\lambda}{2}\frac{d Q^{\minor,\period}_{t}}{d t}.
\end{aligned}
\label{eq:periodic_decomposition_periodic}
\end{equation}
\item $Q^{\major,\trend}_{t}$ and $Q^{\minor,\trend}_{t}$ satisfy the following ODE:
\begin{equation}
\begin{aligned}
a_0 \frac{d^2 Q^{\major,\trend}_{t}}{d t^2} - \phi_0 Q^{\major,\trend}_{t} &= - \frac{\lambda}{2}\frac{d Q^{\minor,\trend}_{t}}{d t},\\ 
a \frac{d^2 Q^{\minor,\trend}_{t}}{d t^2} - \phi Q^{\minor,\trend}_{t} &= - \frac{\lambda_0}{2}\left(\frac{d Q^{\major,\trend}_{t}}{d t} - \frac{q^{\major}_{0}}{T}\right) - \frac{\lambda}{2}\frac{d Q^{\minor,\trend}_{t}}{d t},
\end{aligned}
\label{eq:periodic_decomposition_trend}
\end{equation}
with the boundary conditions $Q^{\major,\trend}_{0} = Q^{\major,\trend}_{T} = -Q^{\major,\period}_{0}$, $Q^{\minor,\trend}_{0} = q^{\minor}_{0} - Q^{\minor,\period}_{0}$, and $Q^{\minor,\trend}_{T} = - Q^{\minor,\period}_{0}$.
\end{itemize}
\end{thm}
The proof of \Cref{thm:major_trader_decomposition} is deferred to \Cref{proof:major_trader_decomposition}.

To gain additional insights, the next result characterizes the strategies in \Cref{thm:major_trader_decomposition} in the special case of no interaction through the price impact. We use the same setup as in \Cref{sec:no_game_case}.
\begin{prop}\label{prop:major_periodicity_no_game}
Under the conditions of Theorem~\ref{thm:major_trader_decomposition},
\begin{enumerate}
    \item If $\lambda = 0$, the trend component satisfies $Q_t^{\major,\trend} \equiv 0$ and the  overall strategy $Q_t^{\major,*} $ is a periodic strategy with $n$ periods for the major trader.
    \item If $\lambda_0 = 0$,  the periodic component satisfies $Q_t^{\minor,\period}\equiv0$ for the representative minor trader.
\end{enumerate}
\end{prop}

As shown in \Cref{prop:major_periodicity_no_game} and \Cref{thm:major_trader_decomposition}, the major and minor traders interact with each other in both the periodic and trend components. As a benchmark, \Cref{prop:major_periodicity_no_game} demonstrates that in the absence of interactions, the major trader employs a periodic strategy, whereas the representative minor trader follows a strategy without periodic patterns. In contrast, when $\lambda\neq 0$ and $\lambda_0\neq 0$, \Cref{thm:major_trader_decomposition} indicates that the strategy of the major trader evolves to incorporate a trend component in response to the price impact caused by the representative minor trader. Simultaneously, the strategy of the representative minor trader shifts to integrate periodic patterns, reacting to the periodic actions of the major trader. The proof of \Cref{prop:major_periodicity_no_game} is deferred to \Cref{proof:major_periodicity_no_game}.

\subsection{Special Case: Cosine Targeting Strategy}
\label{sec:major_trader_decomposition_cosine}
Next, we consider a special case where the targeting strategy $R$ is a cosine trading strategy in \Cref{eg:cos_target} with $n$ periods, initial inventory $q^{\major}_{0}$, and amplitude $b>0$. Formally, we consider
\begin{equation}
R_t = Q^{\cos,n,q^{\major}_0,b}_t = q^{\major}_0 \left(1 - \frac{t}{T}\right) + b \sin\left(\omega t\right),\quad t \in [0,T],
\label{eq:cosine_targeting_strategy}
\end{equation}
where $\omega = 2\pi n / T$.

In this setting, explicit characterizations of the periodic components can be obtained through straightforward calculations.
\begin{prop}\label{prop:periodic_component_fourier}
Under the conditions of Theorem~\ref{thm:major_trader_decomposition}, when the targeting strategy $R$ is the cosine trading strategy \eqref{eq:cosine_targeting_strategy},
we have
\begin{equation}
\begin{aligned}
\tilde{Q}^{\major,\period}_{t} &= \frac{b\phi_0}{K}[\left(d_{0} d_{1}^{2} + d_{1}e_{0}e_{1} + d_{0} e_{1}^{2}\right) \sin\left(\omega t\right) - e_{0} e_{1}^{2} \cos\left(\omega t\right)],\\
Q^{\minor,\period}_{t} &= \frac{b\phi_0}{K}[- d_{0} e_{0} e_{1} \sin\left(\omega t\right) + e_{0}(d_{0} d_{1} + e_{0}e_{1}) \cos\left(\omega t\right)],
\end{aligned}
\label{eq:periodic_component_cosine}
\end{equation}
where $d_{0} = a_0 \omega^2 + \phi_0$, $d_{1} = a \omega^2 + \phi$, $e_{0} = \frac{\lambda_0}{2}\omega$, $e_{1} = \frac{\lambda}{2}\omega$, and $K =  d_{0}^{2} d_{1}^{2} + 2 d_{0}d_{1}e_{0}e_{1} + d_{0}^{2} e_{1}^{2} + e_{0}^{2} e_{1}^{2}$. 

In particular, if $\lambda=0$, then the optimal strategy of the major trader is
\begin{equation}
Q^{\major,*,\lambda=0}_t = q^{\major}_0 \left(1 - \frac{t}{T}\right) + b \frac{\phi_0}{d_0} \sin\left(\omega t\right),
\label{eq:major_rate_lambda_0}
\end{equation}
which is a cosine trading strategy with $n$ periods and amplitude $ b \phi_0 / d_0 $.
\end{prop}

With these results, we can further explore the periodicity in the trading rates of both the major and representative minor traders, as well as the periodicity in the aggregated trading rate and the price of the asset.

\subsubsection{Periodicity in Traders' Behavior}

Define $\tilde{v}^{\major,\period}_{t}:= \dif \tilde{Q}^{\major,\period}_{t} / \dif t$ and $ v^{\minor,\period}_{t}:= \dif Q^{\minor,\period}_{t} / \dif t$ as the periodic components of the trading rates of the major trader and the representative minor trader, respectively. Then we have the following result on the interactions between them.
\begin{prop}[Periodicity in the Trading Rates of the Major and Minor Traders]\label{prop:periodic_trading_rate}
Under the same assumptions of \Cref{prop:periodic_component_fourier}, we have:
\begin{equation}
\begin{aligned}
\tilde{v}^{\major,\period}_{t} =  \mathcal{A}\left[\tilde{v}^{\major,\period}\right] \cos\left(\omega t - \varphi^{\major}\right),\quad v^{\minor,\period}_{t} = \mathcal{A}\left[v^{\minor,\period}\right] \cos\left(\omega t - \varphi^{\minor}\right),
\end{aligned}
\label{eq:trading_rate_amplitude}
\end{equation}
where $\varphi^{\major} \in (-\pi,\pi)$ and $\varphi^{\minor} \in (-\pi,\pi)$ are the phases of $\tilde{v}^{\major,\period}$ and $v^{\minor,\period}$, and $\mathcal{A}\left[\tilde{v}^{\major,\period}\right]$ and $\mathcal{A}\left[v^{\minor,\period}\right]$ are the amplitudes of $\tilde{v}^{\major,\period}$ and $v^{\minor,\period}$ defined in \eqref{eq:defn_amplitude}. Furthermore,
\begin{enumerate}
\item $\varphi^{\major} \in \left[0,\frac{\pi}{2}\right]$, $\varphi^{\minor} \in \left[- \pi, - \frac{\pi}{2}\right]$, and $\varphi^{\major} - \varphi^{\minor} \leq \pi$.
\item $\mathcal{A}\left[\tilde{v}^{\major,\period}\right] \leq \mathcal{A}\left[v^{\major,*,\lambda=0}\right] $ with $v^{\major,*,\lambda=0}_t:= \dif Q^{\major,*,\lambda=0}_{t} / \dif t$.
\end{enumerate}
\end{prop}

\Cref{prop:periodic_trading_rate} implies that the periodic component of the minor trader achieves its peak before the major trader, indicating that the minor trader exhibits {\it front-running behavior} by taking advantage of the periodic patterns from the major trader in equilibrium. The major trader reduces the strength of her periodic behavior in response to the strategic (and adversarial) actions of the minor trader. The proof of Proposition~\ref{prop:periodic_trading_rate} is deferred to Section~\ref{proof:periodic_trading_rate}.

\subsubsection{Periodicity in the Aggregated Market Dynamics}
\label{sec:period_whole_market}
In this section, we study the aggregated market dynamics, which is affected by the behaviors of individual traders. In particular, we consider the periodic components of the price and the aggregated trading rate in equilibrium:
\begin{equation}
S^{\period}:= \lambda_0 Q^{\major,\period}_{t} + \lambda Q^{\minor,\period}_{t},
\end{equation}
and 
\begin{equation}
v^{\agg,\period}:= v^{\major,\period}_{t} + v^{\minor,\period}_{t},
\end{equation}
respectively. The strength of the periodicity in the price and the aggregate trading rate is measured by the amplitudes of their respective periodic components $\mathcal{A}[S^{\period}]$ and $\mathcal{A}[v^{\agg,\period}]$. As a benchmark comparison, we introduce the periodic components of the price and the aggregate trading rate without interaction:
\begin{equation}
S^{\period,\nogame}:= \lambda_0 Q^{\major,\period,\lambda=0}_{t}
\end{equation}
and 
\begin{equation}
v^{\agg,\period,\nogame}:= v^{\major,\period,\lambda=0}_{t}.
\end{equation}
Then we have the following results.
\begin{prop}[Periodicity in the Price]
\label{prop:periodic_aggrate_price}
Assume the same assumptions as in \Cref{prop:periodic_component_fourier}. Then we have
$$\mathcal{A}\left[S^{\period}\right] \leq \mathcal{A}\left[S^{\period,\nogame}\right].$$
\end{prop}
\Cref{prop:periodic_aggrate_price} implies that the interaction between major and minor traders reduces the strength of periodicity in the price. The proof of Proposition~\ref{prop:periodic_aggrate_price} is deferred to Section~\ref{proof:periodic_aggrate_price}.


\section{Numerical Experiments}
\label{sec:numerical_experiments}
In this section, we present numerical experiments to demonstrate the equilibrium behavior of all traders and the marketwide impact of the Major-Minor MFG. 
For simplicity, we focus on the scenario in \Cref{sec:common_inventory} where all minor traders share the same initial inventory and the Nash equilibrium of the Major-Minor MFG is deterministic, as this is sufficient to demonstrate the effect of the periodic targeting strategy.%
\footnote{To our knowledge, no numerical methods for simulating the general mean-field FBSDE \eqref{eq:fbsde_nash} that arise in our game are yet available.} 
\Cref{sec:method_numerical} describes the numerical methods for solving the Nash equilibrium, and \Cref{sec:numerical_results} reports the numerical results  under three typical targeting strategies.

\subsection{Numerical Methods for Solving the Nash Equilibrium}
\label{sec:method_numerical}
We analyze the Nash equilibrium of the Major-Minor MFG by numerically solving the associated ODEs \eqref{eq:major_trader_ODE}, \eqref{eq:periodic_decomposition_periodic} and \eqref{eq:periodic_decomposition_trend} in two steps.  

In the first step, we use the finite difference method (see, e.g., \cite{quarteroni_numerical_2007}) to solve the boundary value problem of the ODE in \eqref{eq:major_trader_ODE}, and the specific implementation is summarized in \Cref{alg:finite_diff_ODE}.
\begin{algorithm}[!ht]
\caption{Finite Difference Method for Solving the ODE in~\eqref{eq:major_trader_ODE}}
\label{alg:finite_diff_ODE}
\small
\begin{algorithmic}
\STATE \textbf{Input:} $T$, $h$, $q^{\major}_0$, $q^{\minor}_0$.
\STATE Define $n^{\mathrm{mesh}}:=T/h$, and set $z_0 = q^{\major}_0$, $z_{n^{\mathrm{mesh}}} = 0$, $w_0 = q^{\minor}_0$, $w_{n^{\mathrm{mesh}}} = 0$.
\STATE Solve the following linear equations for $(z_1,\ldots,z_{n^{\mathrm{mesh}}-1},w_1,\ldots,w_{n^{\mathrm{mesh}}-1})$:
\begin{equation}
\begin{aligned}
a_0 \frac{z_{j+1}-2z_{j}+z_{j-1}}{h^2} - \phi_0 z_{j} + \frac{\lambda}{2}\frac{w_{j+1}-w_{j}}{h} = - \phi_0 R_{jh},\quad 
j=1,2,\ldots,n^{\mathrm{mesh}}-1,\\ 
a \frac{w_{j+1}-2w_{j}+w_{j-1}}{h^2} - \phi w_{j} + \frac{\lambda_0}{2}\frac{z_{j+1}-z_{j}}{h} + \frac{\lambda}{2}\frac{w_{j+1}-w_{j}}{h} = 0,\quad 
j=1,2,\ldots,n^{\mathrm{mesh}}-1.
\end{aligned}
\label{eq:finite_major_trader_ODE}
\end{equation}
\STATE \textbf{Output:} $\hat{Q}^{\major}_{ih} = z_i$ for $i=0,1,2,\dots, n^{\mathrm{mesh}}$ and $\hat{Q}^{\minor}_{ih} = w_i$ for $i=0,1,2,\dots, n^{\mathrm{mesh}}$.
\end{algorithmic}
\end{algorithm}

\begin{algorithm}[!ht]
\caption{Algorithm for the Periodic Solution of the ODE \eqref{eq:periodic_decomposition_periodic}}
\label{alg:periodic_ODE}
\small
\begin{algorithmic}
\STATE \textbf{Input:} $T$, $T^{\period}$, $h$, $N^{\mathrm{iter}}$.
\STATE Initialize $q^{\major,\period}_0$ and $q^{\minor,\period}_0$.
\WHILE{$i < N^{\mathrm{iter}}$}
\STATE Let $z_0 = q^{\major,\period}_0$, $z_n^{\mathrm{mesh}} = q^{\major,\period}_0$, $w_0 = q^{\minor,\period}_0$, $w_n^{\mathrm{mesh}} = q^{\minor,\period}_0$.
\STATE Solve the following linear equations on $(z_1,\ldots,z_{n^{\mathrm{mesh}}-1},w_1,\ldots,w_{n^{\mathrm{mesh}}-1})$:
\begin{equation}
\begin{aligned}
a_0 \frac{z_{j+1}-2z_{j}+z_{j-1}}{h^2} - \phi_0 z_{j} + \frac{\lambda}{2}\frac{w_{j+1}-w_{j}}{h} = - \phi_0 \tilde{R}_{jh},\quad 
j=1,2,\ldots,n^{\mathrm{mesh}}-1,\\ 
a \frac{w_{j+1}-2w_{j}+w_{j-1}}{h^2} - \phi w_{j} + \frac{\lambda_0}{2}\frac{z_{j+1}-z_{j}}{h} + \frac{\lambda}{2}\frac{w_{j+1}-w_{j}}{h} = 0,\quad 
j=1,2,\ldots,n^{\mathrm{mesh}}-1.
\end{aligned}
\end{equation}
\STATE Update $q^{\major,\period}_0 \gets \frac{1}{2} \left(q^{\major,\period}_0 + z_{n^{\mathrm{mesh},\period}}\right)$ and $q^{\minor,\period}_0 \gets \frac{1}{2} \left(q^{\minor,\period}_0 + w_{n^{\mathrm{mesh},\period}}\right)$.
\ENDWHILE
\STATE \textbf{Output:} $\hat{Q}^{\major}_{ih} = z_i$ for $i=0,1,2,\dots, n^{\mathrm{mesh}}$ and $\hat{Q}^{\minor}_{ih} = w_i$ for $i=0,1,2,\dots, n^{\mathrm{mesh}}$.
\end{algorithmic}
\end{algorithm}

In the second step, we numerically decompose the Nash equilibrium into its periodic and trend components according to \Cref{thm:major_trader_decomposition}. In particular, for the periodic component, we use \Cref{alg:periodic_ODE} to find the periodic solution of the ODE \eqref{eq:periodic_decomposition_periodic}, and the key idea is to iteratively update the initial values to ensure the periodicity of the solution. For the trend component, we can  obtain the numerical solution of the ODE in~\eqref{eq:periodic_decomposition_trend} also using \Cref{alg:finite_diff_ODE}, but with the linear equations \eqref{eq:finite_major_trader_ODE} replaced by the following ones:
\begin{equation*}
\begin{aligned}
& a_0 \frac{z_{j+1}-2z_{j}+z_{j-1}}{h^2} - \phi_0 z_{j} + \frac{\lambda}{2}\frac{w_{j+1}-w_{j}}{h} = 0,\quad 
& j=1,2,\ldots,n^{\mathrm{mesh}}-1,\\ 
& a \frac{w_{j+1}-2w_{j}+w_{j-1}}{h^2} - \phi w_{j} + \frac{\lambda_0}{2}\frac{z_{j+1}-z_{j}}{h} + \frac{\lambda}{2}\frac{w_{j+1}-w_{j}}{h} =  -\frac{\lambda_0 q^{0}_{0}}{2T},\quad 
& j=1,2,\ldots,n^{\mathrm{mesh}}-1.
\end{aligned}
\end{equation*}

\subsection{Results}
\label{sec:numerical_results}
We present the numerical results of the Major-Minor MFG under three targeting strategies: (a) the cosine strategy; (b) the TWAP strategy; and (c) the VWAP strategy.

Throughout the numerical experiments, we use the following parameter set-up: $T=10$, $q^{\major}_{0}=10$, $q^{\minor}_0=0$, $\phi_0=0.1$, $\phi=0.01$, $a_0=a=0.001$, $\lambda_0 = 0.01$, and $\lambda=0.005$. Here, we set the initial inventory of the representative minor trader to be $q^{\minor}_0=0$, reflecting a scenario where minor traders have no net buying or selling incentives. In addition, the risk penalty parameter for the major trader $\phi_0$ is larger than that of the minor traders $\phi$. Lastly, the permanent impact of the major trader $\lambda_0$ exceeds that of the minor traders $\lambda$, highlighting the major trader's greater influence on the asset's price.

\subsubsection{Case I: Cosine Targeting Strategy} \label{sec:numerical_cos}
We first consider the case where the targeting strategy $R$ is a cosine strategy with $10$ periods (see \Cref{eg:cos_target}):
\begin{equation}
R_t = Q^{\cos,10,q^{\major}_0,1/(2\pi)}_t := q^{\major}_0 \left(1 - \frac{t}{10}\right) + \frac{1}{2\pi} \sin\left(2\pi t\right),\quad t \in [0,10].    
\label{eq:targeting_trading_rate_cos}
\end{equation}
\Cref{fig:targeting_strategy_cos} shows the trading rates of the cosine targeting strategy \eqref{eq:targeting_trading_rate_cos}, whose speed of selling is the slowest at the edges of each period ($t=0,1,\dots,10$) and the fastest in the middle of each period ($t=0.5,1.5,\dots,9.5$).
\begin{figure}[!ht]
\centering
\includegraphics[width=.8\textwidth]{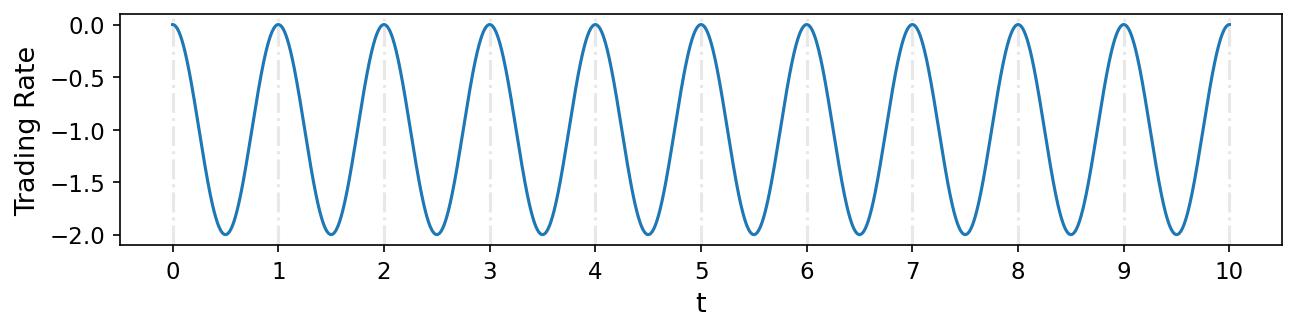}
\caption{The trading rates of the cosine targeting strategy \eqref{eq:targeting_trading_rate_cos}.}
\label{fig:targeting_strategy_cos}
\end{figure}

\paragraph{Strategy of Major and Minor Traders.}
\Cref{fig:major_trader_cos} summarizes the equilibrium trading rates of the major and minor traders in the Major-Minor MFG under the cosine targeting strategy \eqref{eq:targeting_trading_rate_cos}.
\begin{figure}[!ht]
\centering
\subfloat[Trading rates.]{
\includegraphics[width=.8\textwidth]{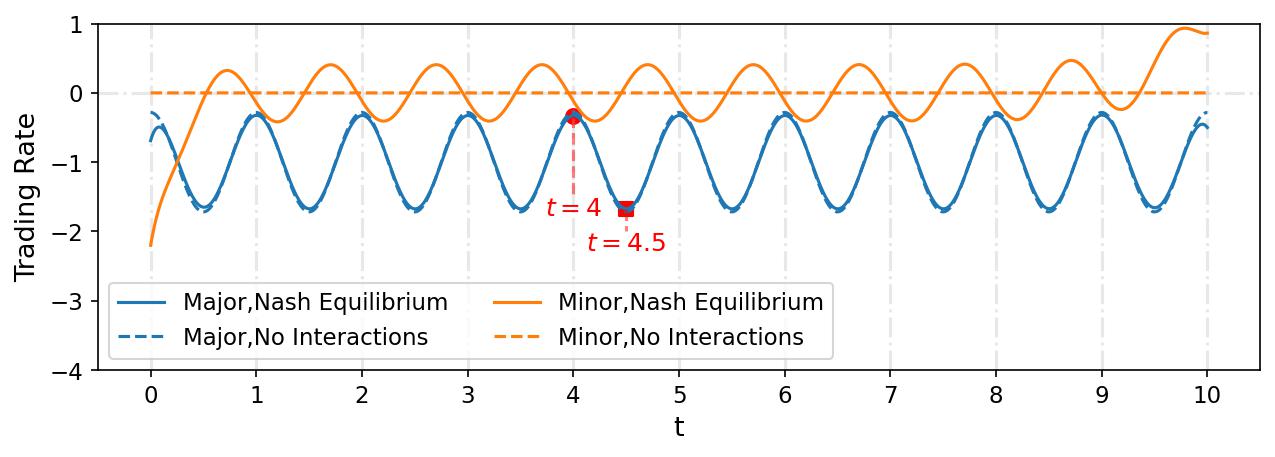}
\label{fig:trading_rate_cos}
}
\\
\subfloat[Difference between the trading rates with and without interactions.]{
\includegraphics[width=.8\textwidth]{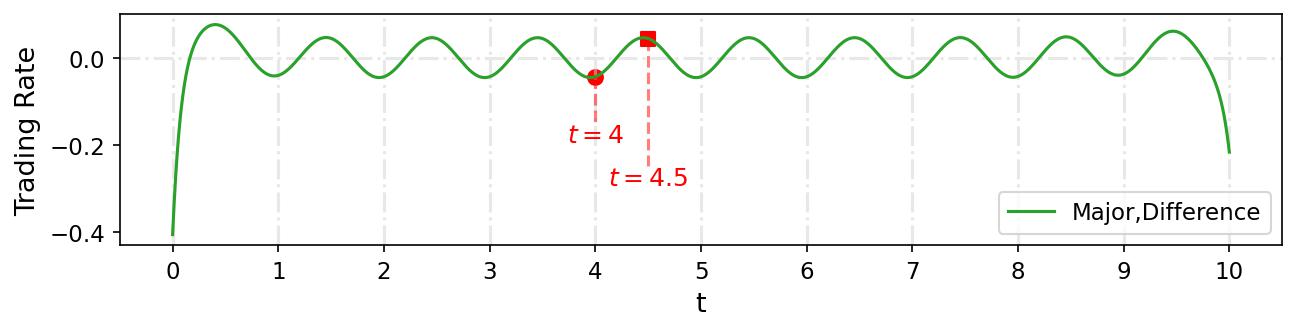}
\label{fig:trading_rate_diff_cos}
}
\\
\subfloat[Periodic components of trading rates (with interactions).]{
\includegraphics[width=.8\textwidth]{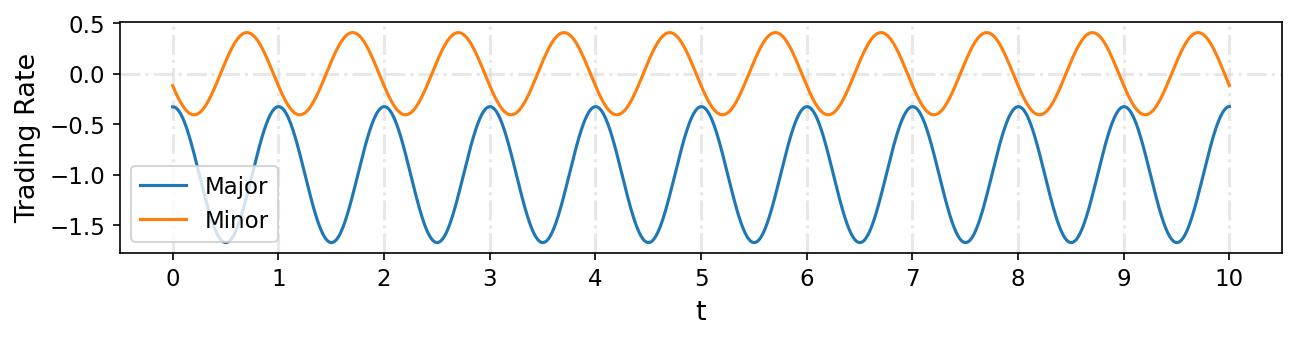}
\label{fig:trading_rate_period_cos}
}
\\
\subfloat[Trend components of trading rates  (with interactions).]{
\includegraphics[width=.8\textwidth]{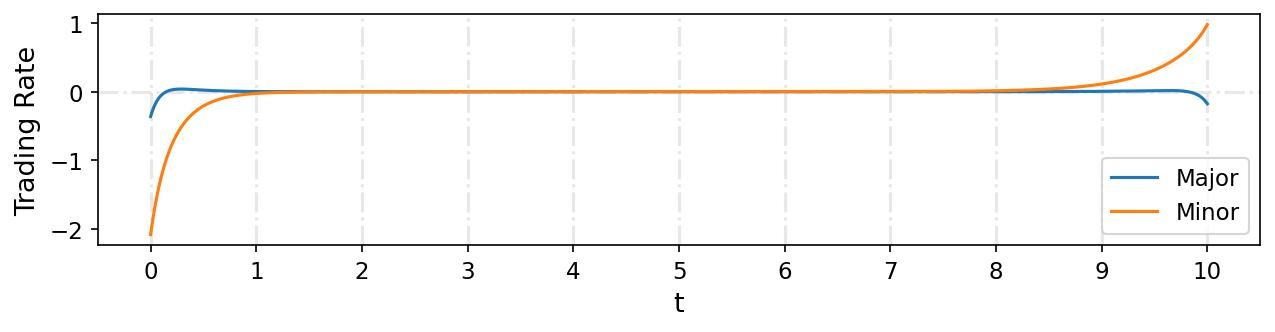}
\label{fig:trading_rate_trend_cos}
}
\caption{The Nash equilibrium strategies under the cosine targeting strategy \eqref{eq:targeting_trading_rate_cos}, compared to optimal strategies without interactions in \Cref{sec:no_game_case}.
}%
\label{fig:major_trader_cos}
\end{figure}

First, \Cref{fig:trading_rate_cos} compares the strategies of major and minor traders with and without interactions. In the absence of interactions, the major trader adopts a periodic trading strategy similar to the cosine targeting strategy in \Cref{fig:targeting_strategy_cos}, and the representative minor trader does not trade, which are both consistent with \Cref{prop:major_periodicity_no_game}. However, in Nash equilibrium, the representative minor trader also trades with periodic patterns as a result of interactions with the major trader, which reflects the equilibrium behavior in Theorem~\ref{thm:major_trader_decomposition}.

In fact, the major trader reduces the strength of her periodic behavior in response to strategic interactions from the minor traders.
\Cref{fig:trading_rate_diff_cos} highlights the difference in trading rates of the major trader with and without interactions. In Nash equilibrium, the major trader sells faster at the edges of each period ($t=4$, for example) and slower in the middle of each period ($t=4.5$, for example), compared to the case without interactions. This confirms the theoretical results in Proposition~\ref{prop:periodic_trading_rate}-2. Although the magnitude is small, this difference offsets the predictable periodic trading pattern of the major trader in \Cref{fig:trading_rate_cos}.

Moreover, the minor traders show front-running behavior in equilibrium, which is most clearly demonstrated by decomposing the Nash equilibrium strategy into the periodic component in \Cref{fig:trading_rate_period_cos} and the trend component in \Cref{fig:trading_rate_trend_cos}, based on \Cref{thm:major_trader_decomposition}. \Cref{fig:trading_rate_period_cos} shows that the periodic component of the minor traders reach its peak {\it before} the major trader, implying that the minor traders exhibit front-running behavior by taking advantage of the periodic patterns of the major trader in equilibrium. \Cref{fig:trading_rate_trend_cos} reveals that minor traders sell the asset more quickly at the beginning of the trading period, a tactic to capitalize on the major trader's liquidation, which is another form of {\it front-running}. In response, the major trader also accelerates her speed of selling early to mitigate losses from interactions with the minor traders.

\paragraph{Cost of Major and Minor Traders.} 
All major and minor traders incur lower costs in the Nash equilibrium compared to the scenario without interactions. 
To understand the different contributing factors to costs, \Cref{tab:cost_functional_cos} breaks down the cost of both parties into different components as in \eqref{eq:cost_functional_minor} and \eqref{eq:cost_functional_major}, and compares the case with and without interactions. 

\begin{table}[!ht]
\centering
\caption{Comparison of costs for the major trader \eqref{eq:cost_functional_major} and the representative minor trader \eqref{eq:cost_functional_minor}  in the Nash equilibrium  under a cosine targeting strategy \eqref{eq:targeting_trading_rate_cos} versus the case without interactions. 
}
\begin{threeparttable}
\begin{tabularx}{0.85\textwidth}{X|YY}
\toprule
      & Nash Equilibrium & No Interactions \\
\midrule
\multicolumn{1}{c|}{\textbf{Major Trader's Cost}$^1$} & \textbf{0.0130} & \textbf{0.0136} \\
$\bullet$ Expected Profit of $Q^{\major}$ & -0.0368 & -0.0126 \\
$\bullet$ Expected Profit of $R$ & -0.0253 & 0.0000 \\
$\bullet$ Risk Penalty & 0.0014 & 0.0010 \\
\midrule
\multicolumn{1}{c|}{\textbf{Minor Trader's Cost}$^2$} & \textbf{-0.0246} & \textbf{0.0000} \\
$\bullet$ Expected Profit of $Q^{\minor}$ & 0.0475 & 0.0000 \\
$\bullet$ Risk Penalty & 0.0229 & 0.0000 \\
\bottomrule
\end{tabularx}%
\begin{tablenotes}
\item[1] {\footnotesize Major Trader's Cost $= -$ Expected Profit of $Q^{\major}$ $+$ Expected Profit of $R$ $+$ Risk Penalty.}
\item[2] {\footnotesize Minor Traders' Cost $= -$ Expected Profit of $Q^{\minor}$ $+$ Risk Penalty.}
\end{tablenotes}
\end{threeparttable}
\label{tab:cost_functional_cos}%
\end{table}%

First, the interactions among the major and minor traders decrease the expected profit of the major trader but increase that of the minor traders, which is consistent with the implications from \possessivecite{brunnermeier_predatory_2005} model of predator tradings. {Importantly, the major trader's cost is reduced despite a decrease in her profit. This occurs because the expected profit of her targeting strategy decreases more significantly than that of her actual strategy.}

\paragraph{Marketwide Impact.}
Finally, we consider the aggregate market impact of the Major-Minor MFG as discussed in Section~\ref{sec:period_whole_market}. 

First, the Major-Minor MFG accelerates the aggregate trading speed and reduces the equilibrium price in the market, summarized in \Cref{fig:total_rate_price_cos}. \Cref{fig:total_trading_rate_cos} shows that the aggregate speed of trading is accelerated at the beginning of the period with interactions among major and minor traders. Moreover, the aggregate trading rate attains its maximum in each period earlier in the Nash equilibrium compared to the case without interactions.
\Cref{fig:fundamental_price_cos} shows that the equilibrium price is always lower than the price without interactions, which reflects the accelerated selling speed in the market and adversely affects the liquidation profit of the major trader.
\begin{figure}[!ht]
\centering
\subfloat[Aggregate trading rates.]{
\includegraphics[width=.8\textwidth]{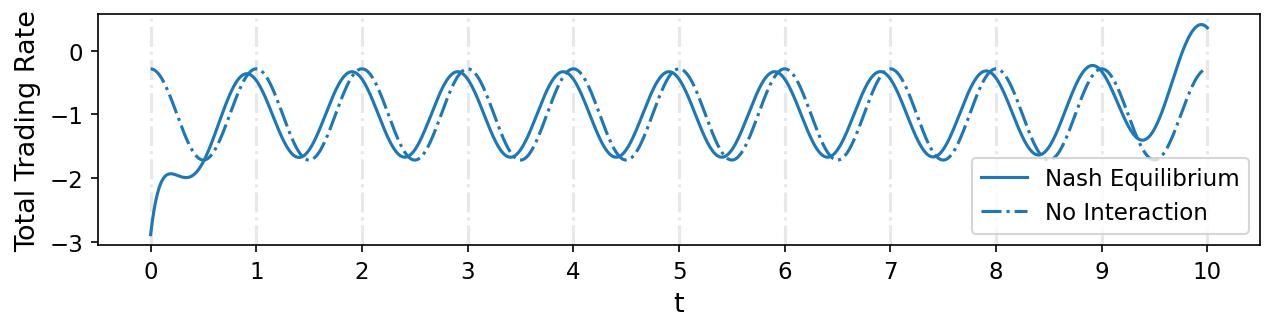}
\label{fig:total_trading_rate_cos}
}
\\
\subfloat[Price.]{
\includegraphics[width=.8\textwidth]{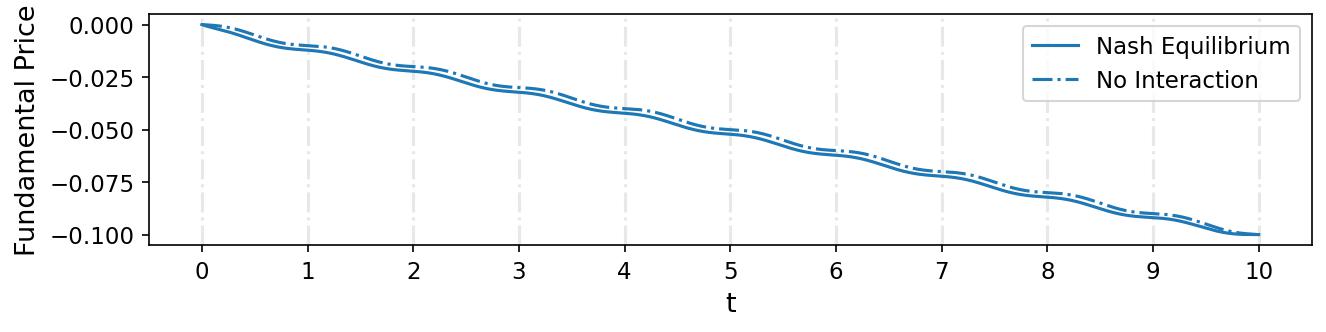}
\label{fig:fundamental_price_cos}
}
\caption{The aggregate trading rates and the price under the Nash equilibrium  with the cosine targeting strategy \eqref{eq:targeting_trading_rate_cos}.
}   
\label{fig:total_rate_price_cos}
\end{figure}

Second, the interactions among the major and minor traders reduce the strength of periodicity in not only the aggregate trading rate, but also the price, and the latter is consistent with Proposition~\ref{prop:periodic_aggrate_price}. In particular, \Cref{tab:amplitude_cos} shows that the amplitudes of the periodic components in both the aggregate trading rates and the price---as defined in \Cref{sec:period_whole_market}---decrease for all traders. In this sense, the interactions between both parties serve as a smoothing mechanism for the aggregate market.
\begin{table}[!ht]
\centering
\caption{Amplitudes of the periodic components in the aggregate trading rate and the price with the cosine targeting strategy \eqref{eq:targeting_trading_rate_cos}.}
\label{tab:amplitude_cos}%
\begin{tabular}{c|cc}
\toprule
      & Nash Equilibrium & No Interactions \\
\midrule
Aggregate Trading Rate & 0.672341 & 0.716953 \\
Price & 0.001020 & 0.001141 \\
\bottomrule
\end{tabular}%
\end{table}%

\subsubsection{Case II: TWAP Strategy} \label{sec:numerical_twap}
Next, we consider the case where the targeting strategy, $R$, is a TWAP strategy as in \Cref{eg:piecewise_const_target} with 10 periods, that is,
\begin{equation}
R_t = Q^{\mathrm{TWAP},q^{\major}_0,10}_t = \begin{cases}
\left(1 - \frac{k}{10}\right) q^{\major}_0,  & t = k,\, k=0,1,\dots,10,\\
\left(1 - \frac{2k - 1}{20}\right) q^{\major}_0,  & t \in (k-1,k), \, k=1,2,\dots,10.\\
\end{cases}
\label{eq:targeting_trading_rate_twap}
\end{equation}
The TWAP targeting strategy only trades at timestamps $t=0,1,2,\dots,10$, and \Cref{fig:targeting_strategy_twap} shows its inventory process. The trading rates of the TWAP strategy are not well-defined due to the jumps in its inventory process.
\begin{figure}[!ht]
\centering
\includegraphics[width=.8\textwidth]{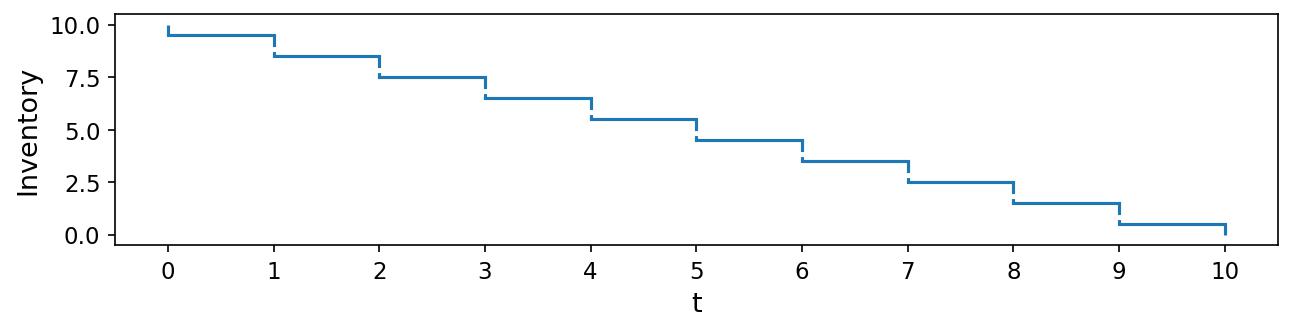}
\caption{The inventory of the TWAP targeting strategy \eqref{eq:targeting_trading_rate_twap}.}
\label{fig:targeting_strategy_twap}
\end{figure}

\paragraph{Strategy of Major and Minor Traders.}
Following the analysis in Section~\ref{sec:numerical_cos}, \Cref{fig:major_trader_twap} summarizes the trading rates of the major trader and the representative minor trader under the TWAP targeting strategy \eqref{eq:targeting_trading_rate_twap} with and without interactions. It also presents the periodic and trend components of the Nash equilibrium. Overall, we observe similar behaviors to those in \Cref{sec:numerical_cos}.

\begin{figure}[!ht]
\centering
\subfloat[Trading rates.]{
\includegraphics[width=.8\textwidth]{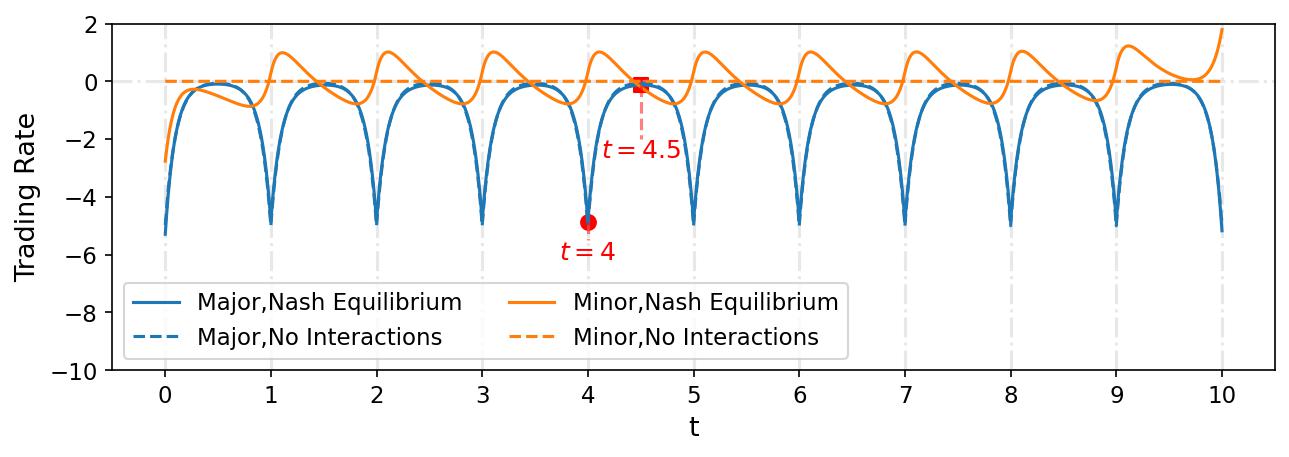}
\label{fig:trading_rate_twap}
}
\\
\subfloat[Difference between the trading rates with and without interactions.]{
\includegraphics[width=.8\textwidth]{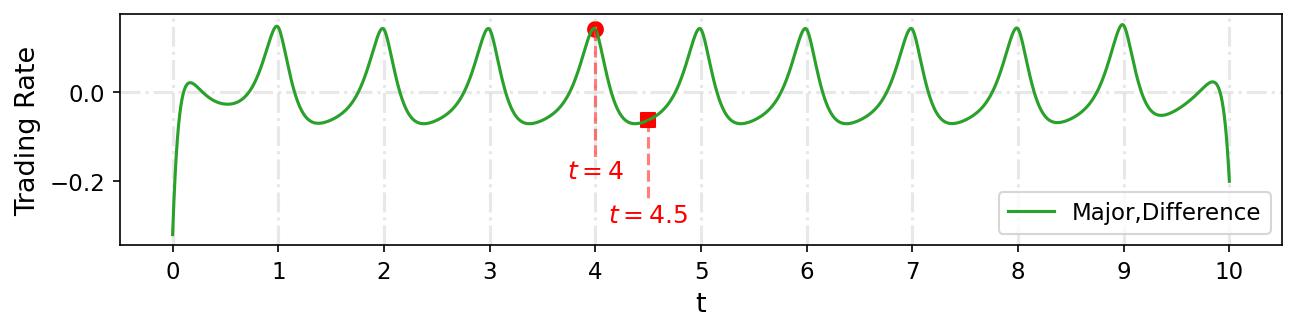}
\label{fig:trading_rate_diff_twap}
}
\\
\subfloat[Periodic components of trading rates (with interactions).]{
\includegraphics[width=.8\textwidth]{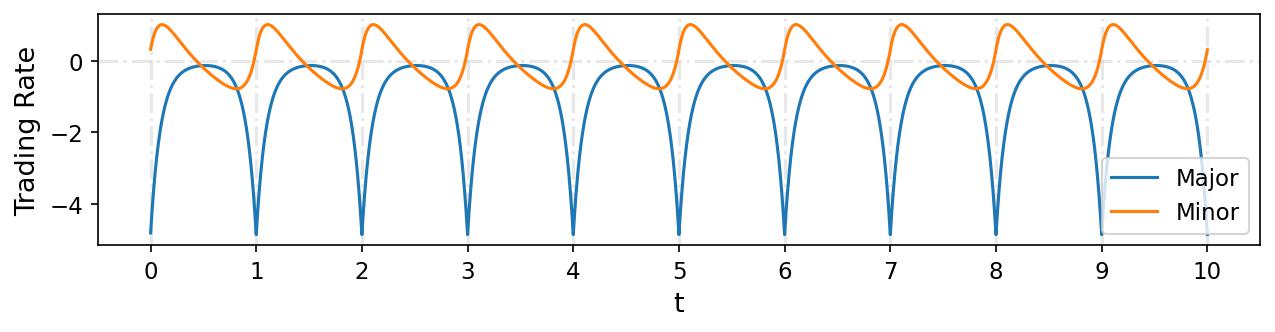}
\label{fig:trading_rate_period_twap}
}
\\
\subfloat[Trend components of trading rates (with interactions).]{
\includegraphics[width=.8\textwidth]{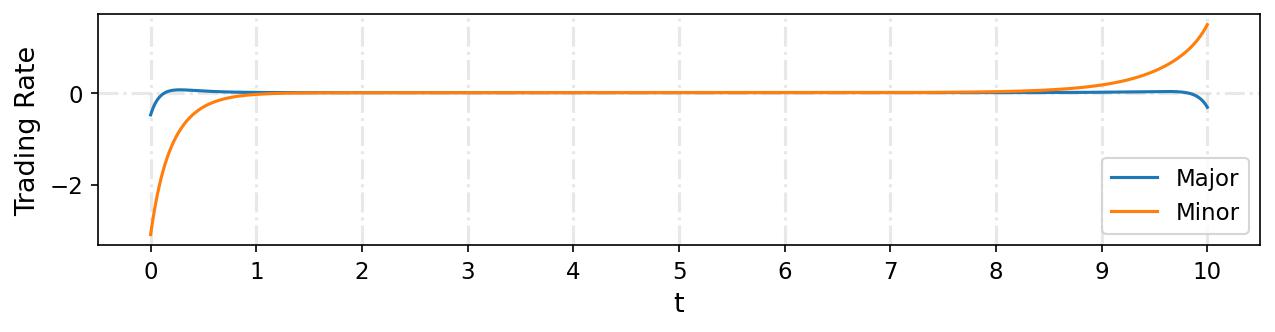}
\label{fig:trading_rate_trend_twap}
}
\caption{The Nash equilibrium strategies with the TWAP targeting strategy \eqref{eq:targeting_trading_rate_twap}, compared to optimal strategies without interactions in \Cref{sec:no_game_case}.
}
\label{fig:major_trader_twap}
\end{figure}

First, the minor traders show front-running behaviors in both the periodic and trend components in equilibrium. As shown in \Cref{fig:trading_rate_period_twap}, the periodic component of the representative minor trader's trading rate reaches its peak before that of the major trader. Moreover, \Cref{fig:trading_rate_trend_twap} shows that both parties accelerate the selling speed at the beginning of the trading interval, which is similar to the patterns observed under the cosine targeting strategy in \Cref{fig:trading_rate_trend_cos}.

In response to strategic actions by minor traders, the major trader reduces the strength of the periodicity in her strategy, as shown in Figures \ref{fig:trading_rate_twap}--\ref{fig:trading_rate_diff_twap}. In particular, \Cref{fig:trading_rate_diff_twap} shows that, compared to the case without interactions, the major trader sells slower at the edges of each period (for example, at $t=4$) and faster in the middle of each period (for example, at $t=4.5$). This can be interpreted as the major trader's response to {\it reduce her predictable periodic trading patterns} to protect herself from the minor traders.

\paragraph{Cost of Major and Minor Traders.} 
Similar to the results in \Cref{sec:numerical_cos}, both parties reduce their costs in the Major-Minor MFG under the TWAP targeting strategy, compared to the case without interactions. As shown in \Cref{tab:cost_functional_twap}, the major trader experiences a reduction in the expected profit due to strategic interactions with the minor traders, who, in contrast, gain from these interactions. However, after factoring in the changes in the cost of the targeting strategy, the overall costs of both parties are reduced in the Major-Minor MFG, compared to the case without interactions.

\begin{table}[!ht]
\centering
\caption{Comparison of costs for the major trader \eqref{eq:cost_functional_major} and the representative minor trader \eqref{eq:cost_functional_minor} in the Nash equilibrium under the TWAP targeting strategy \eqref{eq:targeting_trading_rate_twap} versus the case without interactions. 
}
\begin{threeparttable}
\begin{tabularx}{0.85\textwidth}{X|YY}
    \toprule
          & Nash Equilibrium & No Interactions \\
    \midrule
    \multicolumn{1}{c|}{\textbf{Major Trader's Cost}$^1$} & \textbf{0.0477} & \textbf{0.0500} \\
    $\bullet$ Expected Profit of $Q^{\major}$ & -0.0504 & -0.0250 \\
    $\bullet$ Expected Profit of $R$ & -0.0290 & 0.0000 \\
    $\bullet$ Risk Penalty & 0.0264 & 0.0250 \\
   \midrule
    \multicolumn{1}{c|}{\textbf{Minor Trader's Cost}$^2$} & \textbf{-0.0266} & \textbf{0.0000} \\
    $\bullet$ Expected Profit of $Q^{\minor}$ & 0.0491 & 0.0000 \\
    $\bullet$ Risk Penalty & 0.0225 & 0.0000 \\
    \bottomrule
\end{tabularx}%
\begin{tablenotes}
\item[1] {\footnotesize Major Trader's Cost $= -$ Expected Profit of $Q^{\major}$ $+$ Expected Profit of $R$ $+$ Risk Penalty.}
\item[2] {\footnotesize Minor Trader's Cost $= -$ Expected Profit of $Q^{\minor}$ $+$ Risk Penalty.}
\end{tablenotes}
\end{threeparttable}
\label{tab:cost_functional_twap}%
\end{table}%

\paragraph{Marketwide Impact.}
Finally, the Major-Minor MFG increases the total trading volume at the beginning of the trading interval, and reduces the price in the market, which is consistent with the findings in \Cref{sec:numerical_cos}. In particular, \Cref{fig:total_trading_rate_twap} shows that the interactions between both parties increase market-wide trading activities at the start of trading. Moreover, the highest trading rate is achieved earlier in each period under the Nash equilibrium than in the scenario without interactions. 
\Cref{fig:fundamental_price_twap} indicates that the equilibrium price remains consistently lower than the price without interactions, which negatively affects the liquidation profit of the major trader.

\begin{figure}[!ht]
\centering
\subfloat[Aggregate trading rates.]{
\includegraphics[width=.8\textwidth]{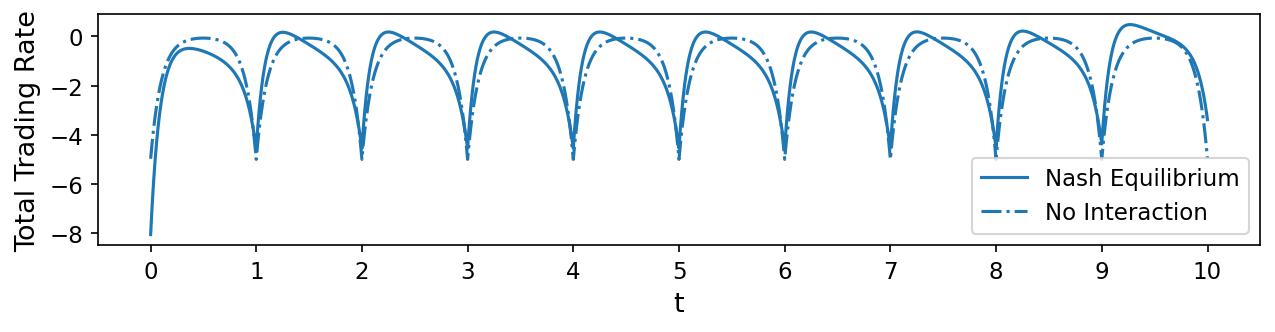}
\label{fig:total_trading_rate_twap}
}
\\
\subfloat[Price.]{
\includegraphics[width=.8\textwidth]{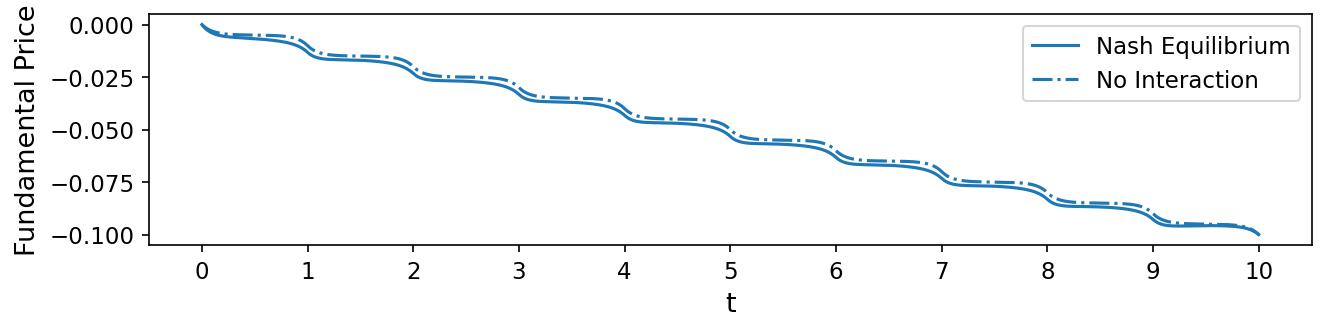}
\label{fig:fundamental_price_twap}
}
\caption{The aggregate trading rates and the price under the Nash equilibrium with the TWAP targeting strategy \eqref{eq:targeting_trading_rate_twap}. 
}%
\label{fig:total_rate_price_twap}%
\end{figure}

Furthermore, \Cref{tab:amplitude_twap} calculates the amplitudes of the periodic components in both the aggregate trading rate and the price. The interactions among the major and minor traders reduce the periodicity in both quantities.

\begin{table}[!ht]
\centering
\caption{Amplitudes of the periodic components in the aggregate trading rate and the price, with the TWAP targeting strategy \eqref{eq:targeting_trading_rate_twap}.}
\label{tab:amplitude_twap}%
\begin{tabular}{c|cc}
\toprule
      & Nash Equilibrium & No Interactions \\
\midrule
Aggregate Trading Rate & 2.364551 & 2.466589 \\
Price & 0.230056 & 0.239414 \\
\bottomrule
\end{tabular}%
\end{table}%

\subsubsection{Case III: VWAP Strategy}
\label{sec:numerical_vwap}
Finally, we consider the following VWAP targeting strategy with a periodic component whose trading rates are described by:
\begin{equation}
v_t = \begin{cases}
- \frac{15}{370} (t-7)^2 - 0.5 + 0.75 \cos(4\pi t), & t\in[0,3),\\     
- \frac{15}{370} (t-7)^2 - 0.5 + 0.5 \cos(2\pi t), & t\in[3,10].\\ 
\end{cases}
\label{eq:targeting_trading_rate_vwap}
\end{equation}
\Cref{fig:targeting_strategy_vwap} visualizes the trading rates of the VWAP strategy \eqref{eq:targeting_trading_rate_vwap}. The trend of the VWAP strategy forms a U-shaped curve, achieving its peak at $t=7$. Moreover, the trading rate exhibits a periodicity with period $1/2$ for $t\in[0,3)$ and a periodicity with period $1$ for $t\in[3,10]$. 

It has been documented that the intraday trading volume, beyond the well-known U-shaped pattern \citep{admati1988theory,jain1988dependence}, contains periodic patterns in frequencies of hours, minutes, and even seconds \citep{heston2010intraday,hasbrouck2013low,wu2022spectral}.
In particular, the VWAP targeting strategy \eqref{eq:targeting_trading_rate_vwap} mimics the main characteristics of periodic intraday trading volumes in the stock market recently documented in \cite{wu2022spectral}.

\begin{figure}[!ht]
\centering
\includegraphics[width=.8\textwidth]{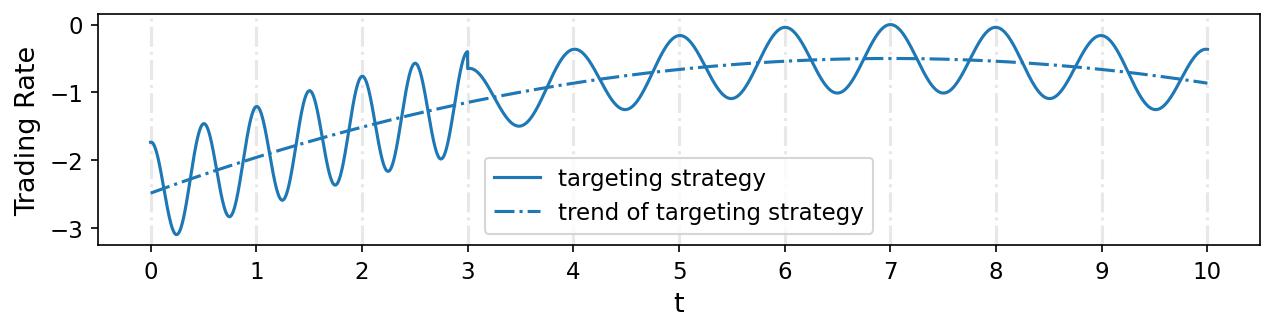}
\caption{The trading rates of the VWAP strategy \eqref{eq:targeting_trading_rate_vwap}.}
\label{fig:targeting_strategy_vwap}
\end{figure}

\paragraph{Strategy of Major and Minor Traders.}
Following the numerical analysis in Sections~\ref{sec:numerical_cos}--\ref{sec:numerical_twap}, \Cref{fig:major_trader_vwap} shows the strategies of major and minor traders with and without interactions under the VWAP targeting strategy \eqref{eq:targeting_trading_rate_vwap}. Although the VWAP strategy \eqref{eq:targeting_trading_rate_vwap} does not follow the exact form of periodic strategies in \Cref{defn:periodic_strategy}, the behaviors of both the major and the representative minor trader are still similar to those in Sections~\ref{sec:numerical_cos}--\ref{sec:numerical_twap}.

\begin{figure}[!ht]
\centering
\subfloat[Trading rates.]{
\includegraphics[width=.8\textwidth]{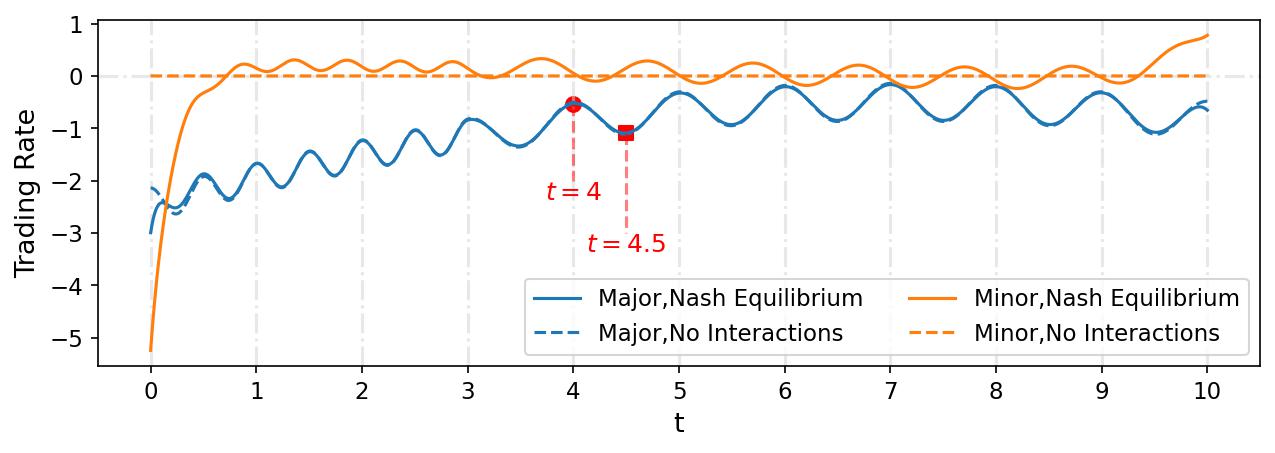}
\label{fig:trading_rate_vwap}
}
\\
\subfloat[Difference between the trading rates with and without interactions.]{
\includegraphics[width=.8\textwidth]{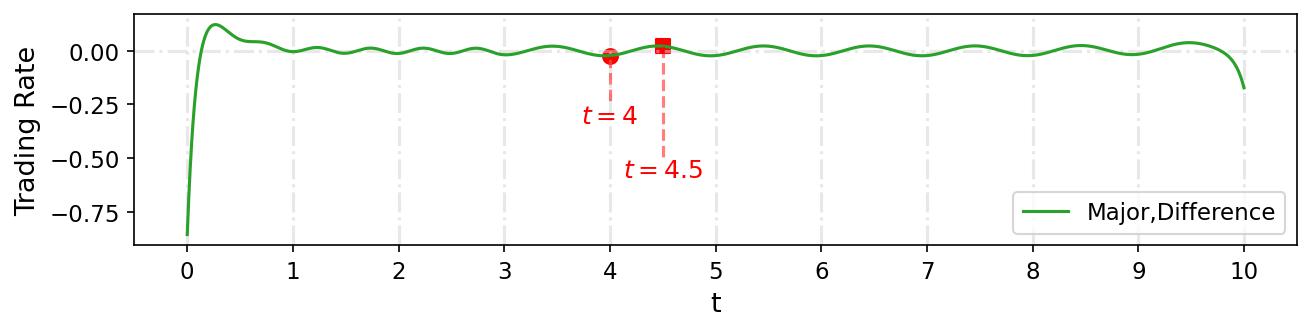}
\label{fig:trading_rate_diff_vwap}
}
\caption{The Nash equilibrium strategies under the VWAP targeting strategy \eqref{eq:targeting_trading_rate_vwap}, compared to the optimal strategies without interactions in \Cref{sec:no_game_case}.
}
\label{fig:major_trader_vwap}
\end{figure}

First, the minor traders show front-running behavior in the beginning and mirrors the periodic behavior of the major trader.
\Cref{fig:trading_rate_vwap} shows that, without interactions, the major trader adopts a strategy with a U-shaped trend and a varying periodic component in $[0,3]$ and $[3,10]$, which is similar to the VWAP targeting strategy. In this case, the representative minor trader does not trade at all. When interactions are taken into account in the equilibrium, the representative minor trader accelerates his selling rate in the beginning, which is a front-running behavior, and trades periodically against the periodic behavior of the major trader. 

Moreover, the major trader again reduces the periodicity in her strategy in the equilibrium in response to the behavior of the minor traders. \Cref{fig:trading_rate_diff_vwap} shows that, compared to the scenario without interactions, the major trader tends to sell faster at the start and the end of each period (e.g., at $t=4$) and slower in the middle of each period (e.g., at $t=4.5$). These findings extend the result of Proposition~\ref{prop:periodic_trading_rate} to the case with the VWAP targeting strategy here.

\paragraph{Cost of Major and Minor Traders.}
Similar to the results in Sections \ref{sec:numerical_cos}--\ref{sec:numerical_twap}, both parties reduce their costs under the VWAP targeting strategy, as shown in \Cref{tab:cost_functional_twap}.

\begin{table}[!ht]
\centering
\caption{Comparison of costs for the major trader \eqref{eq:cost_functional_major} and the representative minor trader \eqref{eq:cost_functional_minor} in the Nash equilibrium under the VWAP targeting strategy \eqref{eq:targeting_trading_rate_vwap} versus the case without interactions. 
}
\begin{threeparttable}
\begin{tabularx}{0.85\textwidth}{X|YY}
\toprule
      & Nash Equilibrium & No Interactions \\
\midrule
\multicolumn{1}{c|}{\textbf{Major Trader's Cost}$^1$} & \textbf{0.0138} & \textbf{0.0140} \\
$\bullet$ Expected Profit of $Q^{\major}$ & -0.0412 & -0.0137 \\
$\bullet$ Expected Profit of $R$ & -0.0279 & 0.0000 \\
$\bullet$ Risk Penalty & 0.0005 & 0.0004 \\
\midrule
\multicolumn{1}{c|}{\textbf{Minor Trader's Cost}$^2$} & \textbf{-0.0276} & \textbf{0.0000} \\
$\bullet$ Expected Profit of $Q^{\minor}$ & 0.0521 & 0.0000 \\
$\bullet$ Risk Penalty & 0.0246 & 0.0000 \\
\bottomrule
\end{tabularx}%
\begin{tablenotes}
\item[1] {\footnotesize Major Trader's Cost $= -$ Expected Profit of $Q^{\major}$ $+$ Expected Profit of $R$ $+$ Risk Penalty.}
\item[2] {\footnotesize Minor Trader's Cost $= -$ Expected Profit of $Q^{\minor}$ $+$ Risk Penalty.}
\end{tablenotes}
\end{threeparttable}
\label{tab:cost_functional_vwap}%
\end{table}%

\paragraph{Marketwide Impact.}
\Cref{fig:total_rate_price_vwap} shows the marketwide impact on the aggregate trading rate and the price. In particular, \Cref{fig:total_trading_rate_vwap} shows that the interactions among the the major and minor traders accelerate the aggregate trading speed in the market at the start of the day. \Cref{fig:fundamental_price_vwap} shows that the price in the Nash equilibrium is always lower than in the case without interactions, which adversely affects the liquidation profit of the major trader.

\begin{figure}[!ht]
\centering
\subfloat[Aggregate trading rates.]{
\includegraphics[width=.8\textwidth]{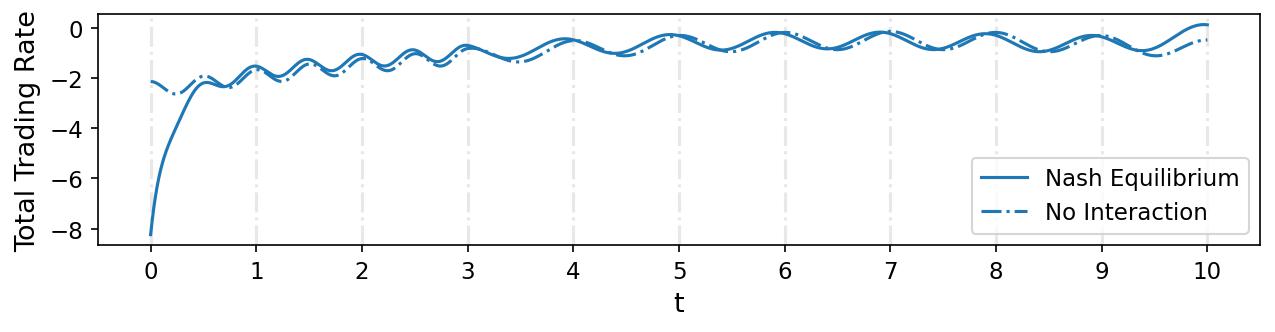}
\label{fig:total_trading_rate_vwap}
}
\\
\subfloat[Price.]{
\includegraphics[width=.8\textwidth]{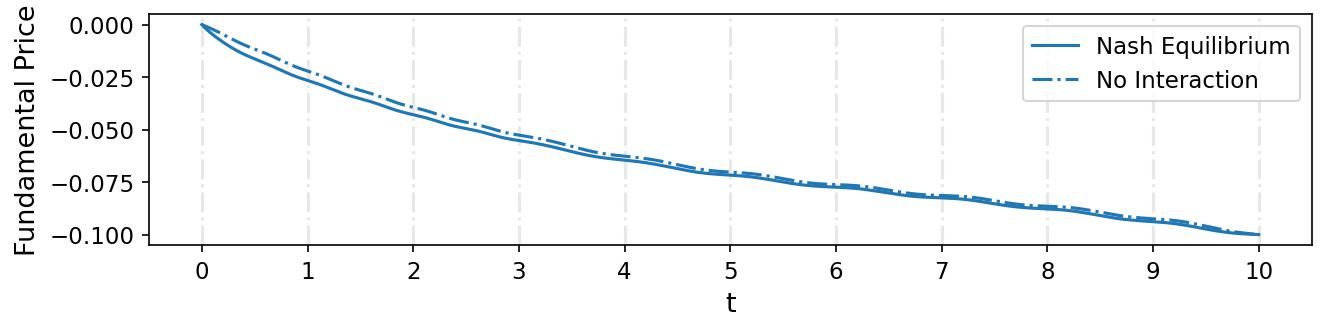}
\label{fig:fundamental_price_vwap}
}
\caption{The aggregate trading rates and the price under the Nash equilibrium with the VWAP targeting strategy \eqref{eq:targeting_trading_rate_vwap}. 
}%
\label{fig:total_rate_price_vwap}%
\end{figure}

Finally, we show that, under the VWAP targeting strategy, the interactions among the major and minor traders reduce the strength of periodicity in both the aggregate trading rate and the price. 
In this case, we cannot extract the periodic components in the aggregate trading rate and the price using Theorem~\ref{thm:major_trader_decomposition}. This is because the VWAP strategy \eqref{eq:targeting_trading_rate_vwap} does not follow the exact form in \Cref{defn:periodic_strategy}. Hence, we numerically estimate these periodic components using the spectral model in \cite{wu2022spectral} and measure the periodicity strength through the amplitudes of the periodic components.
Taking the price as an example, we assume that it has the following spectral representation:
\begin{equation}
S_t = m_t + \sum_{k=1}^{n} A_k \cos\left(\frac{2\pi k}{T} t\right) + \epsilon_t, \quad t\in[0,T],
\label{eq:spectral_model}
\end{equation}
where $m_t$ is the trend term, $A_k \cos\left(\frac{2\pi k}{T} t\right)$ is the periodic component with period $T/k$ for $k=1,2,\dots,n$, and $\epsilon_t$ is the noise term. We adopt the algorithm in \citet[Section 4.2]{wu2022spectral} to estimate the amplitude of the periodic components $\{|A_k|\}_{k=1}^{n}$ with $n = T/(2h) = 5000$.

\Cref{fig:total_rate_price_vwap_amplitude} shows the estimates of $|A_k|$ for $k=1,2,\dots,50$ in the spectral models for both the aggregate trading rates and the price. 
{The strongest periodic components occur at periods $1$ and $1/2$ (that is, $k=10$ and $k=20$). In the Nash equilibrium, the amplitudes of these components are smaller than those observed without interactions, for both the aggregate trading rate and the price.} 
These results suggest that the interactions between major and minor traders reduce the strength of periodicity in the aggregate trading rate and the price, which are consistent with our findings in Sections \ref{sec:numerical_cos}--\ref{sec:numerical_twap}.

\begin{figure}[!ht]
\centering
\subfloat[Aggregate trading rates.]{
\includegraphics[width=.47\textwidth]{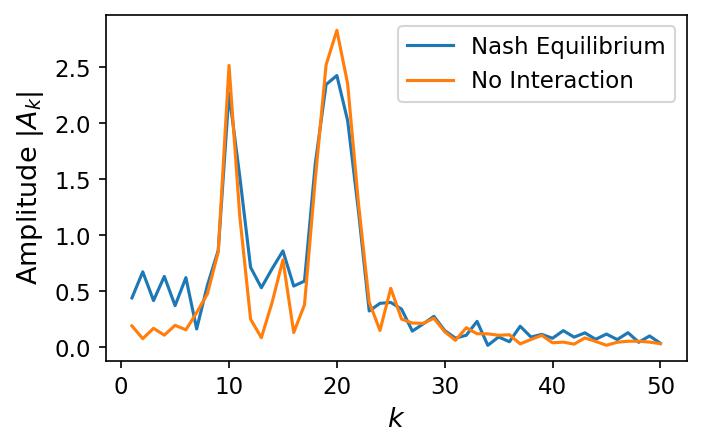}
\label{fig:total_trading_rate_vwap_amplitude}
}
\subfloat[Price.]{
\includegraphics[width=.47\textwidth]{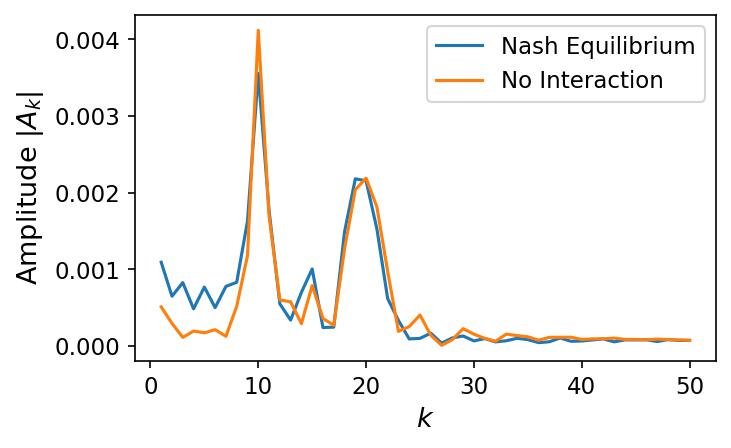}
\label{fig:fundamental_price_vwap_amplitude}
}
\caption{The frequency analysis of aggregate trading rates and the price under the Nash equilibrium under the VWAP targeting strategy \eqref{eq:targeting_trading_rate_vwap}. 
}%
\label{fig:total_rate_price_vwap_amplitude}%
\end{figure}

\subsubsection{Summary}
We summarize the key observations from our numerical experiments. 
From the investors' perspective, in equilibrium, the major trader follows the targeting strategy closely, while minor traders mimic the periodic patterns of the major trader in order to profit from it. Minor traders front-run the major trader in both the periodic and trend components of her strategy. In reaction to minor traders' behavior, the major trader reduces the strength of her periodic activities. These strategic interactions decrease (increase) the major (representative minor) trader's expected profit, which is consistent with the implications from \possessivecite{brunnermeier_predatory_2005} model of predator trading. However, the aggregate costs of both parties are reduced thanks to the decreased expected profit of the targeting strategy in equilibrium.

From a market perspective, the Major-Minor MFG increases the aggregate trading volume, accelerates the aggregate trading speed, and lowers the price in the market. Furthermore, the interactions between both parties decrease the strength of periodicity in both the aggregate trading rate and the price. 

In general, the strategic interactions among the major and minor traders not only reduce individual costs but also enhance market stability by mitigating periodic fluctuations in both the trading rates and the price. 

\section{Proofs}
\subsection{Proof of Proposition~\texorpdfstring{\ref{prop:terminal_cash_detail}}{...}}
\label{sec:proof_terminal_cash_detail}
The terminal cash of the representative minor trader can is derived from \eqref{eq:terminal_cash_minor}:
\begin{equation*}
\begin{aligned}
X^{\minor}_{T} &= - \int_0^T \bar{S}^{\minor}_{t} v^{\minor}_{t} \dif t = - \int_0^T \left(S_t + a v^{\minor}_{t}\right) v^{\minor}_{t} \dif t\\
&= - \int_0^T S_t \dif Q^{\major}_{t} - a \int_0^T \left(v^{\minor}_{t}\right)^2 \dif t = \int_0^T Q^{i}_{t} \dif F_t - a \int_0^T \left(v^{\minor}_{t}\right)^2 \dif t \\
&= \int_0^T Q^{\minor}_{t} \left(\lambda_0 v^{\major}_{t}+\lambda \bar{v}^{\minor}_t\right) \dif t  - a \int_0^T \left(v^{\minor}_{t}\right)^2 \dif t + \sigma \int_0^T Q^{\minor}_{t}  \dif W^{0}_{t}.
\end{aligned}
\end{equation*}
Similarly, the terminal cash of the major trader follows from \eqref{eq:terminal_cash_major}:
\begin{equation*}
\begin{aligned}
X^{\major}_{T} &= \int_0^T Q^{\major}_{t} \left(\lambda_0 v^{\major}_{t} + \lambda \bar{v}^{\minor}_t\right) \dif t  - a_0 \int_0^T \left(v^{\major}_{t}\right)^2 \dif t + \sigma \int_0^T Q^{\major}_{t}  \dif W^{0}_{t} \\
&=  - \frac{\lambda_0}{2}\left(q^{\major}_0\right)^{2} + \lambda \int_0^T Q^{\major}_{t}  \bar{v}^{\minor}_t \dif t  - a_0 \int_0^T \left(v^{\major}_{t}\right)^2 \dif t + \sigma \int_0^T Q^{\major}_{t}  \dif W^{0}_{t}.
\end{aligned}
\end{equation*}

\subsection{Proof of Theorem~\texorpdfstring{\ref{thm:nash_fbsde_unique}}{1}}
\label{proof:nash_fbsde_unique}
We adopt the fixed-point argument inspired by \cite{carmona_probabilistic_2018} and \cite{fu_mean_2021} to establish Theorem~\ref{thm:nash_fbsde_unique}. 

To decouple the FBSDE \eqref{eq:fbsde_nash} and handle the subtle terminal condition issue (as explained below \eqref{eq:fbsde_minor_nash_origin}), we consider the following deterministic functions on $[0,T)$:
\begin{equation}
A^{\major}_t = \frac{2\sqrt{\phi_0 a_0}}{\tanh(\sqrt{\phi_0/a_0}(T-t))} ,\quad 
A^{\minor}_t = \frac{2\sqrt{\phi a}}{\tanh(\sqrt{\phi/a}(T-t))},
\label{eq:linear_ansatz_A}
\end{equation}
which satisfy, respectively,
\begin{equation*}
\begin{aligned}
- \dif A^{\major}_t &=  \left[ - \frac{1}{2 a_0} (A^{\major}_t)^{2} + 2\phi_0 \right] \dif t ,\quad \lim_{t\nearrow T} A^{\major}_t = +\infty,\\
- \dif A^{\minor}_t &= \left[ - \frac{1}{2 a} (A^{\minor}_t)^{2} + 2\phi \right] \dif t ,\quad \lim_{t\nearrow T} A^{\minor}_t = +\infty.
\end{aligned}
\end{equation*}
With $A^{\major}$ and $A^{\minor}$ in \eqref{eq:linear_ansatz_A}, one can show that the MV-FBSDE \eqref{eq:fbsde_nash} has the same solutions as the following system of equations:
\begin{equation}
\left\{
\begin{aligned}
\dif Q^{\major}_{t} & = \frac{P^{\major}_{t}}{2 a_{0}} \dif t, \\
\dif Q^{\minor}_{t} & = \frac{P^{\minor}_{t}}{2a} \dif t, \\
- \dif P^{\major}_{t} & = \left(\frac{\lambda}{2a}\E\left[P^{\minor}_{t}\middle|\F^{\major}_t \right] - 2\phi_0 Q^{\major}_{t} + 2 \phi_0 R_t\right) \dif t - Z^{\major}_t \dif W^0_t, \\
- \dif P^{\minor}_{t} & = \left(\frac{\lambda_0}{2 a_{0}} P^{\major}_{t} + \frac{\lambda}{2a} \E\left[P^{\minor}_{t}\middle|\F^{\major}_t \right] - 2\phi Q^{\minor}_{t}\right) \dif t - Z^{\minor}_t \dif \tilde{W}_t, \\
- \dif B^{\major}_t &= \left(\frac{\lambda}{2a} \E\left[P^{\minor}_{t}\middle|\F^{\major}_t \right] - \frac{1}{2 a_0} A^{\major}_t B^{\major}_t + 2\phi_0  R_t\right) \dif t - Z^{\major}_t \dif W^0_t,\\
- \dif B^{\minor}_t &= \left( \frac{\lambda_0}{2 a_{0}}P^{\major}_{t} +  \frac{\lambda}{2a} \E\left[P^{\minor}_{t}\middle|\F^{\major}_t \right] - \frac{1}{2 a} A^{\minor}_t B^{\minor}_t \right) \dif t - Z^{\minor}_t \dif \tilde{W}_t,\\
P^{\major}_{t} &= - A^{\major}_t Q^{\major}_{t} + B^{\major}_t, \quad P^{\minor}_{t} = - A^{\minor}_t Q^{\minor}_{t} + B^{\minor}_t,\\
Q^{\major}_{0} &= q^{\major}_{0}, \, Q^{\major}_{T} = 0,\, Q^{\minor}_{0} = \mathcal{Q}^{\minor}_{0}, \, Q^{\minor}_{T} = 0.
\end{aligned}\right.
\label{eq:fbsde_nash_expand}
\end{equation}
Hence, it suffices to prove the existence and uniqueness of the solution to  the FBSDE \eqref{eq:fbsde_nash_expand}. To show this, we prove the following general result.
\begin{thm}
\label{thm:fbsde_unique_continuation}
Given $g^{\major}_t,g^{\minor}_t \in L_{\F^{\major}}^2\left([0, T] \times \Omega ; \R\right)$, for any constant $\mathfrak{p}\in[0,1]$, the FBSDE
\begin{equation}
\left\{
\begin{aligned}
\dif Q^{\major}_{t} & = \frac{P^{\major}_{t}}{2 a_{0}} \dif t, \\
\dif Q^{\minor}_{t} & = \frac{P^{\minor}_{t}}{2a} \dif t, \\
- \dif P^{\major}_{t} & = \left(\mathfrak{p} \frac{\lambda}{2a}\E\left[P^{\minor}_{t}\middle|\F^{\major}_t \right] - 2\phi_0 Q^{\major}_{t} + g^{\major}_t\right) \dif t - Z^{\major}_t \dif W^0_t, \\
- \dif P^{\minor}_{t} & = \left(\mathfrak{p} \frac{\lambda_0}{2 a_{0}} P^{\major}_{t} + \mathfrak{p} \frac{\lambda}{2a} \E\left[P^{\minor}_{t}\middle|\F^{\major}_t \right] - 2\phi Q^{\minor}_{t} + g^{\minor}_t \right) \dif t - Z^{\minor}_t \dif \tilde{W}_t, \\
- \dif B^{\major}_t &= \left(\frac{\lambda}{2a} \E\left[P^{\minor}_{t}\middle|\F^{\major}_t \right] - \frac{1}{2 a_0} A^{\major}_t B^{\major}_t + g^{\major}_t \right) \dif t - Z^{\major}_t \dif W^0_t,\\
- \dif B^{\minor}_t &= \left( \frac{\lambda_0}{2 a_{0}}P^{\major}_{t} +  \frac{\lambda}{2a} \E\left[P^{\minor}_{t}\middle|\F^{\major}_t \right] - \frac{1}{2 a} A^{\minor}_t B^{\minor}_t + g^{\minor}_t \right) \dif t - Z^{\minor}_t \dif \tilde{W}_t,\\
P^{\major}_{t} &= - A^{\major}_t Q^{\major}_{t} + B^{\major}_t, \quad P^{\minor}_{t} = - A^{\minor}_t Q^{\minor}_{t} + B^{\minor}_t,\\
Q^{\major}_{0} &= q^{\major}_{0}, \, Q^{\major}_{T} = 0,\, Q^{\minor}_{0} = \mathcal{Q}^{\minor}_{0}, \, Q^{\minor}_{T} = 0
\end{aligned}\right.
\label{eqn:fbsde_continuation}
\end{equation}
has a unique solution $(Q^{\major},Q^{\minor},P^{\major},P^{\minor},B^{\major},B^{\minor},Z^{\major},Z^{\minor})$ in the space of
\begin{equation}
\begin{aligned}
\mathcal{S} &:= \mathcal{H}^{1}_{\F^{\major}} \times \mathcal{H}^{1}_{\F^{\minor}} \times D_{\F^{\major}}^2([0, T] \times \Omega; \R) \times D_{\F^{\minor}}^2([0, T] \times \Omega; \R) \times \mathcal{H}^{\gamma}_{\F^{\major}} \times \mathcal{H}^{\gamma}_{\F^{\minor}} \\
&\qquad\qquad \times L_{\F^{\major}}^2\left([0, T-] \times \Omega ; \R\right) \times L_{\F^{\minor}}^2\left([0, T-] \times \Omega ; \R^m\right),    
\end{aligned} 
\end{equation}
where $\gamma\in(0,1/2)$ is a constant.    
\end{thm}
\noindent
Given the condition $\int_0^T (R_t)^2 dt < +\infty$ in Assumption~\ref{assumption:fbsde}, we can obtain the existence and uniqueness of the FBSDE system \eqref{eq:fbsde_nash_expand} when taking $\mathfrak{p}=1$, $g^{\major}_t = 2\phi R_t $, and $g^{\minor}_t = 0$ in Theorem~\ref{thm:fbsde_unique_continuation}, which leads to the result in Theorem~\ref{thm:nash_fbsde_unique}.

We first provide some inequalities on $A^0_t$ and $A_t$ to control the norms of the solution.
\begin{lemma}
For $A^{\major}_t$ and $A^{\minor}_t$ defined in \eqref{eq:linear_ansatz_A}, we have the following estimates
\begin{equation}
0<\frac{2 a_0}{T-t} \leq A^{\major}_t \leq \frac{2a_0}{T-t} + \frac{2}{3} \phi_0 (T-t),\quad 0< \frac{2a}{T-t} \leq A^{\minor}_t \leq \frac{2a}{T-t} + \frac{2}{3} \phi (T-t). 
\label{eq:linear_ansatz_A_bound}
\end{equation}
Moreover, for $t\in[0,T]$, define that
\begin{equation}
\begin{aligned}
\Phi^{\major}(t,s) &:= e^{- \frac{1}{2a_0}\int_t^s A^0_r dr} = \frac{\sinh(\sqrt{\phi_0/a_0}(T-s))}{\sinh(\sqrt{\phi_0/a_0}(T-t))},\\ \Phi^{\minor}(t,s) &:= e^{- \frac{1}{2a}\int_t^s A_r dr} = \frac{\sinh(\sqrt{\phi/a}(T-s))}{\sinh(\sqrt{\phi/a}(T-t))}.    
\end{aligned}
\label{eq:linear_ansatz_A_factor}
\end{equation}
Then we have
\begin{equation}
\Phi^{\major}(t,s) \leq \frac{T-s}{T-t} \leq 1, \quad \Phi^{\minor}(t,s) \leq \frac{T-s}{T-t} \leq 1.
\label{eq:linear_ansatz_A_factor_bound}
\end{equation}
\begin{proof}
\eqref{eq:linear_ansatz_A_bound} holds due to the fact that $\tanh(x)\leq x\leq (1+x^2/3)\tanh(x)$ for all $x>0$, and \eqref{eq:linear_ansatz_A_factor_bound} results from the fact that $\sinh(x)/x$ is monotonically increasing in $x>0$.
\end{proof}
\end{lemma}

We now proceed to prove Theorem~\ref{thm:fbsde_unique_continuation}. We first show that Theorem~\ref{thm:fbsde_unique_continuation} holds when $\mathfrak{p} = 0$.
\begin{lemma}
\label{lemma:fbsde_solution_p=0}
When $\mathfrak{p} = 0$, Theorem~\ref{thm:fbsde_unique_continuation} holds and the solution of \eqref{eqn:fbsde_continuation} can be expressed  as
\begin{equation}
\left\{
\begin{aligned}
B^{\major}_t &= \E\left[\int_t^T \Phi^{\major}(t,s) g^{\major}_s \dif s\middle|\F^{\major}_t\right],\\ 
B^{\minor}_t &= \E\left[\int_t^T \Phi^{\minor}(t,s) g^{\minor}_s \dif s\middle|\F^{\minor}_t\right],\\
Q^{\major}_t &=  q^{\major}_0 \Phi^{\major}(0,t) + \frac{1}{2a_0} \int_0^t \Phi^{\major}(s,t) B^{\major}_s \dif s,\\
Q^{\minor}_t &= \mathcal{Q}^{\minor}_0 \Phi^{\minor}(0,t) + \frac{1}{2a} \int_0^t \Phi^{\minor}(s,t) B^{\minor}_s \dif s,\\
P^{\major}_{t} &= - A^{\major}_t Q^{\major}_{t} + B^{\major}_t, \quad P^{\minor}_{t} = - A^{\minor}_t Q^{\minor}_{t} + B^{\minor}_t,
\end{aligned}\right.
\label{eq:fbsde_solution_p=0}
\end{equation}
and $Z^0_t$ and $Z_t$ are given by the martingale representation theorem. Moreover,
\begin{equation}
\begin{aligned}
\|B^{\major}\|_{\mathcal{H}^\gamma}^2 \leq C \E\left[\int_0^T\left|g^{\major}_s\right|^2 \dif s\right],\quad 
\|B^{\minor}\|_{\mathcal{H}^\gamma}^2 \leq C \E\left[\int_0^T \left|g^{\minor}_s\right|^2 \dif s\right],
\end{aligned}
\label{eq:fbsde_solution_B_norm_bound_p=0}
\end{equation}
and
\begin{equation}
\begin{aligned}
\|Q^{\major}\|_{\mathcal{H}^{1}}^2 \leq  C q^{\major}_0 + C \|B^{\major}\|_{\mathcal{H}^\gamma}^2, \quad \|Q^{\minor}\|_{\mathcal{H}^{1}}^2 \leq  C\E\left[\left|\mathcal{Q}^{\minor}_0\right|^2\right] + C \|B^{\minor}\|_{\mathcal{H}^\gamma}^2.
\end{aligned}
\label{eq:fbsde_solution_Q_norm_bound_p=0}
\end{equation}
\begin{proof}
When $\mathfrak{p} = 0$, \eqref{eqn:fbsde_continuation} becomes a standard FBSDE, and we can easily get the explicit form \eqref{eq:fbsde_solution_p=0} of the solution by, for example, Proposition 4.1.2 in \cite{zhang_backward_2017}. It remains to establish the integration properties of the solution, and we only need to consider $(Q^{\minor}_t,P^{\minor}_t,B^{\minor}_t)$ without loss of generality. First, for $\gamma\in(0,1/2)$, we can obtain by H\"{o}lder's inequality that
\begin{equation*}
    \frac{\left|B^{\minor}_t\right|}{(T-t)^\gamma} \leq \frac{1}{(T-t)^\gamma} \E\left[\int_t^T\left|g^{\minor}_s\right| \dif s \middle| \F^{\minor}_t \right] \leq\left(\E\left[\int_t^T \left|g^{\minor}_s\right|^{\frac{1}{1-\gamma}} \dif s \middle| \F^{\minor}_t \right]\right)^{1-\gamma}<\infty,
\end{equation*}
and the Doob's maximal inequality and Jensen's inequality imply that
\begin{equation}
\begin{aligned}
\E\left[\sup _{0 \leq t \leq T}\left|\frac{B^{\minor}_t}{(T-t)^\gamma}\right|^2\right] 
&\leq \E\left[\sup _{0 \leq t \leq T}\left(\E\left[\int_0^T\left|g^{\minor}_s\right|^{\frac{1}{1-\gamma}} \dif s \middle| \F^{\minor}_t\right]\right)^{2(1-\gamma)}\right] \\
& \leq \left(\frac{2(1-\gamma)}{1-2\gamma}\right)^{2(1-\gamma)} \E\left[\left(\int_0^T \left|g^{\minor}_s\right|^{\frac{1}{1-\gamma}} \dif s \right)^{2(1-\gamma)}\right]\\
&\leq \left(\frac{2(1-\gamma)}{1-2\gamma}\right)^{2(1-\gamma)} T^{1-2\gamma} \E\left[\int_0^T\left|g^{\minor}_s\right|^2 \dif s\right]\\
&= C \E\left[\int_0^T\left|g^{\minor}_s\right|^2 \dif s\right] < +\infty.
\end{aligned}
\end{equation}
Hence, $B^{\minor}\in\mathcal{H}^{\gamma}_{\F^{\minor}}$ and \eqref{eq:fbsde_solution_B_norm_bound_p=0} holds.

Furthermore, by the inequality \eqref{eq:linear_ansatz_A_factor_bound}, we have
\begin{equation*}
\begin{aligned}
\left|Q^{\minor}_t\right| &\leq \left|\mathcal{Q}^{\minor}_0\right| \frac{T-t}{T}  + \frac{1}{2a} \int_0^t \frac{T-t}{T-s} \left|B^{\minor}_s\right| \dif s \\
& \leq (T-t) \left[\left|\mathcal{Q}^{\minor}_0\right| + \frac{1}{2a} \left(\sup_{0\leq t\leq T} \frac{\left|B^{\minor}_t\right|}{(T-t)^{\gamma}}\right) \int_0^t (T-s)^{\gamma-1} \dif s \right] \\
& \leq (T-t) \left[\left|\mathcal{Q}^{\minor}_0\right| + \frac{1}{2a\gamma}T^{\gamma} \left(\sup_{0\leq t\leq T} \frac{\left|B^{\minor}_t\right|}{(T-t)^{\gamma}}\right) \right],
\end{aligned}
\end{equation*}
and then 
\begin{equation}
\begin{aligned}
\E\left[\sup _{0 \leq t \leq T}\left|\frac{Q^{\minor}_t}{T-t}\right|^2\right]
&\leq \E\left[\left(|\mathcal{Q}^{\minor}_0| + \frac{1}{2a\gamma}T^{\gamma} \left(\sup_{0\leq t\leq T} \frac{\left|B^{\minor}_t\right|}{(T-t)^{\gamma}}\right)\right)^2\right]  \\
&\leq 2 \E\left[\left|\mathcal{Q}^{\minor}_0\right|^2\right] + \frac{1}{2a^2\gamma^2}T^{2\gamma}\E\left[ \sup_{0\leq t\leq T} \left|\frac{\left|B^{\minor}_t\right|}{(T-t)^{\gamma}}\right|^2\right] \\
& = C\E\left[\left|\mathcal{Q}^{\minor}_0\right|^2\right] + C \|B^{\minor}\|_{\mathcal{H}^\gamma}^2 < +\infty,
\end{aligned}
\label{eq:fbsde_Q_norm_bound}
\end{equation}
So $Q^{\minor} \in \mathcal{H}^{1}_{\F^{\minor}}$ and \eqref{eq:fbsde_solution_Q_norm_bound_p=0} holds.

Finally, we establish the integratability property of $P^{\minor}_t$. First, by Cauchy-Schwarz inequality and the upper bound of $A_t$ in \eqref{eq:linear_ansatz_A_bound}, we have
\begin{equation*}
\begin{aligned}
\left|P^{\minor}_t\right|^2 &= (-A^{\minor}_t Q^{\minor}_t + B^{\minor}_t)^2 \leq 2 \left(A^{\minor}_t\right)^2 \left(Q^{\minor}_t\right)^2 + 2 \left(B^{\minor}_t\right)^2 \\
&= 2 [A^{\minor}_t (T-t)]^2 \left|\frac{Q^{\minor}_t}{T-t}\right|^2 + 2 (T-t)^{2\gamma} \left|\frac{B^{\minor}_t}{(T-t)^{\gamma}}\right|^2 \\
& \leq  2 \left[2a + \frac{2}{3} \phi (T-t)^2\right]^2 \left|\frac{Q^{\minor}_t}{T-t}\right|^2 + 2 (T-t)^{2\gamma} \left|\frac{B^{\minor}_t}{(T-t)^{\gamma}}\right|^2.
\end{aligned}
\end{equation*}
So, for any $\tau\in[0,T)$, we have
\begin{equation*}
\begin{aligned}
\E\left[\sup_{0\leq t\leq \tau} \left|P^{\minor}_t\right|^2\right]
&\leq  2 \left[2a + \frac{2}{3} \phi T^2\right]^2 \E\left[\sup_{0\leq t\leq T} \left|\frac{Q^{\minor}_t}{T-t}\right|^2\right] + 2 T^{2\gamma} \E\left[\sup_{0\leq t\leq T} \left|\frac{B^{\minor}_t}{(T-t)^{\gamma}}\right|^2\right]\\
&= C\|Q^{\minor}\|_{\mathcal{H}^1}^2 +C\|B^{\minor}\|_{\mathcal{H}^\gamma}^2 < +\infty.
\end{aligned}
\end{equation*}
In addition, we consider the integration by part of $Q_tP_t$ and obtain that
\begin{equation}
\begin{aligned}
Q^{\minor}_{T-\epsilon} P^{\minor}_{T-\epsilon} - Q^{\minor}_{0} P^{\minor}_{0} &= 2\phi \int_{0}^{T-\epsilon} \left(Q^{\minor}_t\right)^2 \dif t  - \int_{0}^{T-\epsilon} Q^{\minor}_t g^{\minor}_t \dif t \\
&\qquad \qquad + \frac{1}{2a}  \int_{0}^{T-\epsilon} \left(P^{\minor}_t\right)^2 \dif t + \int_{0}^{T-\epsilon} Q^{\minor}_t Z^{\minor}_t \dif \tilde{W}_t.
\end{aligned}
\label{eq:QP_T-eps}
\end{equation}
Because $A^{\minor}_t>0$, $B^{\minor} \in \mathcal{H}^{\gamma}_{\F^{\minor}}$ and $Q^{\minor} \in \mathcal{H}^{1}_{\F^{\minor}}$, we can show that, when $\epsilon \to 0$,
\begin{equation*}
Q^{\minor}_{T-\epsilon}P^{\minor}_{T-\epsilon} = Q^{\minor}_{T-\epsilon}(-A^{\minor}_{T-\epsilon} Q^{\minor}_{T-\epsilon} + B^{\minor}_{T-\epsilon})\leq Q^{\minor}_{T-\epsilon} B^{\minor}_{T-\epsilon} \to 0.    
\end{equation*}
Thus, taking expectations on both sides of \eqref{eq:QP_T-eps} and letting $\epsilon \to 0$ yields that
\begin{equation*}
\begin{aligned}
\frac{1}{2a}  \E \left[\int_{0}^{T} \left|P^{\minor}_t\right|^2 \dif t \right] &\leq - \E \left[ Q^{\minor}_{0} P^{\minor}_{0}\right] - 2\phi \E \left[\int_{0}^{T} \left(Q^{\minor}_t\right)^2 \dif t\right] + \E\left[ \int_{0}^{T} Q^{\minor}_t g^{\minor}_t \dif t\right]\\
&\leq \frac{1}{2}\E \left[ \left(Q^{\minor}_{0}\right)^2\right] + C\|Q^{\minor}\|_{\mathcal{H}_1}^2 +C\|B^{\minor}\|_{\mathcal{H}_\gamma}^2 \\
&\qquad + \frac{1}{2}\E\left[ \int_{0}^{T} \left(Q^{\minor}_t\right)^2 \dif t\right] + \frac{1}{2}\E\left[ \int_{0}^{T} \left(g^{\minor}_t\right)^2 \dif t\right] < +\infty.
\end{aligned}
\end{equation*}
Hence, $P^{\minor} \in D_{\F^{\minor}}^2([0, T] \times \Omega; \R)$.    
\end{proof}
\end{lemma}

Next, we prove the following results by the fixed-point argument. 
\begin{lemma}\label{lemma:fbsde_continuation_add_o}
If Theorem~\ref{thm:fbsde_unique_continuation} holds for some $\mathfrak{p}\in [0,1)$, then there exists a sufficiently small $\bar{\mathfrak{o}}>0$ such that Theorem~\ref{thm:fbsde_unique_continuation} holds for $\mathfrak{p} + \mathfrak{o}$ for any $\mathfrak{o}\in[0,\bar{\mathfrak{o}}]$, and the selection of $\bar{\mathfrak{o}}$ is independent of $\mathfrak{p}$, $g^0_t$ and $g_t$.
\end{lemma}
\begin{proof}
For given $P^{\major}_t \in L_{\F^{\major}}^2\left([0, T] \times \Omega ; \R\right)$ and $P^{\minor}_t \in L_{\F^{\minor}}^2\left([0, T] \times \Omega ; \R\right)$, we consider the following system of equations:
\begin{equation}
\left\{
\begin{aligned}
\dif \tilde{Q}^{\major}_{t} & = \frac{1}{2 a_{0}} \tilde{P}^{\major}_{t} \dif t, \\
d \tilde{Q}^{\minor}_{t} & = \frac{1}{2a} \tilde{P}^{\minor}_{t} \dif t, \\
- \dif \tilde{P}^{\major}_{t} & = \left[\mathfrak{p} \frac{\lambda}{2a}\E\left[\tilde{P}^{\minor}_{t} \middle| \F^{\major}_t \right] - 2\phi_0 \tilde{Q}^{\major}_{t} + \mathfrak{o}\frac{\lambda}{2a}\E\left[P_{t}^{\minor}\middle|\F^{\major}_t \right] + g^{\major}_t \right] \dif t - \tilde{Z}^{\major}_t \dif W^0_t, \\
- \dif \tilde{P}^{\minor}_{t} & = \left[\mathfrak{p} \frac{\lambda_0}{2 a_{0}}\tilde{P}^{\major}_{t} +  \mathfrak{p} \frac{\lambda}{2a}\E\left[\tilde{P}^{\minor}_{t}\middle|\F^{\major}_t \right] - 2\phi \tilde{Q}^{\minor}_{t} \right.\\
& \qquad \qquad \left. + \mathfrak{o} \frac{\lambda_0}{2 a_{0}} P^{\major}_{t} +  \mathfrak{o} \frac{\lambda}{2a}\E\left[P^{\minor}_{t} \middle| \F^{\major}_t \right] + g^{\minor}_t \right] \dif t - \tilde{Z}^{\minor}_t \dif \tilde{W}_t, \\
- \dif \tilde{B}^{\major}_t &= \left[\mathfrak{p}\frac{\lambda}{2a}\E\left[\tilde{P}^{\minor}_{t}\middle|\F^{\major}_t \right] - \frac{1}{2 a_0} A^{\major}_t \tilde{B}^{\major}_t + \mathfrak{o}\frac{\lambda}{2a}\E\left[P^{\minor}_{t} \middle| \F^{\major}_t \right] + g^{\major}_t \right] \dif t - \tilde{Z}^{\major}_t \dif W^0_t,\\
- \dif \tilde{B}^{\minor}_t &= \left[\mathfrak{p} \frac{\lambda_0}{2 a_{0}}\tilde{P}^{\major}_{t} +  \mathfrak{p} \frac{\lambda}{2a}\E\left[\tilde{P}^{\minor}_{t}\middle|\F^{\major}_t \right] - \frac{1}{2 a} A^{\minor}_t \tilde{B}^{\minor}_t \right.\\
&\qquad\qquad \left. + \mathfrak{o} \frac{\lambda_0}{2 a_{0}}P^{\major}_{t} +  \mathfrak{o} \frac{\lambda}{2a}\E\left[P^{\minor}_{t} \middle| \F^{\major}_t \right] + g^{\minor}_t\right] \dif t - \tilde{Z}^{\minor}_t d \tilde{W}_t,\\
\tilde{P}^{\major}_{t} &= - A^{\major}_t \tilde{Q}^{\major}_{t} + \tilde{B}^{\major}_t, \quad \tilde{P}^{\minor}_{t} = - A^{\minor}_t \tilde{Q}^{\minor}_{t} + \tilde{B}^{\minor}_{t},\\
\tilde{Q}^{\major}_{0} &= q^{\major}_{0}, \, \tilde{Q}^{\major}_{T} = 0,\, \tilde{Q}_{\minor} = \mathcal{Q}^{\minor}_{0}, \, \tilde{Q}^{\minor}_{T} = 0.
\end{aligned}\right.
\label{eq:fbsde_continuation_add_o}
\end{equation}
Note that 
\begin{equation*}
g^{\major}_t(P^{\major}_t,P^{\minor}_t) := \mathfrak{o} \frac{\lambda}{2a} \E\left[P^{\minor}_{t} \middle| \F^{\major}_t \right] + g^{\major}_t, \quad 
g^{\minor}_{t}(P^{\major}_t,P^{\minor}_t) = \mathfrak{o} \frac{\lambda_0}{2 a_{0}}P^{\major}_{t} +  \mathfrak{o} \frac{\lambda}{2a}\E\left[P^{\minor}_{t}\middle|\F^{\major}_t \right] + g^{\minor}_t
\end{equation*}
are both in $L_{\F^{\major}}^2\left([0, T] \times \Omega ; \R\right)$. Using the fact that the result of Theorem~\ref{thm:fbsde_unique_continuation} holds for some $\mathfrak{p}\in[0,1)$, as assumed in the statement of Lemma~\ref{lemma:fbsde_continuation_add_o},  we can show that the FBSDE \eqref{eq:fbsde_continuation_add_o} has a unique solution 
\begin{equation*}
(\tilde{Q}^{\major}_t,\tilde{Q}^{\minor}_t,\tilde{P}^{\major}_t,\tilde{P}^{\minor}_t,\tilde{B}^{\major}_t,\tilde{B}^{\minor}_t,\tilde{Z}^{\major}_t,\tilde{Z}^{\minor}_t)\in \mathcal{S}.    
\end{equation*} 
This defines a map $\varphi: (P^{\major}_t,P^{\minor}_t) \mapsto (\tilde{Q}^{\major}_t,\tilde{Q}^{\minor}_t,\tilde{P}^{\major}_t,\tilde{P}^{\minor}_t,\tilde{B}^{\major}_t,\tilde{B}^{\minor}_t,\tilde{Z}^{\major}_t,\tilde{Z}^{\minor}_t)$ and further induces a map $\Phi:(Q^{\major}_t,Q^{\minor}_t,P^{\major}_t,P^{\minor}_t,B^{\major}_t,B^{\minor}_t) \mapsto (\tilde{Q}^{\major}_t,\tilde{Q}^{\minor}_t,\tilde{P}^{\major}_t,\tilde{P}^{\minor}_t,\tilde{B}^{\major}_t,\tilde{B}^{\minor}_t)$. In the following, we prove that $\Phi$ is a contraction for some sufficiently small $\mathfrak{o}>0$ and then has a unique fixed point. For the fixed point, FBSDE \eqref{eq:fbsde_continuation_add_o} reduces to FBSDE \eqref{eqn:fbsde_continuation} with $\mathfrak{p}$ replaced by $\mathfrak{p}+\mathfrak{o}$, which yields the result of \Cref{lemma:fbsde_continuation_add_o}.

To establish the contraction property of $\Phi$, we consider two pairs of processes
\begin{equation*}
(P^{\major}_t,P^{\minor}_t),(P'^{\major}_t,P'^{\minor}_t)\in L_{\F^{\major}}^2\left([0, T] \times \Omega ; \R\right) \times L_{\F^{\major}}^2\left([0, T] \times \Omega ; \R\right),    
\end{equation*}
and their corresponding images of $\varphi$:
\begin{equation*}
\begin{aligned}
(\tilde{Q}^{\major}_t,\tilde{Q}^{\minor}_t,\tilde{P}^{\major}_t,\tilde{P}^{\minor}_t,\tilde{B}^{\major}_t,\tilde{B}^{\minor}_t,\tilde{Z}^{\major}_t,\tilde{Z}^{\minor}_t) &:= \varphi(P^{\major}_t,P^{\minor}_t), \\
(\tilde{Q}'^{\major}_t,\tilde{Q}'^{\minor}_t,\tilde{P}'^{\major}_t,\tilde{P}'^{\minor}_t,\tilde{B}'^{\major}_t,\tilde{B}'^{\minor}_t,\tilde{Z}'^{\major}_t,\tilde{Z}'^{\minor}_t) &:= \varphi(P'^{\major}_t,P'^{\minor}_t).
\end{aligned}
\end{equation*}
Then we have
\begin{equation}
(\Delta P^{\major}_t,\Delta P^{\minor}_t) := (P^{\major}_t,P^{\minor}_t) - (P'^{\major}_t,P'^{\minor}_t) \in L_{\F^{\major}}^2\left([0, T] \times \Omega ; \R\right) \times L_{\F^{\major}}^2\left([0, T] \times \Omega ; \R\right),
\label{eq:delta_P}
\end{equation}
and
\begin{equation}
\begin{aligned}
& (\Delta \tilde{Q}^{\major}_t, \Delta \tilde{Q}^{\minor}_t, \Delta \tilde{P}^{\major}_t, \Delta \tilde{P}^{\minor}_t, \Delta \tilde{B}^{\major}_t, \Delta \tilde{B}^{\minor}_t, \Delta \tilde{Z}^{\major}_t, \Delta \tilde{Z}^{\minor}_t)\\
&\qquad := (\tilde{Q}^{\major}_t,\tilde{Q}^{\minor}_t,\tilde{P}^{\major}_t,\tilde{P}^{\minor}_t,\tilde{B}^{\major}_t,\tilde{B}^{\minor}_t,\tilde{Z}^{\major}_t,\tilde{Z}^{\minor}_t) \\
&\qquad\qquad - (\tilde{Q}'^{\major}_t,\tilde{Q}'^{\minor}_t,\tilde{P}'^{\major}_t,\tilde{P}'^{\minor}_t,\tilde{B}'^{\major}_t,\tilde{B}'^{\minor}_t,\tilde{Z}'^{\major}_t,\tilde{Z}'^{\minor}_t) \in \mathcal{S}.
\end{aligned}
\label{eq:delta_solution}
\end{equation}
Moreover, $ (\Delta \tilde{Q}^{\major}_t, \Delta \tilde{Q}^{\minor}_t, \Delta \tilde{P}^{\major}_t, \Delta \tilde{P}^{\minor}_t, \Delta \tilde{B}^{\major}_t, \Delta \tilde{B}^{\minor}_t, \Delta \tilde{Z}^{\major}_t, \Delta \tilde{Z}^{\minor}_t) $ satisfies the following FBSDE:
\begin{equation}
\left\{
\begin{aligned}
\dif (\Delta\tilde{Q}^{\major}_{t}) & = \frac{1}{2 a_{0}} \Delta\tilde{P}^{\major}_{t} \dif t, \\
\dif (\Delta\tilde{Q}^{\minor}_{t}) & = \frac{1}{2a} \Delta\tilde{P}^{\minor}_{t} \dif t, \\
- \dif (\Delta\tilde{P}^{\major}_{t}) & = \left(\mathfrak{p}\frac{\lambda}{2a}\E\left[\Delta\tilde{P}^{\minor}_{t} \middle| \F^{\major}_t \right] - 2\phi_0 \Delta\tilde{Q}^{\major}_{t} + \mathfrak{o}\frac{\lambda}{2a}\E\left[\Delta P^{\minor}_{t} \middle| \F^{\major}_t \right] \right) \dif t - \Delta\tilde{Z}^{\major}_t \dif W^0_t, \\
- \dif (\Delta\tilde{P}^{\minor}_{t}) & = \left(\mathfrak{p} \frac{\lambda_0}{2 a_{0}}\Delta\tilde{P}^{\major}_{t} +  \mathfrak{p} \frac{\lambda}{2a}\E\left[\Delta\tilde{P}^{\minor}_{t}\middle|\F^{\major}_t \right] - 2\phi \Delta\tilde{Q}^{\minor}_{t} + \right.\\
& \qquad \qquad \qquad \left. \mathfrak{o} \frac{\lambda_0}{2 a_{0}}\Delta P^{\major}_{t} +  \mathfrak{o} \frac{\lambda}{2a}\E\left[\Delta  P^{\minor}_{t}\middle|\F^{\major}_t \right]\right) \dif t - \Delta\tilde{Z}^{\minor}_t \dif \tilde{W}_t, \\
- \dif (\Delta\tilde{B}^{\major}_t) &= \left(\mathfrak{p}\frac{\lambda}{2a} \E\left[\Delta\tilde{P}^{\minor}_{t}\middle|\F^{\major}_t \right] - \frac{1}{2 a_0} A^{\major}_t \Delta\tilde{B}^{\major}_t \right.\\
&\qquad\qquad\qquad \left. + \mathfrak{o}\frac{\lambda}{2a}\E\left[\Delta P^{\minor}_{t}\middle|\F^{\major}_t \right] \right) \dif t - \tilde{Z}^{\major}_t \dif W^0_t,\\
- \dif (\Delta\tilde{B}^{\minor}_t) &= \left(\mathfrak{p} \frac{\lambda_0}{2 a_{0}} \Delta\tilde{P}^{\major}_{t} +  \mathfrak{p} \frac{\lambda}{2a} \E\left[\Delta\tilde{P}^{\minor}_{t} \middle| \F^{\major}_t \right] - \frac{1}{2 a} A^{\minor}_t \Delta\tilde{B}^{\minor}_t   \right.\\
& \qquad\qquad\qquad \left.+ \mathfrak{o} \frac{\lambda_0}{2 a_{0}}\Delta P^{\major}_{t} + \mathfrak{o} \frac{\lambda}{2a}\E\left[\Delta P^{\minor}_{t} \middle| \F^{\major}_t \right] \right) \dif t - \Delta\tilde{Z}^{\minor}_t \dif \tilde{W}_t,\\
\Delta\tilde{P}^{\major}_{t} &= - A^{\major}_t \Delta\tilde{Q}^{\major}_{t} + \Delta\tilde{B}^{\major}_t, \quad \Delta\tilde{P}^{\minor}_{t} = - A^{\minor}_t \Delta\tilde{Q}^{\minor}_{t} + \Delta\tilde{B}^{\minor}_t,\\
\Delta\tilde{Q}^{\major}_{0} &= 0, \, \Delta\tilde{Q}^{\major}_{T} = 0,\, \Delta\tilde{Q}^{\minor}_{0} = 0, \, \Delta\tilde{Q}^{\minor}_{T} = 0.
\end{aligned}\right.
\end{equation}
For any $\epsilon>0$, applying integration by parts on $\Delta\tilde{Q}^{\major}_{t}\Delta\tilde{P}^{\major}_{t}$ and $\Delta\tilde{Q}^{\minor}_{t}\Delta\tilde{P}^{\minor}_{t}$ yield
\begin{equation}
\begin{aligned}
\Delta\tilde{Q}^{\major}_{T-\epsilon}\Delta\tilde{P}^{\major}_{T-\epsilon} &= - \mathfrak{p}\frac{\lambda}{2a} \int_{0}^{T-\epsilon}\Delta\tilde{Q}^{\major}_{t}\E\left[\Delta\tilde{P}^{\minor}_{t} \middle| \F^{\major}_t \right] \dif t + 2\phi_0 \int_{0}^{T-\epsilon}(\Delta\tilde{Q}^{\major}_{t})^2 \dif t \\
& \qquad - \mathfrak{o}\frac{\lambda}{2a} \int_{0}^{T-\epsilon} \Delta\tilde{Q}^{\major}_{t} \E\left[\Delta P^{\minor}_{t} \middle| \F^{\major}_t \right] \dif t + \frac{1}{2a_0} \int_{0}^{T-\epsilon}(\Delta\tilde{P}^{\major}_{t})^2 \dif t \\
& \qquad +  \int_{0}^{T-\epsilon} \Delta\tilde{Q}^{\major}_{t} \Delta\tilde{Z}^{\major}_{t} \dif W^{0}_t,
\end{aligned}
\label{eq:delta_QP_major}
\end{equation}
and
\begin{equation}
\begin{aligned}
\Delta\tilde{Q}^{\minor}_{T-\epsilon} \Delta\tilde{P}^{\minor}_{T-\epsilon} & = - \mathfrak{p} \frac{\lambda_0}{2 a_{0}} \int_{0}^{T-\epsilon} \Delta\tilde{Q}^{\minor}_{t} \Delta\tilde{P}^{\major}_{t} \dif t -  \mathfrak{p} \frac{\lambda}{2a} \int_{0}^{T-\epsilon} \Delta\tilde{Q}^{\minor}_{t} \E\left[\Delta\tilde{P}^{\minor}_{t} \middle| \F^{\major}_t \right] \dif t \\
& \qquad + 2\phi \int_{0}^{T-\epsilon} (\Delta\tilde{Q}^{\minor}_{t})^2 \dif t - \mathfrak{o} \frac{\lambda_0}{2 a_{0}}\int_{0}^{T-\epsilon} \Delta\tilde{Q}^{\minor}_{t} \Delta P^{\major}_{t} \dif t \\
&\qquad - \mathfrak{o} \int_{0}^{T-\epsilon} \Delta\tilde{Q}^{\minor}_{t} \frac{\lambda}{2a}\E\left[\Delta P^{\minor}_{t} \middle| \F^{\major}_t \right] \dif t\\
& \qquad + \frac{1}{2a}\int_{0}^{T-\epsilon}(\Delta\tilde{P}^{\minor}_{t})^2 \dif t + \int_{0}^{T-\epsilon} \Delta\tilde{Q}^{\minor}_{t} \Delta\tilde{Z}^{\minor}_t \dif \tilde{W}_t.
\end{aligned}
\label{eq:delta_QP_minor}
\end{equation}
Given that $A^{\major}_t>0$, $\Delta\tilde{B}^{\major}\in\mathcal{H}^{\gamma}_{\F}$, and $\Delta\tilde{Q}^{\major}\in\mathcal{H}^{1}_{\F}$, we have $$\Delta\tilde{Q}^{\major}_{T-\epsilon}\Delta\tilde{P}^{\major}_{T-\epsilon} = \Delta\tilde{Q}^{\major}_{T-\epsilon}(-A^{\major}_{T-\epsilon} \Delta\tilde{Q}^{\major}_{T-\epsilon} + \Delta\tilde{B}^{\major}_{T-\epsilon})\leq \Delta\tilde{Q}^{\major}_{T-\epsilon} \Delta\tilde{B}^{\major}_{T-\epsilon} \to 0$$ as $\epsilon \to 0$. Similarly, we can obtain that $\Delta\tilde{Q}^{\minor}_{T-\epsilon}\Delta\tilde{P}^{\minor}_{T-\epsilon} \to 0$ as $\epsilon \to 0$. Thus, taking expectations on \eqref{eq:delta_QP_major} and \eqref{eq:delta_QP_minor}, and letting $\epsilon \to 0$ yield that
\begin{equation*}
\begin{aligned}
\frac{1}{2a_0} \E\left[\int_{0}^{T}\left(\Delta\tilde{P}^{\major}_{t}\right)^2\right] \dif t &\leq  \mathfrak{p}\frac{\lambda}{2a} \E\left[\int_{0}^{T}\Delta\tilde{Q}^{\major}_{t}\E\left[\Delta\tilde{P}^{\minor}_{t}\middle|\F^{\major}_t \right] \dif t \right] - 2\phi_0 \E\left[\int_{0}^{T}\left(\Delta\tilde{Q}^{\major}_{t}\right)^2 \dif t\right] \\
& \qquad + \mathfrak{o}\frac{\lambda}{2a} \E\left[\int_{0}^{T} \Delta\tilde{Q}^{\major}_{t} \E\left[\Delta P^{\minor}_{t}\middle|\F^{\major}_t \right] \dif t\right],
\end{aligned}
\end{equation*}
and
\begin{equation*}
\begin{aligned}
 \frac{1}{2a} \E\left[\int_{0}^{T}(\Delta\tilde{P}^{\minor}_{t})^2 \dif t \right]  &\leq \mathfrak{p} \frac{\lambda_0}{2 a_{0}} \E\left[\int_{0}^{T} \Delta\tilde{Q}^{\minor}_{t}\Delta\tilde{P}^{\major}_{t} \dif t\right] +  \mathfrak{p} \frac{\lambda}{2a} \E\left[\int_{0}^{T} \Delta\tilde{Q}^{\minor}_{t} \E\left[\Delta\tilde{P}^{\minor}_{t}\middle|\F^{\major}_t \right]\right] \dif t \\
& \qquad - 2\phi \E\left[\int_{0}^{T} \left(\Delta\tilde{Q}^{\minor}_{t}\right)^2 \dif t\right] + \mathfrak{o} \frac{\lambda_0}{2 a_{0}}\E\left[\int_{0}^{T} \Delta\tilde{Q}^{\minor}_{t} \Delta P^{\major}_{t} \dif t\right] \\
& \qquad + \mathfrak{o} \frac{\lambda}{2a} \left[\int_{0}^{T} \Delta\tilde{Q}^{\minor}_{t} \E\left[\Delta  P^{\minor}_{t}\middle|\F^{\major}_t \right] \dif t\right].
\end{aligned}
\end{equation*}
Given that $\mathfrak{p}\leq 1$ and $2xy \leq \frac{x^2}{\theta} + \theta y^2$ for any $\theta>0$, we obtain that, for any $\theta_1,\theta_2,\theta_3>0$,
\begin{equation*}
\begin{aligned}
 \frac{1}{2a_0} & \E\left[\int_{0}^{T}(\Delta\tilde{P}^{\major}_{t})^2\right] \dif t \leq  \frac{\lambda\theta_1}{4a} \E\left[\int_{0}^{T}(\Delta\tilde{Q}^{\major}_{t})^2 \dif t \right] + \frac{\lambda}{4a\theta_1} \E\left[\int_{0}^{T}(\Delta\tilde{P}^{\minor}_{t})^2 \dif t \right]  \\
& \qquad - 2\phi_0 \E\left[\int_{0}^{T}(\Delta\tilde{Q}^{\major}_{t})^2 \dif t\right]  + \mathfrak{o}\frac{\lambda}{4a} \E\left[\int_{0}^{T} 
(\Delta\tilde{Q}^{\major}_{t})^2 \dif t\right] + \mathfrak{o}\frac{\lambda}{4a} \E\left[\int_{0}^{T} (\Delta P^{\minor}_{t})^2 \dif t\right],
\end{aligned}
\end{equation*}
and 
\begin{equation*}
\begin{aligned}
 \frac{1}{2a} & \E\left[\int_{0}^{T}(\Delta\tilde{P}^{\minor}_{t})^2 \dif t \right]  \leq \frac{\lambda_0\theta_2}{4 a_{0}} \E\left[\int_{0}^{T} (\Delta\tilde{Q}^{\minor}_{t})^2 \dif t\right] + \frac{\lambda_0}{4 a_{0}\theta_2} \E\left[\int_{0}^{T} (\Delta\tilde{P}^{\major}_{t})^2 \dif t\right] \\
 &\qquad + \frac{\lambda\theta_3}{4a} \E\left[\int_{0}^{T} (\Delta\tilde{Q}^{\minor}_{t})^2 \dif t\right] +  \frac{\lambda}{4a\theta_3} \E\left[\int_{0}^{T} (\Delta\tilde{P}^{\minor}_{t})^2 \dif t\right] - 2\phi \E\left[\int_{0}^{T} (\Delta\tilde{Q}^{\minor}_{t})^2 \dif t\right]\\
&\qquad + \mathfrak{o}\frac{\lambda_0}{4 a_{0}} \E\left[\int_{0}^{T} (\Delta\tilde{Q}^{\minor}_{t})^2 \dif t\right] + \mathfrak{o}\frac{\lambda_0}{4 a_{0}} \E\left[\int_{0}^{T} (\Delta P^{\major}_{t})^2 \dif t\right] \\
 &\qquad + \mathfrak{o}\frac{\lambda}{4a} \E\left[\int_{0}^{T} (\Delta\tilde{Q}^{\minor}_{t})^2 \dif t\right] +  \mathfrak{o}\frac{\lambda}{4a} \E\left[\int_{0}^{T} (\Delta P^{\minor}_{t})^2 \dif t\right].
\end{aligned}
\end{equation*}
These two inequalities yield that
\begin{equation*}
\begin{aligned}
 & \left(\frac{1}{2a_0}-\frac{\lambda_0}{4 a_{0}\theta_2}\right)\E\left[\int_{0}^{T}(\Delta\tilde{P}^{\major}_{t})^2 \dif t\right]  + \left(\frac{1}{2a}-\frac{\lambda}{4a\theta_1}-\frac{\lambda}{4a\theta_3}\right) \E\left[\int_{0}^{T}(\Delta\tilde{P}^{\minor}_{t})^2 \dif t \right] \\
 & + \left(2\phi_0-\frac{\lambda\theta_1}{4a}\right) \E\left[\int_{0}^{T}(\Delta\tilde{Q}^{\major}_{t})^2 \dif t\right]  + \left(2\phi-\frac{\lambda_0\theta_2}{4 a_{0}}-\frac{\lambda\theta_3}{4a}\right) \E\left[\int_{0}^{T} (\Delta\tilde{Q}^{\minor}_{t})^2 \dif t\right]\\
 &\leq  \frac{\lambda}{4a} \mathfrak{o}\E\left[\int_{0}^{T}(\Delta\tilde{Q}^{\major}_{t})^2 \dif t\right]  + \left(\frac{\lambda_0}{4 a_{0}}+\frac{\lambda}{4a}\right)\mathfrak{o} \E\left[\int_{0}^{T} (\Delta\tilde{Q}^{\minor}_{t})^2 \dif t\right]\\
&\qquad  + \mathfrak{o}\frac{\lambda_0}{4 a_{0}} \E\left[\int_{0}^{T} (\Delta P^{\major}_{t})^2 \dif t\right]  +  \mathfrak{o}\frac{\lambda}{2a} \E\left[\int_{0}^{T} (\Delta P^{\minor}_{t})^2 \dif t\right].
\end{aligned}
\end{equation*}
Based on Assumption~\ref{assumption:fbsde}, we can choose $\theta_1,\theta_2,\theta_3>0$ such that
\begin{equation*}
\frac{1}{2a_0}-\frac{\lambda_0}{4 a_{0}\theta_2},\quad \frac{1}{2a}-\frac{\lambda}{4a\theta_1}-\frac{\lambda}{4a\theta_3},\quad 2\phi_0-\frac{\lambda\theta_1}{4a},\quad 2\phi-\frac{\lambda_0\theta_2}{4 a_{0}}-\frac{\lambda\theta_3}{4a}
\end{equation*}
are all positive. This implies that there exists a sufficiently small $\mathfrak{o}>0$ such that
\begin{equation}
\begin{aligned}
 \E\left[\int_{0}^{T}[(\Delta\tilde{P}^{\major}_{t})^2+(\Delta\tilde{P}^{\minor}_{t})^2] \dif t\right] \leq   \mathfrak{o} \, C \, \E\left[\int_{0}^{T} [(\Delta P^{\major}_{t})^2 + (\Delta P^{\minor}_{t})^2] \dif t\right].
\end{aligned}
\label{eq:delta_P_contraction}
\end{equation}

Now, using the integrable properties in \eqref{eq:delta_P} and \eqref{eq:delta_solution}, we can apply the result of Lemma~\ref{lemma:fbsde_solution_p=0} to the FBSDE in \eqref{eq:fbsde_continuation_add_o}. More specifically, \eqref{eq:fbsde_solution_B_norm_bound_p=0} implies that
\begin{equation}
\begin{aligned}
\|\Delta\tilde{B}^{\major}\|_{\mathcal{H}^\gamma}^2 
&\leq C \E\left[\int_0^T\left|\mathfrak{p}\frac{\lambda}{2a}\E\left[\Delta\tilde{P}^{\minor}_{t}\middle|\F^{\major}_{t} \right] +  \mathfrak{o}\frac{\lambda}{2a}\E\left[\Delta P^{\minor}_{t}\middle|\F^{\major}_{t} \right]\right|^2 \dif t\right] \\
& \leq C \E\left[\int_0^T (\Delta\tilde{P}^{\minor}_{t})^2 \dif t\right] + C \mathfrak{o}^2 \E\left[\int_0^T (\Delta\tilde{P}^{\minor}_{t})^2 \dif t\right]\\
\end{aligned}
\label{eq:delta_B_contraction_major}
\end{equation}
and
\begin{equation}
\begin{aligned}\|\Delta\tilde{B}^{\minor}\|_{\mathcal{H}^\gamma}^2
&\leq C \E\left[\int_0^T\left|\mathfrak{p} \frac{\lambda_0}{2 a_{0}}\Delta\tilde{P}^{\major}_{t} +  \mathfrak{p} \frac{\lambda}{2a}\E\left[\Delta\tilde{P}^{\minor}_{t}\middle|\F^{\major}_{t} \right] + \mathfrak{o} \frac{\lambda_0}{2 a_{0}}\Delta P^{\major}_{t} +  \mathfrak{o} \frac{\lambda}{2a}\E\left[\Delta P^{\minor}_{t}\middle|\F^{\major}_{t} \right] \right|^2 \dif t\right] \\
& \leq C \E\left[\int_0^T (\Delta\tilde{P}^{\major}_{t})^2 \dif t\right] + C \E\left[\int_0^T (\Delta\tilde{P}^{\minor}_{t})^2 \dif t\right] \\
& \qquad \qquad \qquad + C \mathfrak{o}^2 \E\left[\int_0^T (\Delta\tilde{P}^{\major}_{t})^2 \dif t\right] + C \mathfrak{o}^2 \E\left[\int_0^T (\Delta\tilde{P}^{\minor}_{t})^2 \dif t\right].
\end{aligned}
\label{eq:delta_B_contraction_minor}
\end{equation}
In addition, \eqref{eq:fbsde_Q_norm_bound} yields that
\begin{equation}
\begin{aligned}
\|\Delta\tilde{Q}^{\major}\|_{\mathcal{H}^1}^2
\leq C \|\Delta\tilde{B}^{\major}\|_{\mathcal{H}^\gamma}^2, \quad \|\Delta\tilde{Q}^{\minor}\|_{\mathcal{H}^1}^2
\leq C \|\Delta\tilde{B}^{\minor}\|_{\mathcal{H}^\gamma}^2.
\end{aligned}
\label{eq:delta_Q_contraction}
\end{equation}
We combine  \eqref{eq:delta_P_contraction}--\eqref{eq:delta_Q_contraction} to obtain that, when $\mathfrak{o}$ is sufficiently small,
\begin{equation*}
\begin{aligned}
&\|\Delta\tilde{Q}^{\major}\|_{\mathcal{H}^1}^2 + \|\Delta\tilde{Q}^{\minor}\|_{\mathcal{H}^1}^2 +  \E\left[\int_{0}^{T}[(\Delta\tilde{P}^{\major}_{t})^2+(\Delta\tilde{P}^{\minor}_{t})^2] \dif t\right] + \|\Delta\tilde{B}^{\major}\|_{\mathcal{H}^1}^2 + \|\Delta\tilde{B}^{\minor}\|_{\mathcal{H}^1}^2 \\
& \leq C \mathfrak{o}^2 \E\left[\int_{0}^{T} \left((\Delta P^{\major}_{t})^2 + (\Delta P^{\minor}_{t})^2\right) \dif t\right].
\end{aligned}
\end{equation*}
for some $\mathfrak{o}>0$. Hence, there exists a sufficiently small $\bar{\mathfrak{o}}>0$ such that $\Phi$ is a contraction for any $\mathfrak{o}\in[0,\bar{\mathfrak{o}}]$, and the choice of $\mathfrak{o}$ is independent of $\mathfrak{p}$, $(P^{\major}_t,P^{\minor}_t)$ and $(P'^{\major}_t,P'^{\minor}_t)$.
\end{proof}

To summarize, Lemma~\ref{lemma:fbsde_solution_p=0} proves that Theorem~\ref{thm:fbsde_unique_continuation} holds when $\mathfrak{p} = 0$. In addition, we emphasize that the choice of $\mathfrak{o}$ in Lemma~\ref{lemma:fbsde_continuation_add_o} is independent from the choices of $\mathfrak{p}$,  $(P^{\major}_t,P^{\minor}_t)$ and $(P'^{\major}_t,P'^{\minor}_t)$. Hence, we can prove the statement in Theorem~\ref{thm:fbsde_unique_continuation} by induction on $\mathfrak{p}$, and further prove Theorem~\ref{thm:nash_fbsde_unique}.

\subsection{Proof of Theorem~\texorpdfstring{\ref{thm:nash_fbsde_verification}}{2}} \label{proof:nash_fbsde_verification}

We recall that $\left(Q^{\major,*}, Q^{\minor,*} ,P^{\major,*}, P^{\minor,*}, Z^{\major,*}, Z^{\minor,*}\right)$ is the unique solution of the FBSDE \eqref{eq:fbsde_nash}. In particular, $\left(Q^{\major,*}, P^{\major,*}, Z^{\major,*}\right)$ is the solution of the FBSDE \eqref{eq:fbsde_major_nash} given that $v^{\minor}_t = v^{\minor,*}$, and $\left(Q^{\minor,*}, P^{\minor,*}, Z^{\minor,*}\right)$ is the solution of the MV-FBSDE \eqref{eq:fbsde_minor_nash} given that $v^{\major}_t = v^{\major,*}$.

First, we prove that $v^{\major,*}$ is the optimal strategy of Problem~\ref{Problem:major_trader_major} given the representative minor trader's strategy is $v^{\minor,*}$, where $v^{\major,*}$ and $v^{\minor,*}$ are defined in \eqref{eq:nash_equilibrium_strategies}. To see that, let $v^{\major}\in\mathcal{A}^{\major,\kappa}$  be a generic admissible strategy of the major trader with the corresponding inventory process
$Q^{\major}_t= q^{\major}_0 + \int_{0}^{t} v^{\major}_s \dif s$.  Then we have
\begin{equation}
\begin{aligned}
\left(Q^{\major}_{T-\epsilon} - Q^{\major,*}_{T-\epsilon}\right) P^{\major,*}_{T-\epsilon} &= \int_0^{T-\epsilon} \left(Q^{\major}_{t} - Q^{\major,*}_{t}\right) \left(-\lambda \E\left[v^{\minor,*}_{t} \middle| \F^{\major}_{t} \right] + 2\phi \left(Q^{\major,*}_{t} - R_{t}\right)\right) \dif t \\
&\quad + \int_0^{T-\epsilon} P^{\major,*}_t \left(v^{\major}_{t} - v^{\major,*}_{t} \right) \dif t + \int_0^{T-\epsilon}\left(Q^{\major}_{t} - Q^{\major,*}_{t} \right) Z_t^{\major,*} \dif W_t^0. 
\end{aligned}
\label{eq:major_integrate}
\end{equation}
Given the facts that
\begin{equation*}
\begin{aligned}
\left|\E\left[Q^{\major}_{s} P^{\major,*}_{s}\right]\right| &= \left|\E\left[Q^{\major}_{s} \left(- A^{\major}_s Q^{\major,*}_{s} + B^{\major,*}_{s}\right) \right] \right| \\
&\leq \frac{C}{T-s} \E\left[\left|Q^{\major}_{s}\right|^{2} + \left|Q^{\major,*}_{s}\right|^{2} \right] + \E\left[\left|Q^{\major}_{s} B^{\major,*}_{s} \right| \right]\\
& = \frac{C}{T-s} \E\left[ \left| \int_{s}^{T} v^{\major}_{u} \dif u \right|^{2} + \left| \int_{s}^{T} v^{\major}_{u} \dif u \right|^{2} \right] + \E\left[\left|Q^{\major}_{s} B^{\major,*}_{s} \right|\right] \\
& \leq C \E\left[ \int_{s}^{T} \left|v^{\major}_{u} \right|^{2} \dif u  + \int_{s}^{T} \left|v^{\major}_{u}\right|^{2} \dif u \right] + \E\left[\left|Q^{\major}_{s} B^{\major,*}_{s}\right|\right] 
 \stackrel{s\nearrow T}{\rightarrow} 0
\end{aligned}
\end{equation*}
and
\begin{equation*}
 -\E\left[Q^{\major,*}_{T-\epsilon} P^{\major,*}_{T-\epsilon}\right] =  \E\left[A_{T-\epsilon}^0(Q^{\major,*}_{T-\epsilon})^2\right]  -\E\left[Q^{\major,*}_{T-\epsilon} B^{\major,*}_{T-\epsilon}\right] \geq - \E\left[Q^{\major,*}_{T-\epsilon} B^{\major,*}_{T-\epsilon}\right] \stackrel{\epsilon \searrow 0}{\rightarrow} 0,
\end{equation*}
we take expectations on both sides of \eqref{eq:major_integrate} and let $\epsilon\searrow0$,
\begin{equation*}
\begin{aligned}
\E\left[\int_0^T \left(Q^{\major}_{t}-Q^{\major,*}_{t}\right)  \left(-\lambda \E\left[ v^{\minor,*}_{t} \middle| \F^{\major}_{t} \right] \right. \right.&\left.\left. + 2\phi\left(Q^{\major,*}_{t} - R_t \right)\right) \dif t \right] 
 \geq \E\left[-\int_0^T P_t^{\major,*} \left(v^{\major}_{t} - v^{\major,*}_{t} \right) \dif t\right].    
\end{aligned}
\end{equation*}
Hence, using the FBSDE \eqref{eq:fbsde_major_nash}, we have
\begin{equation*}
\begin{aligned}
& J^{\major}\left(v^{\major};v^{\minor,*},R,q^{\major}_0\right) - J^{\major}\left(v^{\major,*};v^{\minor,*},R,q^{\major}_0\right)  \\
&= \E\left[\int_0^T \left(Q^{\major}_{t} - Q^{\major,*}_{t} \right) \left(-\lambda \E\left[v^{\minor,*}_{t} \middle|\F^{\major}_{t} \right] + 2\phi_0 \left(Q^{\major,*}_{t} - R_t \right) \right) \dif t \right.\\
&\qquad\qquad\qquad \left. + a_0 \int_0^T \left[\left(v^{\major}_{t}\right)^{2} - \left(v^{\major,*}_{t}\right)^{2} \right] \dif t  + \phi_0 \int_{0}^{T} \left( Q^{\major}_{t} - Q^{\major,*}_{t} \right)^{2} \dif t\right]\\
&\geq \E\left[-\int_0^T P_t^{\major,*} \left(v^{\major}_{t} - v^{\major,*}_{t} \right) \dif t + a_0 \int_0^T \left[\left(v^{\major}_{t}\right)^2 - \left(v^{\major,*}_{t}\right)^2\right] \dif t + \phi_0 \int_0^T \left(Q^{\major}_{t} - Q^{\major,*}_t\right)^2 \dif t\right]\\
&= \E\left[-2a_0 \int_0^T v_t^{\major,*} \left(v^{\major}_{t} - v^{\major,*}_{t} \right) \dif t + a_0 \int_0^T \left[\left(v^{\major}_{t}\right)^2 - \left(v^{\major,*}_{t}\right)^2\right] \dif t  + \phi_0 \int_0^T \left(Q^{\major}_{t} - Q^{\major,*}_t\right)^2 \dif t\right] \\
&= \E\left[a_0 \int_0^T \left(v^{\major}_{t} - v^{\major,*}_{t}\right)^2 \dif t  + \phi_0 \int_0^T \left(Q^{\major}_{t} - Q^{\major,*}_t\right)^2 \dif t\right] \geq 0.
\end{aligned}
\end{equation*}
This implies that, given the strategy of the representative minor trader $v^{\minor,*}$, $v^{\major,*}$ minimizes $J^{\major}\left(v^{\major};v^{\minor,*},R,q^{\major}_0\right)$. Therefore, $v^{\major,*}$ is the optimal strategy of the major trader.

Next, we prove that $v^{\minor,*}$ is the optimal strategy of the representative minor trader in Problem~\ref{Problem:minor_trader} given the major trader's strategy $v^{\major,*}$ and $\mu_t = \mu^{*}_t:=\E\left[v^{\minor,*}_t \middle|\F^{\major}_t\right]$. Here $v^{\major,*}$ and $v^{\minor,*}$ are defined in \eqref{eq:nash_equilibrium_strategies}. To see this, let $(Q^{*},P^{*},v^{*})$ be the solution of the FBSDE \eqref{eq:fbsde_minor_nash}, and $v^{\minor}\in\mathcal{A}^{\minor}$ be any admissible strategy of the representative minor trader with the corresponding inventory process $Q^{\minor}_{t} = \mathcal{Q}^{\minor}_0 + \int_{0}^{t} v^{\minor}_s \dif s $. Similarly, we can prove that
\begin{equation*}
\E\left[\int_0^T \left(Q^{\minor}_{t} - Q^{\minor,*}_{t} \right) \left[-\left(\lambda_0 v^{\major}_{t} + \lambda \mu^{*}_t\right) + 2 Q^{\minor,*}_{t}\right] \dif t \middle| \mathcal{Q}^{\minor}_0\right]
\geq \E\left[\int_0^T P^{\minor}_t \left( v^{\minor}_{t} - v^{\minor,*}_{t} \right) \dif t  \middle| \mathcal{Q}^{\minor}_0 \right],    
\end{equation*}
and
\begin{equation*}
\begin{aligned}
& J^{\minor}\left(v^{\minor}; v^{\major,*}, \mu^{*}, \Q^{\minor}_{0}\right) - J^{\minor}\left(v^{\minor,*}; v^{\major,*}, \mu^{*}, \Q^{\minor}_{0}\right) \\
&= \E\left[\int_0^T \left(Q^{\minor}_{t} - Q^{\minor,*}_{t}\right) \left[- \left(\lambda_0 v^{\major,*}_{t} + \lambda \mu^{*}_t\right) + 2 Q^{\minor,*}_{t} \right] \dif t  \right.\\
&\qquad \left. + a \int_0^T \left[\left(v^{\minor}_{t}\right)^{2} - \left(v^{\minor,*}_{t} \right)^{2} \right] \dif t + \phi \int_0^T \left(Q^{\minor}_{t} - Q^{\minor,*}_{t} \right)^{2} \dif t \middle| \mathcal{Q}^{\minor}_0\right]\\
& \geq \E\left[ \int_0^T P^{\minor,*}_t \left( v^{\minor}_{t} - v^{\minor,*}_{t}\right) \dif t  + a \int_0^T \left[\left(v^{\minor}_{t}\right)^{2} - \left(v^{\minor,*}_{t} \right)^{2} \right] \dif t + \phi \int_0^T \left(Q^{\minor}_{t} - Q^{\minor,*}_{t} \right)^{2} \dif t \middle| \mathcal{Q}^{\minor}_0 \right]\\
&= \E\left[2a\int_0^T v^{\minor,*}_{t} \left( v^{\minor}_{t} - v^{\minor,*}_{t}\right) \dif t + a \int_0^T \left[\left(v^{\minor}_{t}\right)^{2} - \left(v^{\minor,*}_{t} \right)^{2} \right] \dif t + \phi \int_0^T \left(Q^{\minor}_{t} - Q^{\minor,*}_{t} \right)^{2} \dif t \middle| \mathcal{Q}^{\minor}_0 \right]\\
&= \E\left[a \int_0^T \left( v^{\minor}_{t} - v^{\minor,*}_{t}\right)^2 \dif t + \phi \int_0^T \left(Q^{\minor}_{t} - Q^{\minor,*}_{t} \right)^{2} \dif t \middle| \mathcal{Q}^{\minor}_0 \right]\geq 0,
\end{aligned}
\end{equation*}
which implies that $v^{\minor,*}$ minimizes $J^{\minor}\left(v^{\minor}; v^{\major,*}, \mu^{*}, \Q^{\minor}_{0}\right)$. Therefore, $v^{\minor,*}$ is the optimal strategy of the representative minor trader.

Under the Nash equilibrium, we can further derive the optimal cost functional of the major trader and the representative minor trader. For the major trader, 
\begin{equation*}
\begin{aligned}
&\qquad J^{\major}\left(v^{\major,*};v^{\minor,*},R,q^{\major}_0\right) \\
&= \E\left[-\lambda \int_0^T \left(Q^{\major,*}_{t} - R_t \right) \E\left[v^{\minor,*}_{t}\middle|\F^{\major}_{t}\right] \dif t + a_0 \int_0^T \left(v^{\major,*}_{t}\right)^2 \dif t  + \phi_0 \int_0^T \left(Q^{\major,*}_{t} - R_t\right)^2 \dif t\right] \\ 
&= \E\left[\frac{1}{2} \int_0^T \left(Q^{\major,*}_{t} - R_t\right) \left[- \lambda \E\left[v^{\minor,*}_{t}\middle|\F^{\major}_{t}\right] + 2 \phi_0 \left(Q^{\major,*}_{t} - R_t\right)\right] \dif t  \right.\\
&\qquad\qquad \left. - \frac{\lambda}{2}  \int_0^T \left(Q^{\major,*}_{t} - R_t\right) \E\left[v^{\minor,*}_{t}\middle|\F^{\major}_{t}\right] \dif t + a_0 \int_0^T \left(v^{\major,*}_{t}\right)^2 \dif t \right]\\
&= \E\left[\frac{1}{2}\int_0^T \left(Q^{\major,*}_{t} - R_t\right) \dif P^{\major,*}_{t} -\frac{\lambda}{2}  \int_0^T \left(Q^{\major,*}_{t} - R_t\right) \E\left[v^{\minor,*}_{t}\middle|\F^{\major}_{t}\right] \dif t + a_0 \int_0^T \left(v^{\major,*}_{t}\right)^2 \dif t \right]\\
&= \E\left[\frac{1}{2}\int_0^T \left(Q^{\major,*}_{t} - R_t\right) \dif P^{\major,*}_{t}  + \frac{1}{2} \int_0^T P^{\major,*}_{t} \dif Q^{\major,*}_t - \frac{\lambda}{2}  \int_0^T \left(Q^{\major,*}_{t} - R_t\right) \E\left[v^{\minor,*}_{t}\middle|\F^{\major}_{t}\right] \dif t \right]\\
&= \E\left[\frac{1}{2} \int_0^T P^{\major,*}_{t} \dif R_t - \frac{\lambda}{2}  \int_0^T \left(Q^{\major,*}_{t} - R_t \right) \E\left[v^{\minor,*}_{t}\middle|\F^{\major}_{t}\right] \dif t \right] \\
&= \E\left[a_0 \int_0^T v^{\major,*}_{t} \dif R_t - \frac{\lambda}{2}  \int_0^T \left(Q^{\major,*}_{t} - R_t \right) \E\left[v^{\minor,*}_{t}\middle|\F^{\major}_{t}\right] \dif t \right],
\end{aligned}
\end{equation*}
where the last two equality follows from
\begin{equation*}
\begin{aligned}
\int_0^T \left(Q^{\major,*}_{t} - R_t\right) \dif P^{0,*}_{t}  +  \int_0^T P^{\major,*}_{t} \dif \left(Q^{\major,*}_{t} - R_t\right) = (Q^{\major,*}_{T}-R_T) P^{\major,*}_{T} - (Q^{\major,*}_{0}-R_0) P^{0,*}_{0} = 0. 
\end{aligned}    
\end{equation*}
For the representative minor trader, we combine \eqref{eq:consistency_condition}, \eqref{eq:fbsde_minor_nash_origin} and \eqref{eq:condidate_control_minor} to obtain
\begin{equation*}
\begin{aligned}
& \qquad J^{\major}\left(v^{\major,*};v^{\minor,*},R,q^{\major}_0\right)\\
&= \E\left[- \int_0^T Q^{\minor,*}_{t} \left(\lambda_0 v^{\major,*}_{t} + \lambda \mu_t\right) \dif t + a \int_0^T \left(v^{\minor,*}_{t}\right)^2 \dif t + \phi \int_0^T  \left(Q^{\minor,*}_{t}\right)^2 \dif t\middle| \mathcal{Q}^{\minor}_0 \right]\\ 
&= \E\left[- \frac{1}{2} \int_0^T Q^{\minor,*}_{t} \left(\lambda_0 v^{\major,*}_{t} + \lambda \mu_t - 2\phi Q^{\minor,*}_{t} \right) \dif t \right.\\
&\qquad\qquad \left. - \frac{1}{2}\int_0^T Q^{\minor,*}_{t} \left(\lambda_0 v^{\major,*}_{t} + \lambda \mu_t\right) \dif t + a \int_0^T \left(v^{\minor,*}_{t}\right)^2 \dif t \middle| \mathcal{Q}^{\minor}_0 \right]\\ 
&= \E\left[\frac{1}{2}\int_0^T Q^{\minor,*}_{t}  \dif P^{\minor,*}_t + \frac{1}{2} \int_0^T P^{\minor,*}_t \dif Q^{\minor,*}_t - \frac{1}{2}\int_0^T Q^{\minor,*}_{t} \left(\lambda_0 v^{\major,*}_{t} + \lambda \mu_t\right) \dif t  \middle| \mathcal{Q}^{\minor}_0 \right]\\ 
&= \E\left[\frac{1}{2} \mathcal{Q}^{\minor}_0 P^{\minor,*}_0 - \frac{1}{2}\int_0^T Q^{\minor,*}_{t} \left(\lambda_0 v^{\major,*}_{t} + \lambda \mu_t\right) \dif t  \middle| \mathcal{Q}^{\minor}_0 \right]\\ 
&= \E\left[a \mathcal{Q}^{\minor}_0 v^{\minor,*}_0 - \frac{1}{2}\int_0^T Q^{\minor,*}_{t} \left(\lambda_0 v^{\major,*}_{t} + \lambda \E\left[v^{\minor,*}_{t}\middle|\F^{\major}_{t}\right] \right) \dif t  \middle| \mathcal{Q}^{\minor}_0 \right].
\end{aligned}
\end{equation*}

\subsection{Proof of Throrem~\texorpdfstring{\ref{thm:approximate_Nash_equilibrium}}
{3}} \label{proof:approximate_Nash_equilibrium}
Without loss of generality, we focus on the major trader and the first minor trader $(i=1)$ in the proof. Recall that $\F^{\major}_t = \sigma\left(W^{0}_{s},0\leq s \leq t \right)$, $\F^{1}_t = \sigma\left(W^{0}_{s},W^{1}_{s},0\leq s \leq t \right)$, and $\mathcal{Q}^{1}_{0}$ is independent of $(W^{0},W^{1})$. 

As a preliminary step, we first state some distributional properties and optimality.  First, given that $v^{\minor,*} = \psi^{\minor}(\mathcal{Q}^{\minor}_0, W^{0}, W)$  and $v^{i,*} = \psi^{\minor}(\mathcal{Q}^{i}_0, W^{0}, W^i) $ (defined in \eqref{eq:approximate_Nash_trading_rate}) have the same distribution conditioned on $\F^{\major}_t$, we have 
\begin{equation}
  \E\left[v^{i,*}_t\middle| \F^{\major}_t\right] = \E\left[v^{\minor,*}_t\middle| \F^{\major}_t\right] = \mu^{*}_t  \text{ for any } i=1,2,\dots,n.   
\end{equation}
Moreover, comparing the definition of $v^{i,*}$ in \eqref{eq:approximate_Nash_trading_rate} and the expression of the Nash equilibrium of the Major-Minor MFG in \eqref{eq:MFG_Nash_Watanabe}, we observe that $v^{\major,*}$ minimizes $J^{\major}(v^{\major};v^{1,*},R,q^{\major}_0)$ and $v^{1,*}$ minimizes $J^{\minor}(v^{1}; v^{\major,*}, \mu^{*}, \Q^{1}_{0})$. More precisely, for any generic $v^{\major}\in\mathcal{A}^{\major,\kappa}$ (with $\mathcal{A}^{\major,\kappa}$ defined in \eqref{eq:admissible_set_major_constrained}), we have
\begin{equation}
J^{\major}(v^{\major,*};v^{1,*},R,q^{\major}_0) \leq J^{\major}(v^{\major};v^{1,*},R,q^{\major}_0),
\label{eq:optimality_major}
\end{equation}
with $J^{\major}$ defined in \eqref{eq:cost_functional_major}, 
and, for any generic $v^{1} \in \mathcal{A}^{1}$, we have
\begin{equation}
J^{\minor}(v^{1,*}; v^{\major,*}, \mu^{*}, \Q^{1}_{0}) \leq J^{\minor}(v^{1}; v^{\major,*}, \mu^{*}, \Q^{1}_{0}),   
\label{eq:optimality_minor}
\end{equation}
where $J^{\minor}$ is defined in \eqref{eq:cost_functional_minor}.

With the above properties of $v^{\major,*}$ and $v^{1,*},\dots, v^{n,*}$, we are now ready to prove \Cref{thm:approximate_Nash_equilibrium}, which contains three key steps:
\begin{itemize}[leftmargin=*]
\item \underline{Step 1:} We derive the convergence rate of the average trading rates of minor traders, $
\bar{\mu}_t^* = \frac{1}{N} \sum_{i=1}^{N} v_t^{i,*} $. More specifically, we provide upper bounds of order $\mathcal{O}(1/N)$ for both $\E\left[\int_0^T \left(\mu^{*}_t - \bar{\mu}_t^*\right)^2 d t \right]$ and $\E\left[\int_0^T \left(\mu^{*}_t - \bar{\mu}_t^* \right)^2 d t \middle| \Q^{1}_{0}=q^{1}_{0} \right]$.
\item \underline{Step 2:} We show the approximate Nash property for the major trader in \eqref{eq:approximate_Nash_major}. More specifically, {we show that $v^{\major,*}$ is $\epsilon$-optimal for the objective $J^{N,\major}\left(v^{\major};v^{1,*},\dots,v^{N,*}, R, q^{\major}_{0} \right)$ with $\epsilon = \mathcal{O}(1/\sqrt{N})$.}
\item \underline{Step 3:} We show the approximate Nash property for the representative minor trader in \eqref{eq:approximate_Nash_minor}. More specifically, {we show that $v^{1,*}$ is $\epsilon$-optimal for the  objective $J^{N,1}\left(v^{1}; v^{\major,*}, v^{2,*},\dots,v^{N,*}, q^{1}_0\right)$ with $\epsilon = \mathcal{O}(1/\sqrt{N})$.}
\end{itemize}
We now proceed to the formal proof.

\underline{Step 1.} We prove upper bounds on $\E\left[\int_0^T \left(\mu^{*}_t - \bar{\mu}_t^*\right)^2 d t \right]$ and $\E\left[\int_0^T \left(\mu^{*}_t - \bar{\mu}_t^* \right)^2 d t \middle| \Q^{1}_{0}=q^{1}_{0} \right]$ of order $\mathcal{O}(1/N)$.  

First, we consider the cross terms in $\E\left[\int_0^T \left(\mu^{*}_t - \bar{\mu}_t^*\right)^2 d t \right]$ and $\E\left[\int_0^T \left(\mu^{*}_t - \bar{\mu}_t^* \right)^2 d t \middle| \Q^{1}_{0}=q^{1}_{0} \right]$. For $i \neq j$ and $i,j \neq 1$, because $(\mathcal{Q}^{i}_0, W^{i})$ and $(\mathcal{Q}^{j}_0, W^{j})$ are independent, we have   $v^{i,*}_t = \psi^{\minor}(\mathcal{Q}^{i}_0, W^{0}, W^i)$ and $v^{j,*}_t = \psi^{\minor}(\mathcal{Q}^{j}_0, W^{0}, W^j)$ are independent conditional on $\mathcal{F}^{\major}_t$. Hence,
\begin{equation*}
\begin{aligned}
&\qquad\E\left[\int_0^T  \left(\mu^{*}_t - v^{i,*}_t\right)\left(\mu^{*}_t - v^{j,*}_t\right) \dif t \middle| \Q^{1}_{0}=q^{1}_{0}\right]\\
&= \int_0^T  \E\left[\left(\mu^{*}_t - v^{i,*}_t\right)\left(\mu^{*}_t - v^{j,*}_t\right) \right] \dif t \\ 
&= \int_0^T  \E\left[\E\left[\left(\mu^{*}_t - v^{i,*}_t\right)\left(\mu^{*}_t - v^{j,*}_t\right)\middle| \mathcal{F}^{\major}_t\right] \right] \dif t \\
&= \int_0^T  \E\left[\E\left[\left(\mu^{*}_t - v^{i,*}_t\right)\middle| \mathcal{F}^{\major}_t\right]\E\left[\left(\mu^{*}_t - v^{j,*}_t\right)\middle| \mathcal{F}^{\major}_t\right] \right] \dif t =0,
\end{aligned}
\end{equation*}
where the second  last equality holds since,  $v^{i,*}_t$ and $v^{j,*}_t$ are independent conditioned on $\mathcal{F}^{\major}_t$.
In addition, for $j\neq 1$ and for any generic strategy $v^{1}\in \mathcal{A}^{1}$ with $\mathcal{A}^{1}$ defined in \eqref{eq:admissible_set_i_constrained},
\begin{equation}
\begin{aligned}
&\E\left[\int_0^T  \left(\mu^{*}_t - v^{1}_t\right)\left(\mu^{*}_t - v^{j,*}_t\right) \dif t \middle| \Q^{1}_{0}=q^{1}_{0}\right] \\
&= \E\left[\int_0^T  \E\left[\left(\mu^{*}_t - v^{1}_t\right)\left(\mu^{*}_t - v^{j,*}_t\right)\middle| \tilde{\F}^{1}_t \right] \dif t \middle| \Q^{1}_{0}=q^{1}_{0}\right] \\ 
&= \E\left[\int_0^T  \left(\mu^{*}_t - v^{1}_t\right) \E\left[\left(\mu^{*}_t - v^{j,*}_t\right)\middle| \F^{\major}_t \right] \dif t \middle| \Q^{1}_{0}=q^{1}_{0}\right] = 0,
\end{aligned}
\label{eq:cross_term_conditional_gerenal_1}
\end{equation}
where the first equality holds by the tower property of the conditional expectation, and the second equality holds because $v^{1}$ is $\tilde{\F}^{1}_t$-measurable and $v^{j,*}$ is independent of $(\mathcal{Q}^{1}_{0},W^1)$. {In particular, substituting $v^{1}_{t} = v^{1,*}_{t}$ into \eqref{eq:cross_term_conditional_gerenal_1} yields}
\begin{equation*}
\E\left[\int_0^T  \left(\mu^{*}_t - v^{1,*}_t\right)\left(\mu^{*}_t - v^{j,*}_t\right) \dif t \middle| \Q^{1}_{0}=q^{1}_{0}\right] = 0.    
\end{equation*}
Therefore, for any $i \neq j$, we have
\begin{equation}
\E\left[\int_0^T  \left(\mu^{*}_t - v^{i,*}_t\right)\left(\mu^{*}_t - v^{j,*}_t\right) \dif t \middle| \Q^{1}_{0}=q^{1}_{0}\right] = 0,
\label{eq:cross_term_conditional}
\end{equation}
and then
\begin{equation}
\E\left[\int_0^T  \left(\mu^{*}_t - v^{i,*}_t\right)\left(\mu^{*}_t - v^{j,*}_t\right) \dif t \right] = 0.    
\label{eq:cross_term_unconditional}
\end{equation}

Next, we consider the quadratic terms in $\E\left[\int_0^T \left(\mu^{*}_t - \bar{\mu}_t^*\right)^2 d t \right]$ and $\E\left[\int_0^T \left(\mu^{*}_t - \bar{\mu}_t^* \right)^2 d t \middle| \Q^{1}_{0}=q^{1}_{0} \right]$.  For $i=1,2,\dots,N$, we have
\begin{equation}
\begin{aligned}
\E\left[\int_0^T \left(\mu^{*}_t - v^{i,*}_t\right)^2 \dif t \right] 
& = \E\left[ \int_0^T \E\left[\left(\mu^{*}_t\right)^2 + \left(v^{i,*}_t\right)^2 - 2 \mu^{*}_t v^{i,*}_t \middle| \F^{\major}_t \right] \dif t \right] \\
&= \E\left[ \int_0^T \left(\left(\mu^{*}_t\right)^2 + \E\left[ \left(v^{i,*}_t\right)^2 \middle| \F^{\major}_t \right] - 2 \mu^{*}_t \E\left[ v^{i,*}_t \middle| \F^{\major}_t \right] \right)\dif t \right]   \\
&= \E\left[ \int_0^T \left( \E\left[ \left(v^{i,*}_t\right)^2 \middle| \F^{\major}_t \right] - \left(\mu^{*}_t\right)^2 \right)\dif t \right] \\
& \leq \E\left[\int_0^T \left(v^{i,*}_t\right)^2 d t\right] \leq  \kappa, 
\end{aligned}
\label{eq:quadratic_term_unconditional}
\end{equation}
where the first equality is due to the tower property of conditional expectation, the second equality holds because $\mu^{*}_t$ is $\mathcal{F}^{\major}_t$-measurable, the third equality is due to the fact $\E\left[ v^{i,*}_t \middle| \F^{\major}_t \right] = \mu^{*}_t$, the first inequality holds because $\left(\mu^{*}_t\right)^2\geq 0$, and the last inequality holds by the definition of $\mathcal{A}^{\major,\kappa}$ in \eqref{eq:admissible_set_major_constrained}. Then, for any $i=2,3,\dots,N$, we further obtain that
\begin{equation}
\begin{aligned}
\E\left[\int_0^T \left(\mu^{*}_t - v^{i,*}_t\right)^2 \dif t \middle| \Q^{1}_{0}=q^{1}_{0} \right] = \E\left[\int_0^T \left(\mu^{*}_t - v^{i,*}_t\right)^2 \dif t \right] \leq  \kappa.    
\end{aligned}
\label{eq:quadratic_term_conditional_not_1}
\end{equation}
In addition, for any generic strategy $v^{1}\in \mathcal{A}^{1}$,
\begin{equation}
\begin{aligned}
&\E\left[\int_0^T \left(\mu^{*}_t - v^{1}_t\right)^2 \dif t \middle| \Q^{1}_{0}=q^{1}_{0} \right] \\
&\leq 2 \E\left[\int_0^T \left(\E\left[v^{1,*}_t\middle| \F^{\major}_t\right]\right)^2 \dif t \middle| \Q^{1}_{0}=q^{1}_{0} \right] + 2 \E\left[\int_0^T \left(v^{1}_t\right)^2 d t \middle| \Q^{1}_{0}=q^{1}_{0}\right]\\
&= 2 \E\left[\int_0^T \left(\E\left[v^{1,*}_t\middle| \F^{\major}_t\right]\right)^2 \dif t \right] + 2 \E\left[\int_0^T \left(v^{1}_t\right)^2 d t \middle| \Q^{1}_{0}=q^{1}_{0}\right]\\
&\leq 2 \E\left[\int_0^T \left| v^{1,*}_t\right|^{2} \dif t \right] + 2 \E\left[\int_0^T \left|v^{1}_t\right|^2 d t \middle| \Q^{1}_{0}=q^{1}_{0}\right] \leq 2 \kappa + 2 K\left(q^{1}_{0}\right),    
\end{aligned}
\label{eq:quadratic_term_conditional_gerenal_1}
\end{equation}
where the first inequality comes from the facts that $(x-y)^{2}\leq 2x^2+2y^2$ $\forall x,y\in\mathbb{R}$ and $\mu^{*}_t = \E\left[v^{1,*}_t\middle| \F^{\major}_t\right]$; the first equality holds by the independence of $\mathcal{F}^{\major}$ and $\mathcal{Q}^{1}_{0}$; the second inequality holds because $\E\left[\left(\E\left[v^{1,*}_t\middle| \F^{\major}_t\right]\right)^2 \right] \leq \E\left[\left| v^{1,*}_t\right|^{2}  \right] $, and the final inequality holds by \eqref{eq:condition_kappa} and the definition of  $\mathcal{A}^{1}$ in \eqref{eq:admissible_set_i_constrained}. In particular, substituting $v^{1}_{t} = v^{1,*}_{t}$ into \eqref{eq:quadratic_term_conditional_gerenal_1} yields
\begin{equation}
\E\left[\int_0^T \left(\mu^{*}_t - v^{1,*}_t\right)^2 \dif t \middle| \Q^{1}_{0}=q^{1}_{0} \right] \leq 2 \kappa + 2 K\left(q^{1}_{0}\right).
\label{eq:quadratic_term_conditional_1}
\end{equation}

Finally, combining \eqref{eq:cross_term_unconditional} and \eqref{eq:quadratic_term_unconditional}, we obtain
\begin{equation}
\begin{aligned}
& \E\left[\int_0^T \left(\mu^{*}_t - \bar{\mu}_t^* \right)^2 d t \right] \\
&= \frac{1}{N^2}\sum_{i\neq j}\E\left[\int_0^T  \left(\mu^{*}_t - v^{i,*}_t\right)\left(\mu^{*}_t - v^{j,*}_t\right) d t \right] + \frac{1}{N^2}\sum_{i=1}^{N}\E\left[\int_0^T \left(\mu^{*}_t - v^{i,*}_t\right)^2 d t \right] \leq \frac{1}{N} \kappa.
\end{aligned}
\label{eq:convengence_rate_unconditional}
\end{equation}
In addition, combining \eqref{eq:cross_term_conditional}, \eqref{eq:quadratic_term_conditional_1} and \eqref{eq:quadratic_term_conditional_not_1}, we  obtain
\begin{equation}
\begin{aligned}
\E\left[\int_0^T \left(\mu^{*}_t - \bar{\mu}_t^* \right)^2 d t \middle| \Q^{1}_{0}=q^{1}_{0} \right]
&= \frac{1}{N^2}\sum_{i\neq j}\E\left[\int_0^T  \left(\mu^{*}_t - v^{i,*}_t\right)\left(\mu^{*}_t - v^{j,*}_t\right) d t \middle| \Q^{1}_{0}=q^{1}_{0} \right] \\
&\qquad\qquad + \frac{1}{N^2}\sum_{i=1}^{N}\E\left[\int_0^T \left(\mu^{*}_t - v^{i,*}_t\right)^2 d t \middle| \Q^{1}_{0}=q^{1}_{0} \right] \\
&\leq \frac{1}{N^2} \left((N+1)\kappa +  2 K\left(q^{1}_{0}\right)  \right).
\end{aligned}
\label{eq:convengence_rate_conditional}
\end{equation}

\underline{Step 2.} Now, we are ready to prove the approximate Nash property for the major trader in \eqref{eq:approximate_Nash_major}.  For any generic admissible strategy $v^{\major}\in \mathcal{A}^{\major,\kappa}$ of the major trader, we denote $Q^{\major}_t = q^{\major}_{0} + \int_0^t v^{\major}_s \dif s$. Then we have
\begin{equation}
\begin{aligned}
&J^{N,\major}\left(v^{\major};v^{1,*},\dots,v^{N,*}, R, q^{\major}_{0} \right) - J^{N,\major}\left(v^{\major,*};v^{1,*},\dots,v^{N,*},R, q^{\major}_{0} \right)\\
& \geq J^{N,\major}\left(v^{\major};v^{1,*},\dots,v^{N,*}, R, q^{\major}_{0} \right) - J^{\major}(v^{\major};v^{1,*},R,q^{\major}_0) \\
& \qquad + J^{\major}(v^{\major,*};v^{1,*},R,q^{\major}_0) - J^{N,\major}\left(v^{\major,*};v^{1,*},\dots,v^{N,*},R, q^{\major}_{0} \right)\\
&= \left\{\E\left[-\lambda \int_0^T \left(Q^{\major}_{t} - R_t \right) \bar{\mu}^{*}_t d t + a_0 \int_0^T \left(v^{\major}_{t}\right)^2 \dif t + \phi_0 \int_0^T \left(Q^{\major}_{t} - R_t\right)^{2} d t \right] \right.\\
& \qquad \left. -  \E\left[-\lambda \int_0^T \left(Q^{\major}_{t} - R_t \right) \mu^{*}_t d t + a_0 \int_0^T \left(v^{\major}_{t}\right)^2 \dif t + \phi_0 \int_0^T \left(Q^{\major}_{t} - R_t \right)^2 d t \right] \right\}\\
& \qquad +  \left\{\E\left[-\lambda \int_0^T \left(Q^{\major,*}_{t} - R_t \right) \mu^{*}_t d t + a_0 \int_0^T \left(v^{\major,*}_{t} \right)^2 \dif t + \phi_0 \int_0^T \left(Q^{\major,*}_{t} - R_t\right)^{2} d t\right] \right.\\
& \qquad \left. - \E\left[-\lambda \int_0^T (Q^{\major,*}_{t}-R_t) \bar{\mu}^{*}_t d t + a_0 \int_0^T (v^{\major,*}_{t})^2 \dif t + \phi_0 \int_0^T (Q^{\major,*}_{t}-R_t)^2 d t\right] \right\} \\
&= \E\left[-\lambda \int_0^T \left(Q^{\major}_{t} - R_t \right) \left(\bar{\mu}^{*}_t - \mu^{*}_t\right) d t \right] + \E\left[-\lambda \int_0^T \left(Q^{\major,*}_{t} - R_t \right) \left(\mu^{*}_t - \bar{\mu}^{*}_t\right) d t \right] \\
&= \lambda\E\left[\int_0^T \left(Q^{\major}_{t} - Q^{\major,*}_{t}\right) \left(\mu^{*}_t - \bar{\mu}^{*}_t \right) d t \right],
\end{aligned}
\label{eq:approximate_Nash_major_bound_1}
\end{equation}
where the inequality holds by \eqref{eq:optimality_major} and the first equality holds by the definitions of $J^{N,\major}$ in \eqref{eq:cost_functional_major_N} and $J^{\major}$ in \eqref{eq:cost_functional_major}.

We note that for any time $t\in[0,T]$, $Q^{\major}_{t} = Q^{\major}_{T} - \int_{t}^{T}v^{\major}_{t} dt = - \int_{t}^{T}v^{\major}_{t} \dif t $. Hence, by Cauchy-Schwarz inequility,
\begin{equation}
\begin{aligned}
\E\left[\left(Q^{\major}_{t}\right)^2\right] &= \E\left[\left(-\int_{t}^{T}v^{\major}_{t} \dif t\right)^2\right]\leq (T-t)\E\left[\int_{t}^{T}\left(v^{\major}_{t}\right)^2 \dif t\right]\leq T \kappa.
\end{aligned}
\label{eq:bound_Q_2_major}
\end{equation}
Then
\begin{equation}
\begin{aligned}
&\left|\E\left[\int_0^T \left(Q^{\major}_{t} - Q^{\major,*}_{t}\right) \left(\mu^{*}_t - \bar{\mu}^{*}_t \right) d t \right]\right| \\
&\leq \left(\E\left[\int_0^T \left(Q^{\major}_{t} - Q^{\major,*}_{t}\right)^{2}  d t \right]\right)^{1/2} \left(\E\left[\int_0^T \left(\mu^{*}_t - \bar{\mu}^{*}_t \right)^2 d t \right]\right)^{1/2}\\
&\leq \left(2 \E\left[\int_0^T \left(Q^{\major}_{t}\right)^{2} d t \right] + 2 \E\left[\int_0^T \left(Q^{\major,*}_{t}\right)^{2}  d t \right] \right)^{1/2}\left(\E\left[\int_0^T \left(\mu^{*}_t - \bar{\mu}^{*}_t \right)^2 d t \right]\right)^{1/2}\\
&\leq \left(4 T^2 \kappa\right)^{1/2}\left(\frac{1}{N}\kappa\right)^{1/2} = \frac{2 T\kappa}{\sqrt{N}},
\end{aligned}
\label{eq:approximate_Nash_major_bound_2}
\end{equation}
where the first inequality holds by the Cauchy-Schwarz inequality, the second inequality is based on the fact that $(x-y)^{2}\leq 2x^2 + 2 y^2$ $\forall x,y\in \mathbb{R}$, and the last inequality holds by \eqref{eq:bound_Q_2_major} and \eqref{eq:convengence_rate_unconditional}.

Finally, we combine \eqref{eq:approximate_Nash_major_bound_1} and \eqref{eq:approximate_Nash_major_bound_2} to obtain 
\begin{equation*}
J^{N,\major}\left(v^{\major};v^{1,*},\dots,v^{N,*}, R, q^{\major}_{0} \right)   \geq J^{N,\major}\left(v^{\major,*};v^{1,*},\dots,v^{N,*}, R, q^{\major}_{0} \right) - \lambda \frac{2T\kappa}{\sqrt{N}},
\end{equation*}
which is the desired result in \eqref{eq:approximate_Nash_major}.

\underline{Step 3.} Next, we prove the approximate Nash property for the first minor trader in \eqref{eq:approximate_Nash_minor}. For a generic admissible strategy $v^{1} \in \mathcal{A}^{1}$ of the first minor trader, we denote $Q^{1}_t = \mathcal{Q}^{1}_{0} + \int_0^t v^{1}_{s} \dif s$.  Then, we have
\begin{equation}
\begin{aligned}
&J^{N,1}\left(v^{1}; v^{\major,*}, v^{2,*},\dots,v^{N,*}, q^{1}_0\right) - J^{N,1}\left(v^{1,*}; v^{\major,*}, v^{2,*},\dots,v^{N,*}, q^{1}_0\right) \\
& \geq J^{N,1}\left(v^{1}; v^{\major,*}, v^{2,*},\dots,v^{N,*}, q^{1}_0\right) - J^{\minor}(v^{1}; v^{\major,*}, \mu^{*}, \Q^{\minor}_{0})\\
&\qquad + J^{\minor}(v^{1,*}; v^{\major,*}, \mu^{*}, \Q^{\minor}_{0})  - J^{N,1}\left(v^{1,*}; v^{\major,*}, v^{2,*},\dots,v^{N,*}, q^{1}_0\right) \\
&= \left\{\E\left[- \int_0^T Q^{1}_{t} \left(\lambda_0 v^{\major}_{t} + \lambda \left(\frac{1}{N}v^{1}_t + \frac{1}{N}\sum_{i=2}^{N} v^{i,*}_t \right)\right) d t + a \int_0^T \left(v^{1}_{t}\right)^{2} \dif t + \phi \int_0^T\left(Q^{1}_{t}\right)^2 \dif t \middle| \Q^{1}_{0}=q^{1}_{0} \right]\right. \\
& \qquad \left. - \E\left[- \int_0^T Q^{1}_{t} \left(\lambda_0 v^{\major,*}_{t} + \lambda \mu^{*}_t\right) d t + a \int_0^T \left(v^{1}_{t}\right)^2 \dif t + \phi \int_0^T\left(Q^{1}_{t}\right)^2 \dif t \middle| \Q^{1}_{0}=q^{1}_{0} \right] \right\} \\
& \qquad + \left\{ \E\left[- \int_0^T Q^{1,*}_{t} \left(\lambda_0 v^{\major,*}_{t} + \lambda \mu^{*}_t \right) d t + a \int_0^T \left(v^{1,*}_{t}\right)^2 \dif t + \phi \int_0^T\left(Q^{1,*}_{t}\right)^2 \dif t \middle| \Q^{1}_{0}=q^{1}_{0} \right] \right. \\
&\qquad \left. - \E\left[- \int_0^T Q^{1,*}_{t} \left(\lambda_0 v^{\major,*}_{t} + \lambda \left(\bar{\mu}^{*}_t \right)\right) d t + a \int_0^T \left(v^{1,*}_{t}\right)^2 \dif t + \phi \int_0^T\left(Q^{1,*}_{t}\right)^2 \dif t \middle| \Q^{1}_{0}=q^{1}_{0} \right] \right\}\\
&= - \lambda \E\left[\int_0^T Q^{1}_{t}   \left(\frac{1}{N}v^{1}_t + \frac{1}{N}\sum_{i=2}^{N} v^{i,*}_t - \mu^{*}_t \right) d t \middle| \Q^{1}_{0}=q^{1}_{0} \right] + \lambda \E\left[ \int_0^T Q^{1,*}_{t} \left(\bar{\mu}^{*}_t - \mu^{*}_t\right) d t \middle| \Q^{1}_{0}=q^{1}_{0} \right].
\end{aligned}
\label{eq:approximate_Nash_minor_1}
\end{equation}
where the inequality holds by \eqref{eq:optimality_minor}, and the first equality holds by  the definitions of $J^{N,1}$ in \eqref{eq:cost_functional_minor_N} and $J^{\minor}$ in \eqref{eq:cost_functional_minor}.

Note that for any $t\in[0,T]$, we have $Q^{1}_{t} = Q^{1}_{T} - \int_{t}^{T}v^{1}_{t} \dif t = - \int_{t}^{T}v^{1}_{t} \dif t $. Hence, by Cauchy-Schwarz inequility and the definition of $\mathcal{A}^{1}$ in \eqref{eq:admissible_set_i_constrained}, it holds that
\begin{equation}
\begin{aligned}
\E\left[\left(Q^{1}_{t}\right)^2\middle| \Q^{1}_{0}=q^{1}_{0} \right] &= \E\left[\left(-\int_{t}^{T} v^{1}_{t} \dif t\right)^2 \middle| \Q^{1}_{0}=q^{1}_{0} \right] \leq (T-t)\E\left[\int_{t}^{T}\left(v^{1}_{t}\right)^2 \dif t \middle| \Q^{1}_{0}=q^{1}_{0} \right] \leq T K\left(q^{1}_{0}\right).
\end{aligned}    
\label{eq:bound_Q_2_minor}
\end{equation}
In addition, we combine \eqref{eq:cross_term_conditional_gerenal_1}, \eqref{eq:cross_term_conditional},   \eqref{eq:quadratic_term_conditional_not_1}, and \eqref{eq:quadratic_term_conditional_gerenal_1} to obtain 
\begin{equation}
\begin{aligned}
&\E\left[\int_0^T \left(\frac{1}{N}v^{1}_t + \frac{1}{N}\sum_{i=2}^{N} v^{i,*}_t - \mu^{*}_t \right)^2 d t \middle| \Q^{1}_{0}=q^{1}_{0} \right] \\
&= \frac{1}{N^2} \left[\E\left[\int_0^T \left(v^{1}_t - \mu^{*}_t \right)^2 d t \middle| \Q^{1}_{0}=q^{1}_{0} \right] + \sum_{i=2}^{N} \E\left[\int_0^T \left( v^{i,*}_t - \mu^{*}_t \right)^2 d t \middle| \Q^{1}_{0}=q^{1}_{0} \right] \right]\\
&\qquad + \frac{1}{N^2}\sum_{i\neq j, i,j\neq 1}\E\left[\int_0^T  \left(\mu^{*}_t - v^{i,*}_t\right)\left(\mu^{*}_t - v^{j,*}_t\right) d t \middle| \Q^{1}_{0}=q^{1}_{0} \right] \\
&\qquad + \frac{1}{N^2}\sum_{j=2}^{N}\E\left[\int_0^T  \left(\mu^{*}_t - v^{1}_t\right)\left(\mu^{*}_t - v^{j,*}_t\right) d t \middle| \Q^{1}_{0}=q^{1}_{0} \right]\\
&\leq \frac{1}{N^2} \left[ 2 \kappa + 2 K\left(q^{1}_{0}\right) + (N-1)\kappa \right] = \frac{1}{N^2} \left[ (N+1) \kappa + 2 K\left(q^{1}_{0}\right)\right].
\end{aligned}
\label{eq:convengence_rate_conditional_general_1}
\end{equation}
Hence, to bound the two terms in the last line of \eqref{eq:approximate_Nash_minor_1}, we have
\begin{equation}
\begin{aligned}
&\left|\E\left[\int_0^T Q^{1}_{t}   \left(\frac{1}{N}v^{1}_t + \frac{1}{N}\sum_{i=2}^{N} v^{i,*}_t - \mu^{*}_t \right) d t  \middle| \Q^{1}_{0}=q^{1}_{0} \right]\right|\\
&\leq \left(\E\left[\int_0^T (Q^{1}_{t})^2  d t \middle| \Q^{1}_{0}=q^{1}_{0} \right]\right)^{1/2}\left(\E\left[\int_0^T \left(\frac{1}{N}v^{1}_t + \frac{1}{N}\sum_{i=2}^{N} v^{i,*}_t - \mu^{*}_t \right)^2 d t \middle| \Q^{1}_{0}=q^{1}_{0} \right]\right)^{1/2} \\
&\leq \left[ T^{2} K\left(q^{1}_{0}\right) \right]^{1/2} \left[\frac{1}{N^2} \left[ (N+1) \kappa + 2 K\left(q^{1}_{0}\right)\right] \right]^{1/2} \\
&= \frac{T}{N} \sqrt{ K\left(q^{1}_{0}\right) \left[ (N+1) \kappa + 2 K\left(q^{1}_{0}\right)\right] },
\end{aligned}
\label{eq:approximate_Nash_minor_2}
\end{equation}
where the first inequality holds by the Cauchy-Schwarz inequality, and the second inequality holds by \eqref{eq:bound_Q_2_minor} and \eqref{eq:convengence_rate_conditional_general_1}.
In addition, we have
\begin{equation}
\begin{aligned}
&\left|\E\left[ \int_0^T Q^{1,*}_{t} \left( \frac{1}{N}\sum_{i=1}^{N} v^{i,*}_t - \mu^{*}_t\right) d t \middle| \Q^{1}_{0}=q^{1}_{0} \right] \right|\\
&\leq \left(\E\left[\int_0^T (Q^{1,*}_{t})^2  d t \middle| \Q^{1}_{0}=q^{1}_{0} \right]\right)^{1/2}\left(\E\left[\int_0^T \left( \frac{1}{N}\sum_{i=1}^{N} v^{i,*}_t - \mu^{*}_t\right)^2 d t \middle| \Q^{1}_{0}=q^{1}_{0} \right]\right)^{1/2} \\
&\leq \left([ T^{2} K\left(q^{1}_{0}\right) \right)^{1/2} \left[\frac{1}{N^2} \left[ (N+1) \kappa + 2 K\left(q^{1}_{0}\right)\right] \right]^{1/2} \\
&= \frac{T}{N} \sqrt{ K\left(q^{1}_{0}\right) \left[ (N+1) \kappa + 2 K\left(q^{1}_{0}\right)\right] },
\end{aligned}
\label{eq:approximate_Nash_minor_3}
\end{equation}
where the first inequality holds by the Cauchy-Schwarz inequality, and the second inequality holds by \eqref{eq:bound_Q_2_minor} and \eqref{eq:convengence_rate_conditional}.
Hence, we combine \eqref{eq:approximate_Nash_minor_1}, \eqref{eq:approximate_Nash_minor_2} and \eqref{eq:approximate_Nash_minor_3} to obtain
\begin{equation*}
\begin{aligned}
J^{N,1}\left(v^{1}; v^{\major,*}, v^{2,*},\dots,v^{N,*}, q^{1}_0\right) &\geq J^{N,1}\left(v^{1,*}; v^{\major,*}, v^{2,*},\dots,v^{N,*}, q^i_0\right) \\
&\quad - \frac{2T}{N} \sqrt{ K\left(q^{1}_{0}\right)  \left[ (N+1) \kappa + 2 K\left(q^{1}_{0}\right)\right] } ,   
\end{aligned}
\end{equation*}
which is the desired result in \eqref{eq:approximate_Nash_minor}.

\subsection{Proof of Theorem~\texorpdfstring{\ref{thm:major_trader_decomposition}}{5}}\label{proof:major_trader_decomposition}
First, based on \eqref{eq:periodic_decomposition_periodic} and \eqref{eq:periodic_decomposition_trend}, we have $\left(Q^{\major,*}_{t},Q^{\minor,*}_{t}\right)$ in 
\eqref{eq:periodic_decomposition} satisfies \eqref{eq:major_trader_ODE}. In addition, the ODE \eqref{eq:periodic_decomposition_trend} has a unique solution. 
Hence, it remains to show that the ODE~\eqref{eq:periodic_decomposition_periodic} has a unique periodic solution with period $T/n$. To see this, we rewrite the ODE~\eqref{eq:periodic_decomposition_periodic} as a first-order ODE in the matrix form:
\begin{equation}
\begin{aligned}
\frac{d \bm{X}_t}{d t} = \bm{A}  \bm{X}_t + \bm{B}_t,
\end{aligned}
\label{eq:major_trader_periodic_term_matrix}
\end{equation}
where
\begin{equation*}
\bm{X}_t = \begin{pmatrix} \tilde{Q}^{\major,\period}_{t} \\ Q^{\minor,\period}_{t} \\ \tilde{v}^{\major,\period}_{t} \\ v^{\minor,\period}_{t} \end{pmatrix},
\quad \bm{A} = \begin{pmatrix}
0 & 0 & 1 & 0\\
0 & 0 & 0 & 1\\
\frac{\phi_0}{a_0} & 0 & 0 & - \frac{\lambda}{2 a_0}\\
0 & \frac{\phi}{a} & - \frac{\lambda_0}{2 a} & - \frac{\lambda}{2 a}
\end{pmatrix}, 
\quad \bm{B}_t = \begin{pmatrix} 0 \\ 0 \\ - \frac{\phi_0}{a_0} \tilde{R}_t\\ 0 \end{pmatrix}.
\end{equation*}
The solution of the initial problem \eqref{eq:major_trader_periodic_term_matrix} can be expressed in a matrix form as
\begin{equation}
    \bm{X}_t = e^{\bm{A}t}\bm{X}_0 + \int_0^t e^{\bm{A}(t-s)}\bm{B}_s \dif s.
\end{equation}
So, $\bm{X}_t$ will be a periodic solution of \eqref{eq:major_trader_periodic_term_matrix} (with period $T/n$) as long as it holds that $\bm{X}_0 = \bm{X}_{T/n}$, i.e.,
\begin{equation}
\left(\bm{I}-e^{\bm{A}T/n}\right)\bm{X}^{(1)}_0 =  \int_0^{T/n} e^{\bm{A}(T/n-s)}\bm{B}_s^{(1)} \dif s.
\end{equation}
 Moreover,  $\bm{I}-e^{\bm{A}T/n}$ is invertible as $\bm{A}$ is invertible, which leads to the uniqueness of $\bm{X}_0$. Hence \eqref{eq:major_trader_periodic_term_matrix} has a unique periodic solution with period $T/n$, implying that the ODE~\eqref{eq:periodic_decomposition_periodic} has a unique periodic solution with period $T/n$.

\subsection{Proof of Proposition~\texorpdfstring{\ref{prop:major_periodicity_no_game}}{2}}\label{proof:major_periodicity_no_game}
By setting $\lambda=0$ in \eqref{eq:major_trader_ODE}, it follows that $Q^{\major}_{t}$ satisfies the following ODE:
\begin{equation*}
\left\{\begin{array}{l}
\dfrac{\dif^2 Q^{\major}_{t}}{\dif t^2} - \dfrac{\phi_0}{a_0} Q^{\major}_{t} = - \dfrac{\phi_0}{a_0} R_t, \\
\text { s.t. } Q^{\major}_{0}=q^{\major}_{0}, Q^{\major}_{T}=0.
\end{array}\right.
\end{equation*}
Denote $Q^{\major}_{t} = q^{\major}_{0} (1-t/T) + \tilde{Q}^{\major}_{t}$ and $R_{t} = q^{\major}_{0} (1-t/T) + \tilde{R}_{t}$, then $\tilde{R}_t$ is a periodic function with period $T/n$ and satisfies $R_{k\cdot T/n} = 0$ for $k=0,1,\ldots,n$ by \Cref{defn:periodic_strategy}. Moreover, $\tilde{Q}^{0}_{t}$ satisfies the following ODE on $[0,T]$:
\begin{equation}
\left\{\begin{array}{l}
\dfrac{\dif^2 \tilde{Q}^{\major}_{t}}{d t^2} - \dfrac{\phi_0}{a_0} \tilde{Q}^{\major}_{t} = - \dfrac{\phi_0}{a_0} \tilde{R}_t, \\
\text { s.t. } \tilde{Q}^{\major}_{0}=0, \tilde{Q}^{\major}_{T}=0.
\end{array}\right.
\label{eq:ODE_Q_tilde}
\end{equation}

Now, by the uniqueness and existence result of \eqref{eq:ODE_Q_tilde}, we can choose $Q^{\star}_t$ over the interval $[0,T/n]$ such that
\begin{equation*}
\left\{\begin{array}{l}
\dfrac{d^2 Q^{\star}_t}{d t^2} - \dfrac{\phi_0}{a_0} Q^{\star}_t = - \dfrac{\phi_0}{a_0} \tilde{R}_t, \\
\text { s.t. } Q^{\star}_{0}=0, Q^{\star}_{T/n}=0.
\end{array}\right.
\end{equation*}
Then, by the periodicity of $\tilde{R}_{t}$,  we have
\begin{equation*}
    \tilde{Q}^{\major}_{t} = Q^{\star}_{t-kT/n},\quad  \frac{k}{n}T \leq t < \frac{k+1}{n}T,\,k=0,1,\ldots,n-1
\end{equation*}
satisfy \eqref{eq:ODE_Q_tilde}. In addition, $\tilde{Q}^{\major}_{t}$ is a piecewise differentiable function with period $T/n$ satisfying
\[\tilde{Q}^{\major}_{k\cdot T/n } = 0, \quad k=0,1,\ldots,n.\]
Hence, it follows from Definition~\ref{defn:periodic_strategy} that the optimal strategy of the major trader, $Q^{\major}_t$, is also a periodic strategy with $n$ periods.

\subsection{Proof of Proposition~\ref{prop:periodic_trading_rate}}\label{proof:periodic_trading_rate}
To start, by straightforward calculations using \eqref{eq:periodic_component_cosine} and \eqref{eq:major_rate_lambda_0}, we have
\begin{align}
\tilde{v}^{\major,\period}_{t} &= \frac{b\phi_0\omega}{K}[\left(d_{0} d_{1}^{2} + d_{1}e_{0}e_{1} + d_{0} e_{1}^{2}\right) \cos\left(\omega t\right) + e_{0} e_{1}^{2} \sin\left(\omega t\right)], \label{eq:trading_rate_periodic_component_cosine_major}\\
v^{\minor,\period}_{t} &= \frac{b\phi_0\omega}{K}[- d_{0} e_{0} e_{1} \cos\left(\omega t\right) - e_{0}(d_{0} d_{1} + e_{0}e_{1}) \sin\left(\omega t\right)], \label{eq:trading_rate_periodic_component_cosine_minor}\\
v^{\major,*,\lambda=0}_t &= - \frac{q^{\major}_0}{T} + b \frac{\phi_0}{d_0}\omega \cos\left(\omega t\right). \label{eq:trading_rate_periodic_component_cosine_major_nogame}
\end{align}

First, comparing \eqref{eq:trading_rate_periodic_component_cosine_major} and \eqref{eq:trading_rate_periodic_component_cosine_minor} with \eqref{eq:trading_rate_amplitude}, we have 
\begin{equation*}
\begin{aligned}
&\cos(\varphi^{\major}) = \frac{b\phi_0\omega}{K\mathcal{A}\left[\tilde{v}^{\major,\period}\right]} \left(d_{0} d_{1}^{2} + d_{1}e_{0}e_{1} + d_{0} e_{1}^{2}\right)>0, \quad 
\sin(\varphi^{\major}) = \frac{b\phi_0\omega}{K\mathcal{A}\left[\tilde{v}^{\major,\period}\right]} e_{0} e_{1}^{2} > 0,\\
&\cos(\varphi^{\minor}) = 
- \frac{b\phi_0\omega}{K\mathcal{A}\left[v^{\minor,\period}\right]} d_{0} e_{0} e_{1} < 0,\quad 
\sin(\varphi^{\minor}) = - \frac{b\phi_0\omega}{K\mathcal{A}\left[v^{\minor,\period}\right]}e_{0}(d_{0} d_{1} + e_{0}e_{1}) < 0,
\end{aligned}
\end{equation*}
and
\begin{equation*}
\tan(\varphi^{\major}) = \frac{e_{0} e_{1}^{2}}{d_{0} d_{1}^{2} + d_{1}e_{0}e_{1}+d_{0} e_{1}^{2}},\quad \tan(\varphi^{\minor}) = \frac{d_{0} d_{1} + e_{0}e_{1}}{d_{0}e_{1}}.
\end{equation*}
Hence, 
\begin{equation*}
\varphi^{\major} = \arctan\left(\frac{e_{0} e_{1}^{2}}{d_{0} d_{1}^{2} + d_{1}e_{0}e_{1}+d_{0} e_{1}^{2}}\right) \in \left[0,\frac{\pi}{2}\right],\quad \varphi^{\minor} = \arctan\left(\frac{d_{0} d_{1} + e_{0}e_{1}}{d_{0}e_{1}}\right) - \pi \in \left[- \pi, - \frac{\pi}{2}\right].
\end{equation*}
Moreover, given that
\begin{equation*}
\frac{e_{0} e_{1}^{2}}{d_{0} d_{1}^{2} + d_{1}e_{0}e_{1}+d_{0} e_{1}^{2}} \leq \frac{d_{0} d_{1} e_{1} + e_{0} e_{1}^{2}}{d_{0} e_{1}^{2}} =  \frac{d_{0} d_{1} + e_{0} e_{1}}{d_{0} e_{1}},    
\end{equation*}
we have
\begin{equation*}
\varphi^{\major} - \varphi^{\minor} = \arctan\left(\frac{e_{0} e_{1}^{2}}{d_{0} d_{1}^{2} + d_{1}e_{0}e_{1}+d_{0} e_{1}^{2}}\right) - \arctan\left(\frac{d_{0} d_{1} + e_{0}e_{1}}{d_{0}e_{1}}\right) + \pi \leq \pi.
\end{equation*}
Next, it can be calculated from \eqref{eq:trading_rate_periodic_component_cosine_major} and \eqref{eq:trading_rate_periodic_component_cosine_major_nogame} that
\begin{equation*}
\begin{aligned}
\left(\mathcal{A}\left[\tilde{v}^{\major,\period}\right]\right)^2 = \left(\frac{b\phi_0\omega}{K}\right)^2[\left(d_{0} d_{1}^{2} + d_{1}e_{0}e_{1} + d_{0} e_{1}^{2}\right)^2 + e_{0}^{2} e_{1}^{4} ],\quad \left(\mathcal{A}\left[v^{\major,*,\lambda=0}\right]\right)^2 = \left(\frac{b\omega\phi_0}{d_0}\right)^2.
\end{aligned}
\end{equation*}
Therefore,
\begin{equation*}
\begin{aligned}
&\left(\frac{d_0 K}{b\omega\phi_0}\right)^2\left[ \left(\mathcal{A}\left[v^{\major,*,\lambda=0}\right]\right)^2 - \left(A^{\major}\right)^2 \right] = K^2 - d_0^2 [\left(d_{0} d_{1}^{2} + d_{1}e_{0}e_{1} + d_{0} e_{1}^{2}\right)^2 + e_{0}^{2} e_{1}^{4} ]\\
&= \left[ d_{0}^{2} d_{1}^{2} + 2 d_{0}d_{1}e_{0}e_{1} + d_{0}^{2} e_{1}^{2} + e_{0}^{2} e_{1}^{2}\right]^2 - \left(d_{0}^{2} d_{1}^{2} + d_{0} d_{1}e_{0}e_{1} + d_{0}^2 e_{1}^{2}\right)^2 - d_{0}^{2}e_{0}^{2} e_{1}^{4}\\
&= 2\left(d_{0}^{2} d_{1}^{2} + d_{0}d_{1}e_{0}e_{1}+ d_{0}^{2} e_{1}^{2}\right)\left(d_{0}d_{1}e_{0}e_{1} +  e_{0}^{2} e_{1}^{2}\right) + \left(d_{0}d_{1}e_{0}e_{1} +  e_{0}^{2} e_{1}^{2} \right)^2 - d_{0}^{2}e_{0}^{2} e_{1}^{4}\\
&\geq 2 d_{0}^{2} e_{1}^{2}\cdot e_{0}^{2} e_{1}^{2} + 0 - d_{0}^{2}e_{0}^{2} e_{1}^{4} = d_{0}^{2}e_{0}^{2} e_{1}^{4} \geq 0.
\end{aligned}
\end{equation*}
Hence, $\mathcal{A}\left[\tilde{v}^{\major,\period}\right] \leq \mathcal{A}\left[v^{\major,*,\lambda=0}\right] $.

\subsection{Proof of Proposition~\ref{prop:periodic_aggrate_price}} \label{proof:periodic_aggrate_price}
By straightforward calculations using \eqref{eq:periodic_component_cosine} and \eqref{eq:major_rate_lambda_0}, we have
\begin{equation}
\begin{aligned}
&S^{\period} := \lambda_0 \tilde{Q}^{\major,\period}_{t} + \lambda Q^{\minor,\period}_{t} \\
&=  \frac{b\phi_0}{K}[\left(\lambda_0 (d_{0} d_{1}^{2} + d_{1}e_{0}e_{1}+d_{0} e_{1}^{2})  - \lambda d_{0} e_{0} e_{1} \right) \sin\left(\omega t\right) + \left(- \lambda_0 e_{0} e_{1}^{2}  + \lambda e_{0}(d_{0} d_{1} + e_{0}e_{1})  \right) \cos\left(\omega t\right)] \\
&= \frac{2b\phi_0}{K\omega} [\left(e_0 (d_{0} d_{1}^{2} + d_{1}e_{0}e_{1}+d_{0} e_{1}^{2})  - e_1 d_{0} e_{0} e_{1} \right) \sin\left(\omega t\right) + \left(- e_0 e_{0} e_{1}^{2}  + e_1 e_{0}(d_{0} d_{1} + e_{0}e_{1})  \right) \cos\left(\omega t\right)]\\
&= \frac{2b\phi_0}{K\omega} d_{1}e_{0} [\left(d_{0} d_{1} + e_{0}e_{1}\right) \sin\left(\omega t\right) + d_{0} e_{1} \cos\left(\omega t\right)],
\end{aligned}    
\end{equation}
and
\begin{equation}
\begin{aligned}
S^{\period,\nogame} &:= \lambda_0 \tilde{Q}^{\major,\period,\lambda=0}_{t} = \lambda_0 b \frac{\phi_0}{d_0} \sin\left(\omega t\right) = \frac{2b\phi_0}{\omega} \frac{e_0}{\phi_0}.
\end{aligned}    
\end{equation}
Then
\begin{equation*}
\begin{aligned}
\left(\mathcal{A}\left[S^{\period}\right]\right)^2 
&= \left(\frac{2b\phi_0}{\omega}\right)^2 \frac{d_{1}^{2} e_{0}^{2} \left[\left(d_{0} d_{1} + e_{0}e_{1}\right)^2 + \left( d_{0} e_{1} \right)^2 \right]}{K^2},
\end{aligned}
\end{equation*}    
and
\begin{equation*}
\begin{aligned}
& \left(A\left[S^{\period,\nogame}\right]\right)^2 = \left(\frac{2b\phi_0}{\omega}\right)^2 \frac{d_{1}^{2} e_{0}^{2} \left(d_{0} d_{1} \right)^2}{\left(d_{0}^{2} d_{1}^{2}\right)^2} = 
\left(\frac{2b\phi_0}{\omega}\right)^2 \frac{ e_{0}^{2} }{d_{0}^{2} }.
\end{aligned}
\end{equation*} 
Hence,
\begin{equation*}
\begin{aligned}
& \left(A\left[S^{\period,\nogame}\right]\right)^2 - \left(A\left[S^{\period}\right]\right)^2   \\
&= \left(\frac{2b\phi_0}{\omega d_0 K}\right)^2\left[ e_0^2\left(d_{0}^{2} d_{1}^{2} + 2 d_{0}d_{1}e_{0}e_{1} + d_{0}^{2} e_{1}^{2} + e_{0}^{2} e_{1}^{2}\right)^{2} - d_{0}^{2} d_{1}^{2} e_{0}^{2} \left[\left(d_{0} d_{1} + e_{0}e_{1}\right)^2 + \left( d_{0} e_{1} \right)^2 \right]\right]\\
&= \left(\frac{2b\phi_0}{\omega d_0 K}\right)^2  e_{0}^{2} \left[\left(d_{0} d_{1} + e_{0}e_{1}\right)^2 + \left( d_{0} e_{1} \right)^2 \right] \left[ 2 d_{0}d_{1}e_{0}e_{1} + d_{0}^{2} e_{1}^{2} + e_{0}^{2} e_{1}^{2}\right]\geq 0,
\end{aligned}
\end{equation*}
which shows that $A\left[S^{\period}\right] \leq A\left[S^{\period,\nogame}\right] $.

\section{Conclusion}

To explore the causes and consequences of periodic trading activities, this paper develops a novel mean-field liquidation game involving a major trader and a continuum of minor traders. The major trader trades against a periodic and deterministic targeting strategy, while the minor traders maximize profit subject to inventory risk.

Theoretically, our study significantly extends the existing results of MFGs to a novel setting with both major and minor players, interactions through actions, and terminal state constraints. We adopt a probabilistic approach to establish the existence and uniqueness of an open-loop Nash equilibrium, and demonstrate an $\mathcal{O}(1/\sqrt{N})$ approximation rate for the mean-field solution to a major-minor game with $N$ minor players, building on approaches from extended MFGs \citep{fu_mean-field_2020} and MFGs with a major player \citep{carmona_probabilistic_2018}.

Furthermore, we contribute to the understanding of driving factors behind periodic trading activities by modeling the interactions between market participants with different trading purposes. We show that minor traders exhibit front-running behaviors in both the periodic and trend components of their strategies, reducing the major trader's profit and the intensity of her periodic trading behavior. These strategic interactions also reduce periodicity in both the aggregate trading volume and the asset price in the market, thereby serving as a stabilizing force.

\bibliography{ref}

\end{document}